\documentclass[11pt,journal, onecolumn]{IEEEtran}
\usepackage[margin=1in]{geometry}
\usepackage{amsmath}
\usepackage{amssymb} 
\usepackage{hyperref}
\usepackage{wrapfig}
\usepackage{tabularx}
\usepackage{graphicx}
\usepackage{epstopdf}
\usepackage{subcaption}
\usepackage{algorithm}
\usepackage{algorithmic}
\usepackage{xcolor,url}
\usepackage[numbers]{natbib} 
\usepackage{bbm}
\usepackage{enumitem}

\usepackage{stfloats}
\usepackage{tikz}

\usepackage{nicematrix}
\usetikzlibrary{fit}

\usepackage{amsthm}
\renewenvironment{proof}{{\sffamily\bfseries Proof. }}{\qed}
\newtheorem*{theoremnonum}{Theorem}
\newtheorem{theo}{Theorem}
\newtheorem{lem}{Lemma}

\theoremstyle{definition}
\newtheorem{de}{Definition}

\newtheorem{remark}{Remark}

\newcommand{\GW}{Galton-Watson }

\usepackage{thmtools}
\usepackage{thm-restate}
\usepackage{enumitem}
\graphicspath{ {./images/} }

\def\erdos{Erd\H os}
\def\renyi{R\'enyi }
\def\E{{\mathsf E}}
\def\P{{\mathsf P}}
\def\cg{{\mathcal{G}}}

\def\dist{{\mathrm{dist}}}
\def\ca{{\mathcal{A}}}
\def\cb{{\mathcal{B}}}
\def\cc{{\mathcal{C}}}
\def\cd{{\mathcal{D}}}
\def\ce{{\mathcal{E}}}
\def\cf{{\mathcal{F}}}
\def\ci{{\mathcal{I}}}

\def\ch{{\mathcal{H}}}
\def\cl{{\mathcal{L}}}
\def\cn{{\mathcal{N}}}
\def\cs{{\mathcal{S}}}
\def\cw{{\mathcal{W}}}
\def\ct{{\mathcal{T}}}
\def\cx{{\mathcal{X}}}

\def\cz{{\mathcal{Z}}}

\def\cm{{\mathcal{M}}}
\def\cv{{\mathcal{V}}}
\def\cp{{\mathcal{P}}}
\def\bp{{\mathbb{P}}}
\def\bq{{\mathbb{Q}}}
\def\deg{{\mathrm{deg}}}

\def\be{{\mathsf{E}}}
\def\gw{{\mathrm{GW}}}
\def\indi{{\mathbbm{1}}}
\def\poi{{\mathrm{Poi}}}
\def\sto{{\stackrel{\mathrm{sto.}}{\le}}}

\def\bin{{\mathrm{Binom}}}

\def\dtv{{d_\mathrm{TV}}}
\def\rmtb{{\mathrm{TB}}}
\def\rmrb{{\mathrm{RB}}}

\def\sfch{{\mathsf{Ch}}}

\def\dist{\mathrm{dist}}

\def\excl{{\mathsf{EXC}}}
\def\com{{\mathsf{COM}}}

\def\tdct{\tilde{\mathcal{T}}}

\newcommand{\td}[1]{\tilde{#1}}
\newcommand{\wdtd}[1]{\widetilde{#1}}
\newcommand{\cond}{\mathchoice{\,\vert\,}{\mspace{2mu}\vert\mspace{2mu}}{\vert}{\vert}}

\title{\fontsize{18}{24}\selectfont Achieving Almost Exact Recovery in Almost Quadratic Time: Rank-Based Graph Matching via Local Tree Correlation Tests}

\author{Jiale Cheng$^*$, Ziao Wang$^*$ and Lei Ying
\thanks{$*$ Jiale Cheng and Ziao Wang contributed equally to this work.}
\thanks{Jiale Cheng is with the Department of Electrical and Computer Engineering, University of Michigan, Ann Arbor, MI 48109, USA (email:jlcheng@umich.edu).}
\thanks{Ziao Wang is with the Department of Electrical and Computer Engineering, University of Michigan, Ann Arbor, MI 48109, USA (email: ziaow@umich.edu).}
\thanks{Lei Ying is with the Department of Electrical and Computer Engineering, University of Michigan, Ann Arbor, MI 48109, USA (email: leiying@umich.edu).}
}

\begin{document}
\maketitle
\begin{abstract}
    This paper studies graph matching under the correlated \erdos--\renyi (ER) graph pair model. This model first samples an $\mathrm{ER}(n,\frac{\lambda}{ns})$ base graph, whose edges are then independently subsampled twice with probability $s$ to produce two correlated $\mathrm{ER}(n,\frac{\lambda}{n})$ graphs. We propose a graph matching algorithm that has $n^{2+o(1)}$ time complexity and achieves almost exact recovery with high probability under the assumptions $\lambda=(\log n)^{\alpha+o(1)}$ for some $\alpha\in(0,1)$ and $s\in(\sqrt{C_{\mathrm{Otter}}},1]$, where $C_{\mathrm{Otter}}\approx 0.338$ is Otter's tree-counting constant. This is the first algorithm with almost quadratic time complexity in this regime of $\lambda$, while the best known result in this regime is the chandelier-counting algorithm with time complexity $O(n^{c(s)})$, where $c(s)\rightarrow \infty$ as $s$ approaches $\sqrt{C_\mathrm{Otter}}$ from above. The proposed algorithm is based on local tree correlation tests. It uses a rank-based algorithm to match the vertex pairs instead of threshold-based rules in the literature.  
    This avoids the need of computing an explicit threshold, which is computationally difficult to obtain. To prove the almost exact recovery result, we establish a new analysis of tree correlation tests in the diverging-degree regime, where both the mean degree and the tree depth grow with $n$. Based on this new result, we establish the {\em existence} of a threshold for a threshold-based graph matching algorithm via local tree correlation tests. Finally, we couple the performance of the rank-based algorithm with the threshold-based algorithm to show almost exact recovery. 
\end{abstract}

\section{Introduction}
Graph matching, also known as network alignment or noisy graph isomorphism, is the problem of recovering the latent correspondence between the vertices of two correlated graphs. A classical motivation for the problem comes from social network deanonymization~\cite{Nar-Shm-deanonymizing2009,Nitish14Reconciliation}, where one attempts to identify the user correspondence between an anonymized social network (e.g. Twitter) and a reference social network (e.g. LinkedIn) for which the user identities are publicly available. Similar correspondence-recovery applications appear in various other fields including bioinformatics~\cite{singh2007pairwise}, computer vision~\cite{Minsu12Pattern}, and natural language processing~\cite{Hag-Ng-robust2005}.
In general, the standard optimization formulation of graph matching can be viewed as a special case of the quadratic assignment problem (QAP)~\cite{Burkard1998}, which is known to be NP-hard to solve or even approximate~\cite{QAPNPApproxi}. This worst-case computational barrier has motivated a line of work studying graph matching  assuming some random graph generative models such as the correlated \erdos--\renyi graph pair model~\cite{Pedram11GA}. Under the correlated \erdos--\renyi graph pair model, denoted by $\mathrm{CER}(n,\lambda,s),$  a base \erdos--\renyi graph with $n$ vertices  is first generated where edges between vertex pairs are generated independently with probability $\frac{\lambda}{ns}.$ Two \erdos--\renyi  graphs, both with marginal distribution $\mathrm{ER}(n,\lambda/n)$, are then subsequently generated by independently subsampling each edge in the base graph with probability $s,$ the edge correlation parameter. The goal is to find the latent vertex correspondence between the two graphs without observing their true identities. 
\begin{figure}
    \centering
    \includegraphics[width=0.6\linewidth]{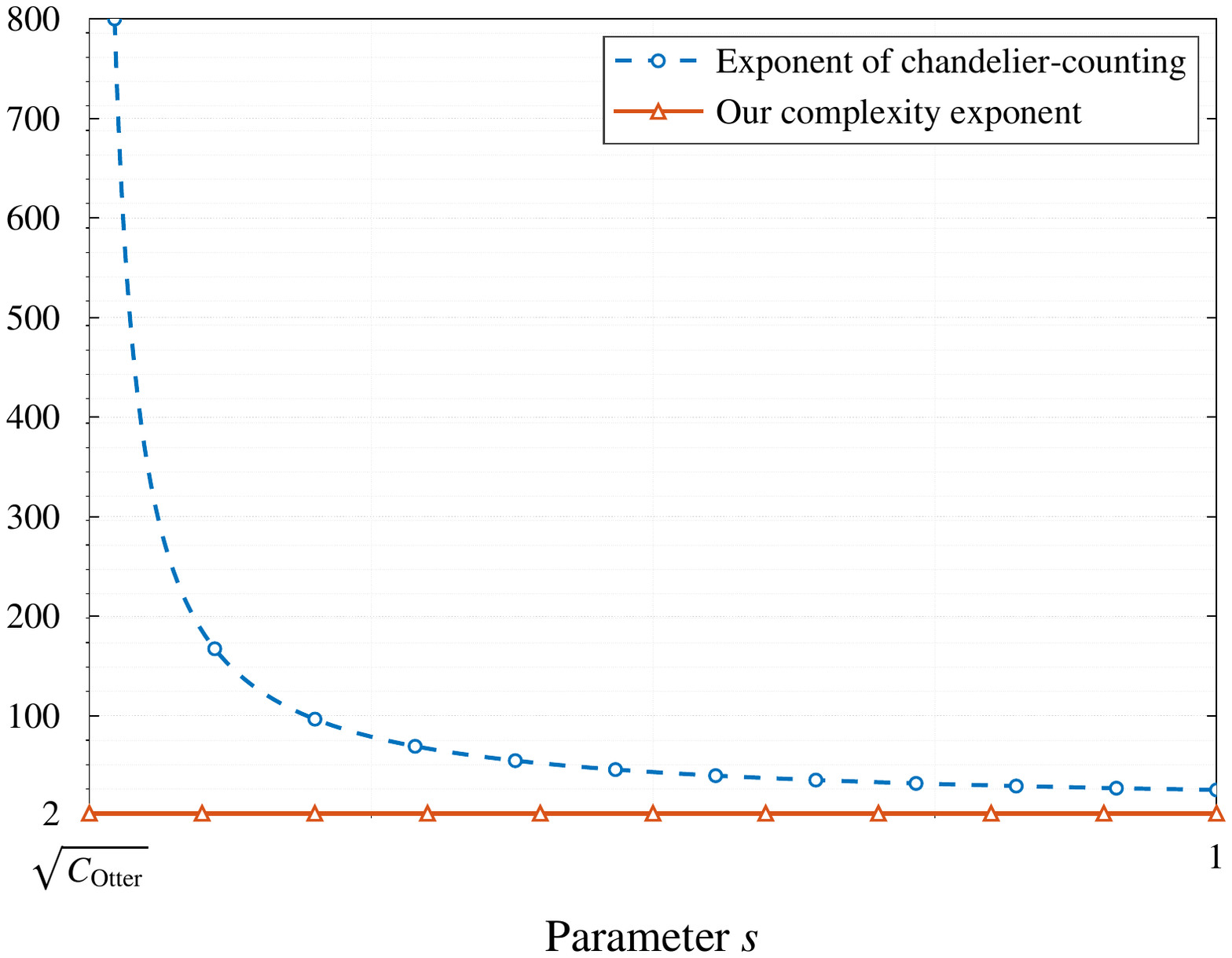}
    \caption{Comparison of the time-complexity exponents between our algorithm and the chandelier-counting algorithm in~\cite{Cheng26Chandelier}.}
    \label{fig:exponent}
\end{figure}

This paper studies graph matching of $\mathrm{CER}(n,\lambda,s)$ with $\lambda =(\log n)^{\alpha+o(1)}$ for $\alpha\in(0,1)$. This belongs to the {\em almost exact recovery regime,} where it is information-theoretically possible to recover the correspondence for all but a vanishing fraction of the vertices, while exact recovery of the full latent correspondence is infeasible due to the existence of isolated vertices~\cite{Yihong22Threshold}. The main contribution of this paper is an algorithm with $n^{2+o(1)}$ time complexity. The algorithm is  built upon \emph{tree correlation testing}~\cite{Luca24LRT}, and matches a vertex pair by examining the local trees rooted at their children sequentially according to the correlation rank. The proposed algorithm achieves almost exact recovery with high probability when $s\in (\sqrt{C_\mathrm{Otter}},1],$ where $C_\mathrm{Otter}\approx 0.338$ is the Otter's tree-counting constant~\cite{Richard48TreeCount}.

To the best of our knowledge, this is the {\em first} algorithm with $n^{2+o(1)},$ {\em almost quadratic}, time complexity in this almost exact recovery regime.  The only existing polynomial-time algorithm that achieves almost exact recovery in this regime 
is the chandelier-counting algorithm proposed in~\cite{Cheng26Chandelier}. 
However, the time complexity of their algorithm is $O(n^{c(s)})$, where the exponent $c(s)\rightarrow \infty$ as $s\rightarrow \sqrt{C_\mathrm{Otter}}$ as shown in Figure \ref{fig:exponent}, and satisfies $c(s)>9$ for all $ s\in (\sqrt{C_\mathrm{Otter}},1]$ based on the analysis in \cite{Cheng26Chandelier}\footnote{In the work~\cite{Cheng26Chandelier}, the correlated \erdos--\renyi graph pair model is parameterized by the edge correlation coefficient $\rho$ between the two graphs instead of the subsampling probability $s$, and their assumption is stated as $\rho\in (\sqrt{C_\mathrm{Otter}},1]$. In the sparse regime we consider, one has $s=\rho(1-\lambda/n)+\lambda/n=\rho(1+o(1))$. Thus their assumption is asymptotically equivalent to $s\in(\sqrt{C_{\mathrm{Otter}}},1]$. For this reason, we also refer to $s$ as the correlation parameter of the model.}.

\section{Related Work}
\subsection[Information-theoretic limits and polynomial-time algorithms for CER(n,lambda, s)]{Information-theoretic limits and polynomial-time algorithms for $\mathrm{CER}(n,\lambda,s)$}
Prior work on the correlated \erdos--\renyi graph pair model can be broadly divided into two lines: one studying the information-theoretic limits and the other developing computationally efficient algorithms. The first line of work has considered several recovery criteria. Besides almost exact recovery, as defined in the previous section, one may also consider exact recovery, which requires correctly recovering the entire latent permutation, and partial recovery, which refers to correctly matching a non-vanishing fraction of vertices. The information-theoretic thresholds have been identified in a series of works studying exact recovery~\cite{Daniel16GA,Daniel18GA,Yihong22Threshold}, almost exact recovery~\cite{Daniel19PartialAlign,Yihong22Threshold} and partial recovery~\cite{Luca21Partial,Hall23Partial,Yihong22Threshold,Jian23Partial,du2025optimal}.

The second line of work develops various polynomial-time matching algorithms with theoretical guarantees~\cite{Osman19CanonicalLabel,Jian21DegProfile,Fan23Spectral1,Fan23Spectral2,Luca24LRT,Luca20Tree,maier2025,Cheng23ConstantCorrelation,mao21a-pmlr,Cheng26Chandelier,ding2026polynomial}. 
In the sparse regime, where $\lambda=n^{o(1)}$, the feasibility of efficient algorithms is closely related to the Otter's constant $C_\mathrm{Otter}$. Under the assumption $s>\sqrt{C_\mathrm{Otter}}$, efficient algorithms are known to achieve different recovery criteria in their different degree regimes:
\begin{itemize}
    \item \emph{\underline{Partial recovery regime $\lambda=\Theta(1)$:}} In this regime, the chandelier-counting algorithm in~\cite{Cheng26Chandelier} and the tree-correlation testing algorithm in~\cite{Luca24LRT,Luca24Stat} achieve partial recovery.
    \item 
    \emph{\underline{Almost exact recovery regime $\lambda=\omega(1)$, $\lambda s\le (1-\Omega(1))\log n$:}} Prior to this work, the only efficient algorithm known to achieve almost exact recovery in this regime is the chandelier-counting algorithm in~\cite{Cheng26Chandelier}. The setting studied in this paper lies in this regime.
    \item
    \emph{\underline{Exact recovery regime $\lambda s\ge (1+\Omega(1))\log n$:}} In this regime, the chandelier-counting algorithm in~\cite{Cheng26Chandelier} is known to achieve exact recovery. An algorithm with $n^{2+o(1)}$ complexity to achieve exact recovery is proposed in~\cite{Cheng23ConstantCorrelation}, yet it 
    instead requires a more stringent condition $s>s_0$, where $s_0$ is an unspecified constant that is understood to be close to one.
\end{itemize}

On the converse side, when $s<\sqrt{C_\mathrm{Otter}}$, evidence for the non-existence of such efficient algorithms has been provided through low-degree hardness~\cite{li2025algorithmic}.

\subsection{Tree correlation testing} Our proposed algorithm is based on \emph{tree correlation testing}, which was introduced and developed in the line of work~\cite{Luca20Tree,Luca24LRT,Luca24Stat,maier2025}. This method reduces the task of matching a pair of vertices in the two correlated graphs $G$ and $\td G$ to the hypothesis testing of whether their local tree-like neighborhoods are independent or correlated Poisson Galton--Watson trees. We provide a formal definition of this testing problem in Section~\ref{sec:tree-test}. 

Previous works~\cite{Luca24LRT,Luca24Stat,maier2025} studied the likelihood-ratio test problem in the constant-degree regime, where the mean degree of the trees remains fixed while the depth tends to infinity. 
They then applied tree correlation tests in the graph matching setting with constant $\lambda$ to achieve partial recovery.
However, before the present work, it was unknown whether almost exact recovery could be achieved with this method in the diverging-degree regime. A key technical contribution is to analyze the tree correlation testing problem in a regime where both the mean degree and tree depth diverge and grow as functions of $n$. This analysis is the key ingredient that enables our algorithm to achieve almost exact recovery, and may also be of independent interest for tree correlation test in the diverging-degree regime.

Our graph matching algorithm also differs from the previous ones in how the tree correlation tests are used. Previous works use a \emph{threshold-based} rule: a candidate pair is accepted if its likelihood ratio exceeds a predefined threshold. In contrast, our algorithm is \emph{rank-based}: it ranks the likelihood ratios of all candidate tree pairs and constructs the matching by sequentially examining the local tree pairs according to the ranking. This avoids the use of an explicit threshold, which is computationally difficult to obtain with diverging degrees.

\subsection{Graph matching under other random graph models}
Beyond the correlated \erdos--\renyi graph pair model, graph matching has also been studied in more general models, including correlated stochastic block models~\cite{racz2021correlated,Shuwen24SBMChandellier,Joonhyuk23SBM}, correlated geometric graph models~\cite{Haoyu22Geo}, and attributed graph models with publicly available side information~\cite{Ning24AGA,Ziao24AGA,Ziao25AGA2,huang2026attributed}. Recent studies have also extended to the settings of simultaneously matching multiple graphs~\cite{ameen2025aligning,ameen2024aligning} and matching partially observed graphs~\cite{wang26d,shiu2025information}. These extensions incorporate richer structural information, side information, or observation models, and therefore lead to statistical and algorithmic phenomena that are different from those in the homogeneous \erdos--\renyi setting. 

While this paper focuses on the canonical correlated \erdos--\renyi graph pair model, we expect that the tree correlation test analysis developed here may also be useful for designing low-complexity algorithms in some of these broader graph matching models.

\section{Model and Main Results}
\subsection{Notation}
Given a simple graph $H$, we use $\cv(H)$ to denote its vertex set and $\ce(H)$ to denote its edge set. The distance between two vertices $u,v\in \mathcal{V}(H)$ is the length of the shortest path between them. A tree is a connected, acyclic simple graph, and a rooted tree is a tree together with a distinguished vertex, designated as the root. For a vertex in a rooted tree, its depth is its distance from the root. The depth of a rooted tree is the maximum depth among all its vertices. For a correlated \erdos--\renyi graph pair $(G,\td G)\sim \mathrm{CER}(n,\lambda,s)$, we use short-hand notations $\cv:=\cv(G)$, $\ce:=\ce(G)$, $\td\cv:=\cv(\td G)$ and $\td\ce:=\ce(\td G)$. Recall that $\cv=\td\cv=[n]$. To distinguish vertices in the two graphs, we use variables without tildes, such as $u,v$, for vertices in $G$, and variables with tildes, such as $\td u,\td v$, for vertices in $\td G$. For $u\in\cv$, let $\cn(u)$ denote the set of neighbors of $u$ in $G$. Similarly, for $\td u\in\td\cv$, $\td\cn(\td u)$ denotes the set of neighbors of $\td u$ in $\td G$.

\begin{figure}
    \centering 
    \includegraphics[width=0.6\linewidth]{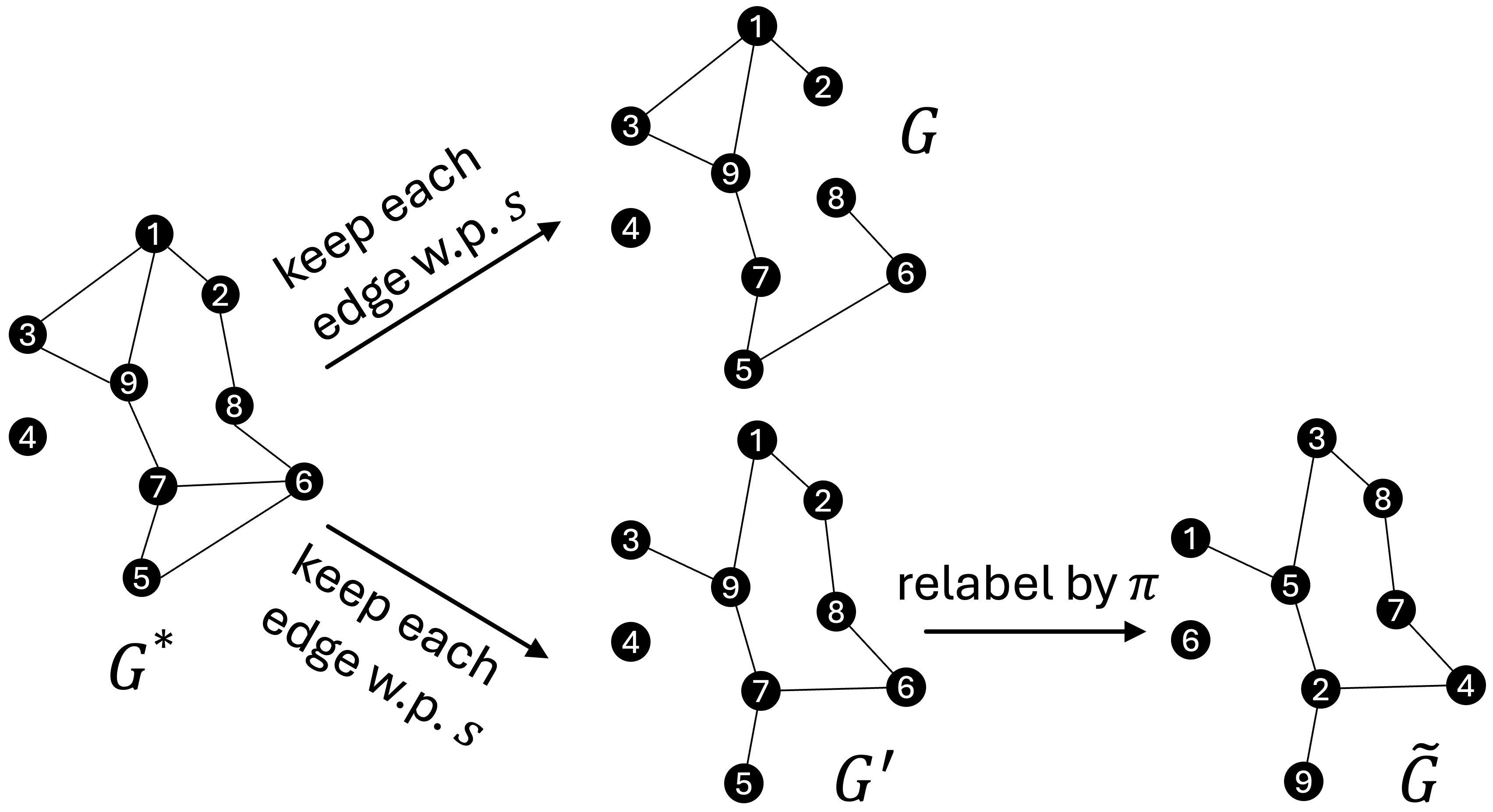}
        \caption{An example of $\mathrm{CER}(n,\lambda,s)$ model with $n=9$.} 
        \label{fig:CER_model}
\end{figure}

We use $\P(\cdot)$ to denote the probability of an event and use $\E[\cdot]$ to denote the expectation with respect to the underlying probability distribution. We sometimes explicitly include the probability distribution as the subscript of the expectation to make it clear. We follow the standard asymptotic notation: $f(n) = O(g(n))$ if $\limsup_{n \to \infty} \frac{|f(n)|}{g(n)} < \infty$; $f(n) = \Omega(g(n))$ if $\liminf_{n \to \infty}\frac{|f(n)|}{g(n)} >0$; $f(n) = \Theta(g(n))$ if $f(n) = O(g(n))$ and $f(n) = \Omega(g(n))$; $f(n) = o(g(n))$ if $\lim_{n \to \infty}\frac{f(n)}{g(n)} = 0$; $f(n) = \omega(g(n))$ if $\lim_{n \to \infty}\frac{|f(n)|}{|g(n)|} = \infty$; and $f(n) \sim g(n)$ if $\lim_{n \to \infty}\frac{f(n)}{g(n)} = 1$.  In addition, we use $\td O(\cdot)$ notation when logarithmic factors are omitted.

\subsection{Correlated \erdos--\renyi graph pair model}
Given a positive integer $n$, real numbers $s\in (0,1]$ and $\lambda \in (0,ns]$, the correlated \erdos--\renyi graph pair model $\mathrm{CER}(n,\lambda,s)$ generates a pair of graphs $(G,\td G)$ through the following steps, as illustrated in Figure \ref{fig:CER_model}: 
\begin{itemize}
    \item[(i)] We first sample a base \erdos--\renyi graph $G^*\sim\mathrm{ER}(n,\frac{\lambda}{ns})$, where $G^*$ has vertex set $[n]:=\{1,\ldots,n\}$, and edges are independently generated between vertex pairs with probability $\frac{\lambda}{ns}$; \item[(ii)] Given $G^*$, graph $G$ is generated by keeping each edge in $G^*$ independently with probability $s$, and the same subsampling process is independently repeated to generate graph $G'$; 
    \item[(iii)] A permutation $\pi:[n]\rightarrow[n]$ is sampled uniformly at random, and used to relabel $G'$ to obtain graph $\td G$. In other words, vertices $u$ and $v$ are connected in $G'$ if and only if $\pi(u)$ and $\pi(v)$ are connected in $\td G$;
    \item[(iv)] Finally, the model outputs the two graphs $G$ and $\td G$ both with marginal distribution $\mathrm{ER}(n,\frac{\lambda}{n})$. We denote $(G,\td G)\sim\mathrm{CER}(n,\lambda,s)$. 
\end{itemize}
After observing $(G,\td G)$, the goal of graph matching is to find a matching $\hat \pi: \cv(G)\rightarrow \cv(\td G),$ which matches every vertex in $G$ with a distinct vertex in $\td G.$ Since $G$ retains the true labels and $\hat \pi$ is also a  permutation $[n]\rightarrow[n]$,  $\hat \pi$ is an estimator of the latent permutation $\pi.$ We say that the matching $\hat \pi$  achieves \emph{almost exact recovery} if
\[
|\{u\in [n]:\hat\pi(u)= \pi(u)\}|=n-o(n),
\]
i.e., the estimate recovers the true labels of all but a vanishing fraction
of the vertices of $\td G$.  This paper studies efficient algorithms that achieve almost exact recovery. 
The following informal theorem summarizes our main results.

\begin{theoremnonum}[informal statements of Theorem 1 and 2]
Assume $s\in (\sqrt{C_\mathrm{Otter}},1]$ and $\alpha\in (0,1)$ are two constants independent of $n.$ 
Given a pair of graphs $(G,\td G)\sim \mathrm{CER}(n,\lambda,s)$ with $\lambda=(\log n)^{\alpha+o(1)}$, the rank-based alignment algorithm, named \texttt{RBAlign} in Section~\ref{sec:rbalign},  achieves almost exact recovery with high probability and has $n^{2+o(1)}$ time complexity.
\end{theoremnonum}
 
\texttt{RBAlign} is a matching algorithm based on local tree correlation tests, including three key steps:
\begin{itemize}
   \item \emph{Extract local trees:} For each vertex $u,$ \texttt{RBAlign} considers the local trees rooted at $u$'s children, say vertex $w,$ and with depth $d$, which, in most cases, is the depth-$d$ neighborhood of $w$ after removing the vertex $u$ from $G.$

    \item \emph{Compute local tree-pair likelihood ratios:} For each pair of local trees, one from $G$ and one from $\td G$, the algorithm computes the likelihood ratio of the correlated Galton--Watson tree-pair distribution and the independent Galton--Watson tree-pair distribution, which measures the correlation strength of the tree pair. 

    \item \emph{Rank and match:} \texttt{RBAlign} ranks all local tree pairs in a decreasing order of their likelihood ratios. It then processes this ranked list progressively: at each step, the next highest-ranked local tree pair is classified as correlated. Two vertices $u\in G$ and $\td v\in \td G$ are declared to be matched once  three local tree pairs rooted at their children are classified as correlated.
    This matching criteria is known as the \emph{3-dangling-tree} test in the literature~\cite{Luca24LRT}.
    The procedure terminates once $n-o(n)$ vertex pairs have been matched.
\end{itemize}

In the next section, we present a detailed description of \texttt{RBAlign} and its analysis. We will begin by introducing tree correlation tests, which serve as the main building block of our algorithm. We then present a threshold-based alignment algorithm \texttt{TBAlign}. \texttt{TBAlign} relies on a threshold for the local tree correlation tests,  but the threshold is computationally infeasible to obtain. However, \texttt{TBAlign} provides an important stepping stone for the design and analysis of the rank-based algorithm. In fact, to prove the existence of the threshold in \texttt{TBAlign}, we have to develop a different type of local tree correlation tests, which again is computationally infeasible, and then use it to prove the existence of the threshold using the Neyman-Pearson lemma.  Finally, we will formally present \texttt{RBAlign}, analyze its time complexity, and establish its recovery guarantee by coupling it with \texttt{TBAlign}. The relationship between these main results of the paper are summarized in Figure \ref{fig:overview}.
\begin{figure}[htbp]
    \centering
    \includegraphics[width=1\linewidth]{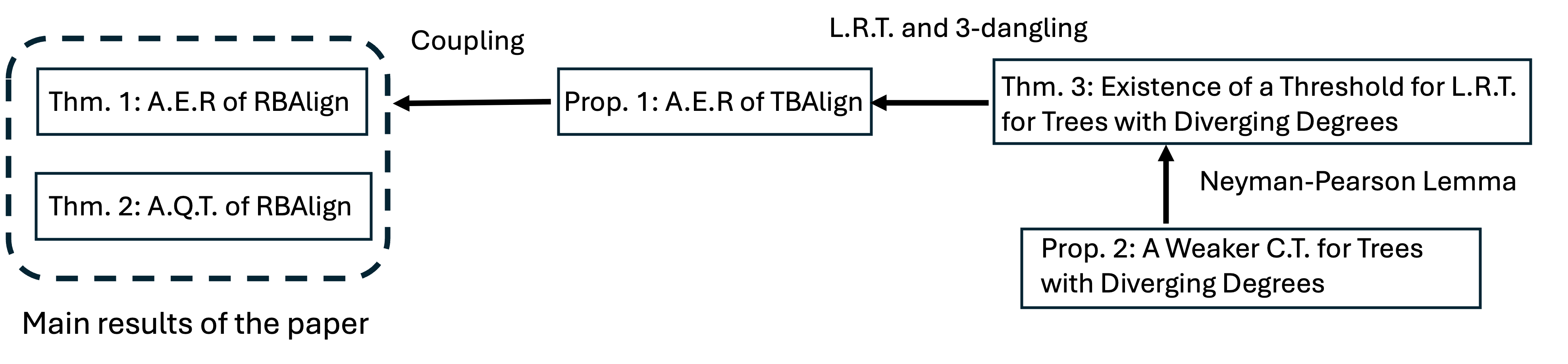}
    \caption{Relationships among the main theoretical results in the paper. A.E.R. - Almost Exact Recovery; A.Q.T. - Almost Quadratic Time; L.R.T. - Likelihood Ratio Test; C.T. - Correlation Test}
    \label{fig:overview}
\end{figure}

\section{Algorithms and Analysis}
\subsection{ Tree Correlation Tests}
\label{sec:tree-test}
\textbf{\underline{Poisson Galton--Watson trees}.}
Let $\lambda>0$ be a real number. We define the Poisson Galton--Watson tree with offspring mean $\lambda$ as the random rooted tree $\ch$ generated by the following branching process.

The process starts from a single root, whose label is $(1)$. For each vertex $v$ that appears in the tree, independently sample $C_v\sim \poi(\lambda)$. The vertex $v$ generates $C_v$ children sequentially, called its first child, second child, and so on up to its $C_v$-th child. This sequential generation provides an ordering of the children of $v$. If $v$ has label $(a_0,\ldots,a_k)$, then its $i$-th child is assigned the label $(a_0,\ldots,a_k,i)$. The process continues recursively for all newly generated vertices. It stops if there are no vertices at some depth; otherwise, it continues indefinitely. We denote the support of the random rooted tree $\ch$ by $\cx$, and the set of elements in $\cx$ of depth less than or equal to $k$ by $\cx_k$.

For an integer $k\ge 0$, let $\ch_k$ denote the truncation of $\ch$ up to depth $k$. More precisely, $\ch_k$ is the subgraph of $\ch$ induced by the set of vertices with depth less than or equal to $k$, with the same root and the inherited order labels. Then for $h\in\cx_k$, we have
\[
\gw(\ch_k=h)=\prod_{i=0}^{k-1}
\prod_{v\in\cv(h,i)}
\P(\poi(\lambda)=c_v),
\]
where $\cv(h,i)$ denotes the set of vertices at depth $i$ in $h$, and $c_v$ is the number of children of vertex $v$ in $h$. By convention, the inner product over $v\in\cv(h,i)$ equals to one when $\cv(h,i)=\emptyset$.

\textbf{\underline{Unlabeled rooted trees and isomorphism classes}.}
For a labeled rooted tree $\ch$, we write $[\ch]$ for the unlabeled rooted tree obtained from $\ch$ by removing all vertex labels while preserving the identity of the root. In other words, $[\ch]$ records only the rooted tree structure of $\ch$.

This notation naturally induces an equivalence relation on labeled rooted trees. We say that two labeled rooted trees $\ch$ and $\ch'$ are isomorphic,
and write
$\ch\cong \ch',$
if their corresponding unlabeled rooted trees are identical, i.e.,
$[\ch]=[\ch'].$
Equivalently, $\ch\cong\ch'$ if there exists a bijection between their vertex sets that preserves the identity of the root and all the edges.

We define $\cz:=\{[\ch]:\ch\in\cx\}$, and $\cz_k:=\{[\ch]:\ch\in\cx_k\}$ for each integer $k\ge 0$. Thus, $\cz$ (resp. $\cz_k$) is the collection of all unlabeled rooted trees that can be obtained from removing the labels from rooted trees in $\cx$ (resp. $\cx_k$).

\textbf{\underline{Independent and correlated tree distributions}.}
Let $\lambda>0$ and $s\in(0,1]$. We now define two probability distributions on $\cx\times\cx$: the independent tree-pair distribution $\bq(\ch,\td\ch)$ and the correlated tree-pair distribution $\bp(\ch,\td\ch)$.

\begin{itemize}
    \item  Independent tree-pair distribution $\bq(\ch,\td\ch)$: Under this distribution, the two trees $\ch$ and $\td{\ch}$ are two independently sampled Poisson Galton--Watson trees with offspring mean $\lambda$.

    \item Correlated tree-pair distribution $\bp(\ch,\td\ch):$
    To generate $\ch$ and $\td\ch$, we first sample an \emph{intersection tree} $\ch^*$, which is a Poisson Galton--Watson tree with offspring mean $\lambda s$. Starting from $\ch^*$, we generate the first tree $\ch$ as follows. For each vertex $v$ of $\ch^*$, independently sample $X_v^+ \sim \poi(\lambda(1-s)).$ Let $c_v^*$ denote the number of children of $v$ in $\ch^*$. We attach $X_v^+$ new children to $v$, ordered after the existing children of $v$ in $\ch^*$. That is, these newly added children are assigned orders
    \[
    c_v^*+1,\ c_v^*+2,\ \ldots,\ c_v^*+X_v^+,
    \]
    and their labels are defined accordingly by concatenating the label of $v$ with their respective orders. For each newly added child $u$ of $v$, which lies at depth $k+1$, we independently sample $\ch_u$, which is a Poisson Galton--Watson tree with offspring mean $\lambda$. We then attach $\ch_u$ at $u$ by identifying the root of $\ch_u$ with $u$. Equivalently, if the label of $u$ is $a$, then the root label $(1)$ of $\ch_u$ is replaced by $a$, and every descendant label $(1,i_1,\ldots,i_m)$ in $\ch_u$ is replaced by $(a,i_1,\ldots,i_m).$ After performing this augmentation for every vertex of $\ch^*$, we obtain the rooted tree $\ch$. The second tree $\td\ch$ is generated by repeating the same augmentation procedure independently, starting from the same intersection tree $\ch^*$.

\end{itemize}

For an integer $k\ge 0$, let $\ch_k$ and $\td\ch_k$ denote the truncations of $\ch$ and $\td\ch$ up to depth $k$, respectively. 
For two rooted trees $h,\td h\in \cx_k$, $\bq(\ch_k=h,\td\ch_k=\td h)$ is the probability, under the independent tree-pair distribution $\bq$, that the depth-$k$ truncations of $\ch$ and $\td\ch$ are equal to $h$ and $\td h$, respectively. Similarly, $\bp(\ch_k=h,\td\ch_k=\td h)$ is the probability of the same event under the correlated tree-pair distribution $\bp$.
For two unlabeled rooted trees $t,\td t\in\cz_k$, $\bq([\ch_k]=t,[\td\ch_k]=\td t)$ is the probability, under the independent tree-pair distribution $\bq$, that the unlabeled versions of the depth-$k$ truncations of $\ch$ and $\td\ch$ are equal to $t$ and $\td t$, respectively. Similarly, $\bp([\ch_k]=t,[\td\ch_k]=\td t)$ is the probability of the same event under the correlated tree-pair distribution $\bp$. The \emph{likelihood ratio} $L_k$ is defined as
\begin{equation}
\label{eq:LHR}
L_k(t,\td t):=\frac{\bp([\ch_k]=t,[\td\ch_k]=\td t)}{\bq([\ch_k]=t,[\td\ch_k]=\td t)}.
\end{equation}

\subsection{Threshold-based Alignment Algorithm: \texttt{TBAlign}}

\begin{de}[Local tree $\ct_{u\setminus v,d}^{\tau}$]
\label{de:local_tree}
    Let $\tau\ge 1$, and $d$ be a positive integer.
    For two adjacent vertices $u,v\in \cv$ such that $(u,v)\in \ce$, we define $\ct^{\tau}_{u\setminus v,d}$ to be the local tree rooted at vertex $u$ in the graph $G\setminus v$ such that $\ct^{\tau}_{u\setminus v,d}$ is the $d$-hop-neighborhood of $u$ in $G\setminus v$ if it is a tree and every node in this neighborhood has degree at most $\tau$; otherwise, $\ct^{\tau}_{u\setminus v,d}$ is an empty tree without any vertices.
\end{de}
We use $\cc_{d}^{\tau}$ to denote the collection of all of such local trees in $G$.  Notice that for each edge $(u,v)\in \ce$, we have two local trees $\ct_{u\setminus v,d}^{\tau}$ and $\ct_{v\setminus u,d}^{\tau}$ according to the definition above. Therefore, we have $|\cc_{d}^{\tau}|=2|\ce|$. 
$\tilde{\ct}^{\tau}_{\td u\setminus\td v,d}$ and $\tilde{\cc}^{\tau}_{d}$ are similarly defined on $\tilde G.$

Recall the definition in~\eqref{eq:LHR}. For a pair of local trees $\ct^{\tau}_{w\setminus u,d}$ and ${\td \ct}^{\tau}_{\td x\setminus\td v,d}$, $L_d([\ct^{\tau}_{w\setminus u,d}],[{ \td\ct}^{\tau}_{\td x\setminus\td v,d}])$ is the likelihood ratio evaluated at such a tree pair. If at least one of these two local trees is an empty tree, we define $L_d([\ct^{\tau}_{w\setminus u,d}],[{\td\ct}^{\tau}_{\td x\setminus\td v,d}]):=0$. In the rest of this section, we omit the subscript $d$ and superscript $\tau$ from notations $\ct^{\tau}_{w\setminus u,d},\td\ct^{\tau}_{\td x\setminus\td v,d},\cc^{\tau}_d$ and $\td\cc^{\tau}_d$
because $\tau$ and $d$ are fixed parameters.

The algorithm consists of the following key steps

\begin{itemize}
    \item {\bf Step 1:} The algorithm computes a $2|\ce|\times 2|\td\ce|$  real-valued matrix ${\bf L}$ that records all the likelihood ratios of local tree pairs. In particular, each row of ${\bf L}$ represents a local tree $\ct_{w\setminus u}\in \cc$, and each column represents a local tree $\td\ct_{\td x\setminus\td v}\in \td\cc$. The entry ${\bf L}(\ct_{w\setminus u},\td\ct_{\td x\setminus\td v})$ is the likelihood ratio $L_d([\ct_{w\setminus u}],[\td\ct_{\td x\setminus\td v}])$.

\item {\bf Step 2:} Initialize an $n\times n$ matrix ${\bf M}_\mathrm{TB}.$ For each pair of vertices $(u,\td v)$, define a {\em witness} graph $\cw(u, \td v)$, which is a weighted bipartite graph between  vertex sets $\cn(u)$ and $\td\cn(\td v)$ with the weight of edge $(w, \td x)$ being 
 ${\bf L}(\ct_{w\setminus u},\td\ct_{\td x\setminus\td v}).$ 
 
 \item {\bf Step 3:} Given threshold $\theta^*,$ we run a greedy matching algorithm to determine ${\bf M}_\mathrm{TB}:$
 \[
{\bf M}_\mathrm{TB}(u, \td v)=\texttt{GreedyMatching}(\cw(u,\td v), \theta^*, 3)\in\{0,1\}.
\] 
\texttt{GreedyMatching} selects edges greedily according to their weights and stops once it finds a matching of size $3$ (output 1) or once all the remaining edges either have weights smaller than $\theta^*$, or have weights equal to $\theta^*$ but are rejected  by tie-breaking rule $f$ (output 0). See Figure~\ref{fig:greedy} for an example. Here a matching is a set of edges in which no two edges share a common vertex, and the tie-breaking rule $f:\cz_d\times\cz_d\rightarrow [0,1]$ is an injective mapping that induces a canonical ordering of all possible tree pairs.
Our results hold for any tie-breaking rules, an explicit example of the tie-breaking rule can be found in Appendix~\ref{appd:treecode}. 

\end{itemize} 
\begin{figure}
    \centering 
        \includegraphics[width=1\textwidth]{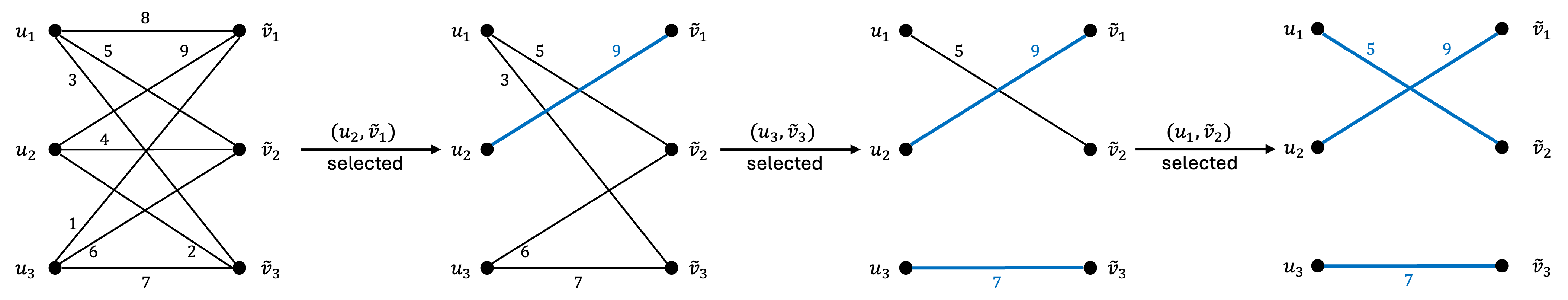}
        \caption{An example of \texttt{GreedyMatching} with $\theta^*=4$. The number next to each edge represents its weight. In this example, we obtain a matching of size $3$, and hence the subroutine outputs $1$. If $\theta^*$ is instead chosen to be $6$, then the edge $(u_1,\td v_{2})$ cannot be included, and hence the subroutine outputs $0$.} 
        \label{fig:greedy}
\end{figure}

\begin{algorithm}[H]
\caption{\texttt{TBAlign}: Threshold-based Alignment Algorithm }
\begin{algorithmic}[1]\label{algo:TBAlign}
\REQUIRE Two graphs $G$ and $\td G$, average degree $\lambda$, correlation $s$, threshold $\theta^*$, parameters $d$ and $\tau$.
\ENSURE Matrix ${\bf M}_\mathrm{TB}\in \{0,1\}^{n\times n}$.
\STATE Initialize all-zero matrices ${\bf M}_\mathrm{TB} \gets{\bf 0}_{n\times n}$ and ${\bf L}\gets{\bf 0}_{2|\ce|\times2|\td \ce|}$.
\FOR{each $(\ct_{w\setminus u},\td \ct_{\td x\setminus \td v})\in
\cc\times \td \cc$}
\STATE Compute the likelihood ratio ${\bf L}(\ct_{w\setminus u},\td\ct_{\td x\setminus \td v})$.
\ENDFOR

\FOR{each $(u,\td v)\in \cv\times \td {\cv}$}
\STATE Construct the witness graph $\mathcal{W}(u,\td v)$.
            \STATE ${\bf M}_\mathrm{TB}(u, \td v)\gets\texttt{GreedyMatching}(\cw(u,\td v), \theta^*, 3)\in\{0,1\}.$
    \ENDFOR
\RETURN ${\bf M}_\mathrm{TB}$
\end{algorithmic}
\end{algorithm}

The following proposition provides the theoretical guarantee for \texttt{TBAlign}. In stating this proposition and throughout the rest of this paper, we assume without loss of generality that the latent permutation $\pi$ is the identity permutation, which implies that $u$ and $\td u$ are two corresponding vertices in $G$ and $\td G$. Of course, none of our algorithms knows or uses this information. The proof of this proposition is deferred to Section \ref{sec:prop1}. 
\begin{restatable}{prop}{TBAlign}
	\label{th:TBAlign}
Consider $(G,\td G)\sim \mathrm{CER}(n,\lambda,s)$ with $\lambda=(\log n)^{\alpha+o(1)}$ for some $\alpha\in (0,1)$ and $s\in (\sqrt{C_\mathrm{Otter}},1].$  
Choose  
$d=(\log n)^{\gamma}$ for some $\gamma\in (1-\alpha,1)$, and $\tau=(\log n)^{t}$ for some $t\in(\max\{\alpha,\gamma\},1)$ when constructing the local trees\footnote{To lighten the notation, we do not explicitly round $d$ to an integer as one can always choose $\gamma\in (1-\alpha,1)$ so that $(\log n)^\gamma$ is an integer, given that $n$ is suffciently large.}. Assume $\alpha,$ $s,$ $\gamma,$ and $t$ are all constants independent of $n.$
Then there exists a threshold $\theta^*,$ depending on $n$, $\lambda$, $d$ and $s,$ such that the output matrix ${\bf M}_\mathrm{TB}$ of \texttt{TBAlign} satisfies
\begin{align}
   \sum_{u\in \cv} \mathbbm{1}\{{\bf M}_\mathrm{TB}(u,\td u) =1 \text{ and } \|{{\bf M}_\mathrm{TB}(u,:)}\|_1=1\}\ge n-3n\exp\left(-\frac12(\log n)^{\alpha/2}\right)\label{eq:TBAlign}
\end{align}
with probability $1-O(\exp(-\frac12(\log n)^{\alpha/2}))$.
\end{restatable}

\begin{remark}\label{rem:TB_almost_exact}
We note that Proposition~\ref{th:TBAlign} implies that \texttt{TBAlign} achieves almost exact recovery with high probability. To see this, we construct a permutation $\hat{\pi}:[n]\to[n]$ from $\bf M_\mathrm{TB}$ by the following procedure. First, for each vertex $u\in\cv$, set the corresponding row ${\bf M}_\mathrm{TB}(u,:)$ to all zeros if $\|{\bf M}_\mathrm{TB}(u,:)\|_1\geq 2$. Second, for each vertex $\td v\in\td\cv$, set the corresponding column ${\bf M}_\mathrm{TB}(:,\td v)$ to all zeros if $\|{\bf M}_\mathrm{TB}(:,\td v)\|_1\geq 2$. After these two steps, the resulting matrix has at most one entry equal to $1$ in each row and each column. Finally,  for each $u\in\cv$ with $\|{\bf M}_\mathrm{TB}(u,:)\|_1=1$, we set $\hat{\pi}(u)$ to be the unique $\td v\in \td\cv$ such that ${\bf M}_\mathrm{TB}(u,\td v)=1$. 
    This procedure gives a partial permutation on $[n]$. We then arbitrarily match the rest of the vertices to obtain a complete permutation $\hat \pi$.

    By~\eqref{eq:TBAlign}, at least $n-3n\exp\left(-\frac12(\log n)^{\alpha/2}\right)$ vertices $u$ has a unique correct entry ${\bf M}_\mathrm{TB}(u,\td u)=1$, each of them will be correctly matched by $\hat \pi$ unless the second step erases column $\td u$. Since there are at most $3n\exp\left(-\frac12(\log n)^{\alpha/2}\right)$ conflicting vertices that could erase correct entry, $\hat{\pi}$ will agree with $\pi$ on at least $n-6n\exp\left(-\frac12(\log n)^{\alpha/2}\right)$ vertices. Therefore, almost exact recovery is achieved.
\end{remark}

\subsection{Rank-based Alignment Algorithm: \texttt{RBAlign}}
\label{sec:rbalign}
\texttt{TBAlign} is not directly implementable in practice because $\theta^*$ is hard to compute. Inspired by the threshold-based algorithm, we introduce a rank-based alignment algorithm, namely \texttt{RBAlign}. The idea is to consider the pairs of local trees in a decreasing order  according to their likelihood ratios until most of the vertices are matched. It is equivalent to gradually reducing threshold $\theta$ in the \texttt{TBAlign} until it matches a target number of matched pairs. Given the existence of threshold $\theta^*$ for correctly aligning at least $n-3n\exp\left(-\frac12(\log n)^{\alpha/2}\right)$ vertex pairs, \texttt{RBAlign} can correctly align at least $n-n/\log n$ vertex pairs with a high probability when $\theta$ approaches $\theta^*.$ The algorithm however does not need to know $\theta^*$ and only sets the target number of matched pairs to be smaller than that under  \texttt{TBAlign}.  
More specifically, \texttt{RBAlign} ranks all local tree pairs based on the likelihood ratios, breaking ties according to the rule $f$ as in \texttt{TBAlign}, such that the pair with the  largest likelihood ratio has rank 1.  Given rank $k\in [4|\ce|\,|\td\ce|],$ we define $r^{-1}(k)=(\ct_{w\setminus u},\td \ct_{\td x\setminus \td v}),$ which is the pair of local trees ranked as $k,$  and $\theta_k={\bf L}(r^{-1}(k)),$ which is the corresponding likelihood ratio.

 \texttt{RBAlign} includes the following key steps:  
\begin{itemize}
    \item {\bf Step 1 and Step 2:} The same as in \texttt{TBAlign}, except that ${\bf M}_\mathrm{TB}$ being replaced by ${\bf M}_\mathrm{RB}$. 
    \item {\bf Step 3:}    Loop over rank $k$ in the increasing order, assuming $r^{-1}(k)=(\ct_{w\setminus u},\td\ct_{\td x\setminus \td v}),$ to do the following: 
    \begin{itemize}
        \item Skip if ${\bf M}_\mathrm{RB}(u, \td v)=1;$ and otherwise
        \item Set 
 \[
{\bf M}_\mathrm{RB}(u, \td v)=\texttt{GreedyMatching}(\cw(u,\td v), \theta_k, 3).
\]  
Here, \texttt{GreedyMatching} only considers edges in $\cw(u,\td v)$ with rank $\leq k.$
        \item Terminate the algorithm when the number of rows with at least one ``1'' in ${\bf M}_\mathrm{RB}$ reaches $n-\lceil n/\log n\rceil.$
    \end{itemize}
    
\end{itemize}

\begin{remark} 
    \texttt{RBAlign}  and \texttt{TBAlign} use the same greedy matching but different thresholds. Therefore, when \texttt{RBAlign} terminates at rank $k_e,$ and $\theta_{k_e}={\bf L}(r^{-1}(k_e))\geq \theta^*,$ i.e., the likelihood ratio is at least $\theta^*,$ $u$ and $\td v$ are matched only if they are also matched under \texttt{TBAlign}. Therefore, \texttt{RBAlign} preserves the matched pairs under \texttt{TBAlign} but stops early after finding enough number of matched pairs.  
\end{remark}

\begin{algorithm}[htbp]
\caption{\texttt{RBAlign}: Rank-based alignment algorithm}
\begin{algorithmic}[1]\label{algo:RBAlign}
\REQUIRE Two graphs $G$ and $\td G$, average degree $\lambda$, correlation $s$, parameters $d$ and $\tau$.
\ENSURE Matrix ${\bf M}_\mathrm{RB}\in \{0,1\}^{n\times n}$.
\STATE Initialize all-zero matrices ${\bf M}_\mathrm{RB} \gets{\bf 0}_{n\times n}$ and ${\bf L}\gets{\bf 0}_{2|\ce|\times2|\td \ce|}$.
\STATE Initialize $\Delta\gets 0$ that counts the number of rows in ${\bf M}_{\mathrm{RB}}$ that has at least one ``$1$''.
\FOR{each $(\ct_{w\setminus u},\td \ct_{\td x\setminus \td v})\in
\cc\times \td \cc$}
\STATE Compute the likelihood ratio ${\bf L}(\ct_{w\setminus u},\td\ct_{\td x\setminus \td v})$.
\ENDFOR
\FOR{each $(u,\td v)\in \cv\times \td {\cv}$}
\STATE Construct the witness graph $\mathcal{W}(u,\td v)$.
    \ENDFOR
\STATE Rank all the entries in $L_d$ in a non-increasing order, breaking ties according to tie-breaking rule $f$, to obtain the bijection $r:\cc\times \td \cc\rightarrow[4|\ce|\,|\td \ce|]$.
\FOR{$k=1:4|\ce|\,|\td\ce|$}
\STATE Find vertices $u,w,\td v,\td x$ such that $r^{-1}(k)=(\ct_{w\setminus u},\td\ct_{\td x\setminus \td v})$.
\STATE $\theta_k\gets {\bf L}(\ct_{w\setminus u},\td\ct_{\td x\setminus \td v})$.
\IF{${\bf M}_\mathrm{RB}(u,\td v)=0$}
\STATE ${\bf M}_\mathrm{RB}(u,\td v)\gets \texttt{GreedyMatching}(\mathcal{W}(u,\td v),\theta_k,3)$.
\IF{${\bf M}_\mathrm{RB}(u,\td v)=1$ and $\|{\bf M}_\mathrm{RB}(u,:)\|_1=1$}
\STATE $\Delta\gets \Delta+1$.
\ENDIF
\ENDIF
\IF{$\Delta = n-\lceil n/\log n\rceil$}
\RETURN ${\bf M}_\mathrm{RB}$
\ENDIF
\ENDFOR
\RETURN ${\bf M}_\mathrm{RB}$
\end{algorithmic}
\end{algorithm}

The following two theorems are the main results of this paper, which provide the performance guarantee and the run-time guarantee of \texttt{RBAlign}.
\begin{restatable}{theo}{RBAlign}
	\label{th:RBAlign}
Under the same assumptions as in Proposition \ref{th:TBAlign}, the output matrix $M_\mathrm{RB}$ of \texttt{RBAlign} satisfies
\begin{align}
    \sum_{u\in V} \mathbbm{1}\{{\bf M}_\mathrm{RB}(u,\td u) =1 \text{ and } \|{\bf M}_\rmrb(u,:)\|_1=1 \}= n-O(n/\log n)\label{eq:RBAlign}
\end{align}
with probability $1-O(\exp(-\frac12(\log n)^{\alpha/2}))$. In particular, \eqref{eq:RBAlign} implies that \texttt{RBAlign} achieves almost exact recovery by an analogous argument as in Remark~\ref{rem:TB_almost_exact}.
 \end{restatable}

\begin{theo}\label{prop:Runtime}
Under the same assumptions as in Proposition \ref{th:TBAlign}, 
\texttt{RBAlign} has a time complexity of $|\ce|\,|\td\ce| \,n^{o(1)}$ deterministically and has a time complexity of $n^{2+o(1)}$ with probability $1-o(e^{-n})$.
\end{theo}

In the following, we prove Theorem~\ref{th:RBAlign} based on Proposition~\ref{th:TBAlign} in Section~\ref{sec:proof_rb}, and prove Theorem~\ref{prop:Runtime} in Section~\ref{sec:proof_rb_runtime}.

\subsection{Proof of Theorem~\ref{th:RBAlign}}
\label{sec:proof_rb}
To prove Theorem~\ref{th:RBAlign}, we first couple the outputs of \texttt{RBAlign} and \texttt{TBAlign} in the following lemma.

\begin{lem}
\label{lem:TBRB}
    Given the $\theta^*$ defined in Proposition \ref{th:TBAlign} and any realization $(g,\td g)$ of $(G,\td G)\sim \mathrm{CER}(n,\lambda,s)$ such that result \eqref{eq:TBAlign} holds under \texttt{TBAlign}, we have 
\begin{itemize}
    \item[(i)] ${\bf M}_\mathrm{RB}(u, \td v)=1$ only if 
${\bf M}_\mathrm{TB}(u, \td v)=1,$ and
    \item[(ii)] $\Delta = n-\lceil n/\log n\rceil$ when \texttt{RBAlign} terminates. 
\end{itemize} 
\end{lem}

\begin{proof} We first note that $\texttt{GreedyMatching}(\cw(u,\td v), \theta, 3)\in\{0, 1\}$ is a non-increasing function in $\theta$ because the set of edges that can be included in a matching shrinks when the threshold $\theta$ increases.

    Let $k_e$ be the rank index when \texttt{RBAlign} terminates. 
    Let $\theta_e:={\bf L}(r^{-1}(k_e))$.
    Suppose that $\theta_e<\theta^*$ or $\theta_e=\theta^*$ but its corresponding tree pair will be rejected by the breaking rule $f$ under  \texttt{TBAlign}. Then for any $(u,\td v)$ such that    
\[
{\bf M}_\mathrm{TB}(u, \td v)=\texttt{GreedyMatching}(\cw(u,\td v), \theta^*, 3)=1,
\]  we also have
\[
{\bf M}_\mathrm{RB}(u, \td v)=\texttt{GreedyMatching}(\cw(u,\td v), \theta_e, 3)=1
\]  because  $\texttt{GreedyMatching}(\cw(u,\td v), \theta, 3)$ is a non-increasing function in $\theta.$
Therefore, when \texttt{RBAlign} terminates, the number of rows with at least one ``1''s in ${\bf M}_\mathrm{RB}$ is lower bounded by that of ${\bf M}_\mathrm{TB}.$ This implies that when  \eqref{eq:TBAlign} holds under \texttt{TBAlign},   $$\Delta\geq n-3n\exp\left(-\frac12(\log n)^{\alpha/2}\right)>  n-\lceil n/\log n\rceil,\text{ for all large enough }n,$$ which contradicts the termination condition of \texttt{RBAlign}. Therefore, we can conclude that statement (i) holds. 

We further note that looping over all local tree pairs under  \texttt{RBAlign} is equivalent to  \texttt{TBAlign} with threshold $\theta=0$.  With result  \eqref{eq:TBAlign}, \texttt{TBAlign} with $\theta=0$ will output an ${\bf M}_\mathrm{TB}$ with at least $n-3n\exp\left(-\frac12(\log n)^{\alpha/2}\right)$ rows with at least one ``1''s because $\theta=0\leq \theta^*.$ Therefore,  with result  \eqref{eq:TBAlign}, \texttt{RBAlign} will terminate before looping all local tree pairs, which means it terminates because $\Delta$ reaches the target value, i.e., result (ii) holds.  
\end{proof}

With this lemma, we now move on to prove Theorem~\ref{th:RBAlign}.
Consider a realization $(g,\td g)$ of $(G,\td G)\sim \mathrm{CER}(n,\lambda,s)$ such that \eqref{eq:TBAlign} holds under \texttt{TBAlign}. 
We first note that 
\begin{align}
    &\sum_{u\in \cv}\mathbbm{1}\{{\bf M}_\mathrm{RB}(u,\td u) =1 \text{ and } \|{\bf M}_\rmrb(u,:)\|_1=1\}\nonumber\\
    ={}& \sum_{u\in \cv} \mathbbm{1}\{\|{\bf M}_\rmrb(u,:)\|_1\ge 1\}-\sum_{u\in \cv} \mathbbm{1}\{\exists\td v\neq \td u:{\bf M}_\rmrb(u,\td v)=1\}.\label{eq:1}
\end{align}

According to conclusion (ii) in Lemma \ref{lem:TBRB}, ${\bf M}_\rmrb$ must have exactly $n-\lceil n/\log n\rceil$ rows that are not all zero, i.e., 
\begin{equation}
        \label{eq:num_row}
        \sum_{u\in \cv} \mathbbm{1}\{\|{\bf M}_\rmrb(u,:)\|_1\ge 1\}= n-\lceil n/\log n\rceil.
    \end{equation}
Notice that event $\{{\bf M}_\rmtb(u,\td u) =1 \text{ and } \|{{\bf M}_\rmtb(u,:)}\|_1=1\}$ implies event $\{{\bf M}_\rmtb(u,\td v)=0,\forall\td v\neq \td u\}$. By property~\eqref{eq:TBAlign}, we have
\begin{equation*}
\sum_{u\in \cv}\mathbbm{1}\{\exists\td v\neq\td u:{\bf M}_\rmtb(u,\td v)=1\}\le 3n\exp\left(-\frac12(\log n)^{\alpha/2}\right),
\end{equation*}
and it follows by conclusion (i) in Lemma~\ref{lem:TBRB} that
\begin{equation}
\label{eq:num_bad_row}
\sum_{u\in \cv}\mathbbm{1}\{\exists\td v\neq\td u:{\bf M}_\rmrb(u,\td v)=1\}\le 3n\exp\left(-\frac12(\log n)^{\alpha/2}\right).
\end{equation}
Finally, substituting equations~\eqref{eq:num_row} and~\eqref{eq:num_bad_row} into equation~\eqref{eq:1} yields
\begin{align*}
\sum_{u\in \cv}\mathbbm{1}\{{\bf M}_\mathrm{RB}(u,\td u) =1 \text{ and } \|{\bf M}_\rmrb(u,:)\|_1=1\}
    &\ge n-\left\lceil n/\log n\right\rceil - 3n\exp\left(-\frac12(\log n)^{\alpha/2}\right)\\
    &=n-O\left(n/\log n\right),
\end{align*}
and we conclude the theorem by the fact that~\eqref{eq:TBAlign} holds with probability $1-O(\exp(-\frac12(\log n)^{\alpha/2}))$ according to Proposition \ref{th:TBAlign}.

\subsection{Proof of Theorem~\ref{prop:Runtime}}
\label{sec:proof_rb_runtime}
To prove this theorem, we first show that the total time required to compute all entries of the likelihood-ratio matrix $\mathbf L$ is 
    \[
    |\mathcal E|\,|\td{\mathcal E}|\,n^{\frac{(2t+o(1))\log\log n}{(\log n)^{1-t}}}.
    \]
    We then show that the time complexity of remaining steps of Algorithm~\ref{algo:RBAlign} is of lower order. 

    \underline{\textbf{Complexity of computing $\mathbf L$:}}
    We first focus on bounding the time complexity of computing a single likelihood
    ratio ${\bf L}(\ct_{w\setminus u},\td\ct_{\td x\setminus \td v})$.  
    Note that the vertex degrees of $\ct_{w\setminus u}$ and $\td\ct_{\td x\setminus\td v}$  are upper bounded by $\tau$ due to our construction. 
    
To compute
    ${\bf L}(\ct_{w\setminus u},\td\ct_{\td x\setminus \td v})$, we use the recursive
    formula \eqref{eq:rec} from Lemma~2.1 of~\cite{Luca24LRT}, with the following notations: 
    \begin{itemize}
        \item  Let $\ct$ and $\td\ct$ be two rooted
    trees with depth $\leq d$ and degree $\leq \tau$.
    \item For $l\in\{0,\ldots,d\}$, let
    $\mathcal V_l$ and $\td{\mathcal V}_l$ denote the sets of vertices at depth $l$
    in $\ct$ and $\td\ct$, respectively. 

    \item Given $y\in\mathcal V_l$ and
    $\td z\in\td{\mathcal V}_l$, let $\ct_y$ and $\td\ct_{\td z}$ denote the
    subtrees rooted at $y$ and $\tilde z,$ respectively.
    
    \item Given $y\in\mathcal V_l$ and
    $\td z\in\td{\mathcal V}_l$, let $\cc_y$ and $\td\cc_{\td z}$ denote the
    sets of children of $y$ and $\tilde z,$ respectively.
    \item        For finite sets $\ca$ and $\cb$, let $\mathrm{Inj}({\ca},{\cb})$ denote the set of injective mappings from $\ca$ to $\cb$.
    \end{itemize}
 Then the likelihood ratio  satisfies: 
    \begin{equation}
    \label{eq:rec}
    L(\ct_y,\td\ct_{\td z})
    =
    \sum_{k=0}^{\min(|\mathcal C_y|,|\td{\mathcal C}_{\td z}|)}
    \psi(k,|\mathcal C_y|,|\td{\mathcal C}_{\td z}|)
    \sum_{\substack{\sigma\in \mathrm{Inj}([k], \mathcal C_y)\\
    \td\sigma\in \mathrm{Inj}([k],\td{\mathcal C}_{\td z})}}
    \prod_{j=1}^k
    L (\ct_{\sigma(j)},\td\ct_{\td\sigma(j)}),
    \end{equation}
    where
     $\psi(k,|\mathcal C_y|,|\td{\mathcal C}_{\td z}|)
    :=
    \frac{e^{\lambda s}s^k(1-s)^{|\mathcal C_y|+|\td{\mathcal C}_{\td z}|-2k}}
    {\lambda^k k!},$
    and the inner sum is defined to be one for $k=0$. This recursion starts from the leaves and is initialized by
    $L(\ct_y,\td\ct_{\td z})=1$ for any $y\in {\cv}_d$ and $\td z\in \td{\cv}_d.$

    Now given $ L(\ct_a,\td\ct_{\td b})$ for all $a\in{\cv}_{l+1}$ and $\tilde{b}\in\td{\cv}_{l+1},$ we bound the time complexity of a single recursion using~\eqref{eq:rec} to compute $ L(\ct_y,\td\ct_{\td z})$ for $y\in{\cv}_l$ and $\td z\in \td{\cv}_l$, which is at most
    \begin{align}
    \label{eq:tc-ub}
    \sum_{k=0}^{\min(|\mathcal C_y|,|\td{\mathcal C}_{\td z}|)}\underbrace{O(\lambda s+k+|\cc_y|+|\td\cc_{\td z}|)}_{\#\text{product operations for computing } \psi}+\underbrace{|\mathrm{Inj}([k],\cc_y)|\,|\mathrm{Inj}([k],\td\cc_{\td z})|}_{\#(\sigma,\td\sigma)\text{ pairs}} \times\underbrace{k}_{\#\text{terms in the product}}.
    \end{align}
    Using
    \[
    |\mathrm{Inj}([k],{\mathcal C}_y)|\le |{\mathcal C}_y|!\le \tau!,
    \qquad
    |\mathrm{Inj}([k],\td{\mathcal C}_{\td z})|\le |\td\cc_{\td z}|!\le \tau!,
    \qquad
    k\le \tau\quad \text{and} \quad \lambda s<\tau
    \]
    we obtain that the computation complexity of a single iteration of ~\eqref{eq:rec} can be upper bounded by  $$\eqref{eq:tc-ub}=O\bigl(\tau^2(\tau!)^2\bigr).$$
    Therefore, the total time complexity of computing the likelihood of all pairs of vertices at depth $l$ is $O\bigl(\tau^{2l+2}(\tau!)^2\bigr)$ because    
    $|\mathcal V_l|\,|\td{\mathcal V}_l|\le \tau^{2l}.$
    Summing over $l=0,\ldots,d-1$ gives $O\bigl(d\tau^{2d+2}(\tau!)^2\bigr)$
    as an upper bound on the time required to compute
    ${\bf L}(\ct_{w\setminus u},\td\ct_{\td x\setminus\td v})$.
    The matrix $\bf L$ has in total $4|\ce|\,|\td\ce|$ entries, and hence, the complexity of computing $\bf L$ is
    $O(|\ce|\,|\td\ce|\,d \tau^{2d+2}(\tau!)^2).$ Since $\tau=(\log n)^t$ and $d=(\log n)^{\gamma}$, we have
    \begin{align*}
    d\,\tau^{2d+2}(\tau!)^2
    &\le
    \exp\Big(
    (2d+2)\log\tau
    +2\tau\log\tau
    +\log d
    \Big)\\
    &=\exp\Big((2t(\log n)^{\gamma}+2t)\log\log n+2t(\log n)^t\log\log n+\gamma\log\log n\Big)\\
    &=
    \exp\Big(
    (2t+o(1))(\log n)^{\gamma}\log\log n
    +
    2t(\log n)^t\log\log n
    \Big).
    \end{align*}
    By the assumption that $t\in (\max\{\alpha,\gamma\},1)$,  the second term in the exponent dominates, and hence
    \[
    d\,\tau^{2d+2}(\tau!)^2
    =
    \exp\Big((2t+o(1))(\log n)^t\log\log n\Big)
    =
    n^{\frac{(2t+o(1))\log\log n}{(\log n)^{1-t}}},
    \]
    so the time complexity of computing $\bf L$ is $|\mathcal E|\,|\td{\mathcal E}|\,n^{\frac{(2t+o(1))\log\log n}{(\log n)^{1-t}}}.$

    \underline{\textbf{Complexity of constructing the local trees:}}
    We next bound the time complexity of constructing all local trees in $\cc$ and $\td\cc$. We first consider the construction of a single local tree $\ct_{w\setminus u}$. The construction proceeds in two steps. First, we perform a breadth-first search starting from $w$ in the graph $G\setminus u$, up to depth $d$. Recall that the local tree is declared to be empty if any vertex explored during this procedure has more than $\tau$ children. Hence, whenever the resulting local tree is nonempty, the number of vertices explored by the breadth-first search is at most $\sum_{l=0}^d \tau^l = O(\tau^d).$
    Thus, the time complexity of this search is $O(\tau^d)$. 
    
    After the breadth-first search, we verify that the explored neighborhood is indeed a tree. Let $\cv_l$ denote the set of vertices discovered at depth $l$. For each depth $l\in[d]$, we check whether there are any edges among vertices in $\cv_l$. This requires at most $O(|\cv_l|^2)$ time complexity. We also check whether every vertex in $\cv_l$ has exactly one
    parent in $\cv_{l-1}$, which requires at most $O(|\cv_{l-1}|\,|\cv_l|)$ time complexity. Since $|\cv_l|\le \tau^l$, the total cost of these checks is
    bounded by
    \[
    \sum_{l=1}^dO\bigl(|\cv_l|^2+|\cv_{l-1}|\,|\cv_l|\bigr) \le \sum_{l=1}^d O\bigl(\tau^{2l}\bigr) =
    O(\tau^{2d}).
\]
    Therefore, the time complexity of constructing one local tree $\ct_{w\setminus u}$ is $O(\tau^{2d})$.

    Since $|\cc|=2|\ce|$, constructing all local trees in $\cc$ requires time $O(2|\ce|\tau^{2d})$. Similarly, constructing all local trees in $\td\cc$ requires time $O(2|\td\ce|\tau^{2d}).$ Therefore, the cost of constructing all these local trees is given by $$O(2|\ce|\tau^{2d}+2|\td\ce|\tau^{2d})=o(|\ce|\,|\td\ce|\,d\,\tau^{2d+2}(\tau!)^2)=o\left(|\mathcal E|\,|\td{\mathcal E}|\,n^{\frac{(2t+o(1))\log\log n}{(\log n)^{1-t}}}\right).$$

    \underline{\textbf{Complexity of ranking all the entries in $\bf L$:}} Notice the for a list of $K$ real numbers, their total ranking can be found in $O(K\log K)$ time. Therefore, the time complexity of ranking all the entries of $\bf L$ is given by 
    \[
    O(|\ce|\,|\td\ce|\,\log(|\ce|\,|\td\ce|))=O(|\ce|\,|\td\ce|\,\log n)=o\left(|\mathcal E|\,|\td{\mathcal E}|\,n^{\frac{(2t+o(1))\log\log n}{(\log n)^{1-t}}}\right),
    \]
    where the last equality follows because $\log n=n^{\frac{\log\log n}{\log n}}$ and $1-t\in(0,1)$.

    \underline{\textbf{Complexity of looping over ranks $k=1,\ldots,4|\ce|\,|\td\ce|$}:} Fix a rank $k$. Let $r^{-1}(k)=(\ct_{w\setminus u},\ct_{\td x\setminus \td v})$. Recall that $\theta_k={\bf L}(\ct_{w\setminus u},\ct_{\td x\setminus \td v})$. If $\mathbf{M}(u,\td v)=1$, we perform no operations and continue to the next rank. If $\mathbf{M}(u,\td v)=0$, then for each edge $(y,\td z)$ with weight $\theta_k$ in the witness graph $\mathcal{W}(u,\td v)$, we check whether the edge $(y,\td z)$ shares an endpoint with any edge in the matching found previously. If no endpoint is shared, $(y,\td z)$ is added to the matching because it has the highest edge ``weight'' among the remaining edges. Since $\mathbf{M}(u,\td v)=0$ implies the size of previous matching is at most $2$, we only need to check whether $(y,\td z)$ has any common endpoint with at most $2$ other edges. This
    requires $O(1)$ time per edge.

    Notice that the total number of edges in all the witness graphs is given by $4|\ce|\,|\td\ce|$, and each of these edges is processed at once as $k$ ranges from
    $1$ to $4|\ce|\,|\td\ce|$. Therefore, the total complexity of this for loop is 
    \[
    O(|\ce|\,|\td\ce|)=o\left(|\mathcal E|\,|\td{\mathcal E}|\,n^{\frac{(2t+o(1))\log\log n}{(\log n)^{1-t}}}\right).
    \]

    \underline{\textbf{Total complexity of Algorithm~\ref{algo:RBAlign}:}} Summing over the complexity of these four steps yields that the time complexity of Algorithm~\ref{algo:RBAlign} is 
    \begin{align}
        |\mathcal E|\,|\td{\mathcal E}|\,n^{\frac{(2t+o(1))\log\log n}{(\log n)^{1-t}}}=|\mathcal E|\,|\td{\mathcal E}|\,n^{o(1)},\label{eq:det_complexity}
    \end{align}
    because we have $\frac{(2t+o(1))\log\log n}{(\log n)^{1-t}}=o(1)$ under the assumption that $t\in (\max\{\alpha,\gamma\},1)$.
    
    We can further conclude that the time complexity Algorithm~\ref{algo:RBAlign} is $n^{2+o(1)}$ with probability $1-o(e^{-n})$. To see this, by the Chernoff bound, we have $\P(|\ce|\ge n\lambda) = o(e^{-n})$ and  $\P(|\td\ce|\ge n\lambda) = o(e^{-n})$.
    Therefore, with probability at least $1-o(e^{-n})$, both $|\ce|$ and $|\td\ce|$ are bounded by $n\lambda$. On this event, \eqref{eq:det_complexity} implies that the running time of Algorithm~\ref{algo:RBAlign} is $n^2\lambda^2n^{o(1)}=n^{2+o(1)}$.

    \begin{remark}
   We note that the high-probability time complexity guarantee can be made deterministic by adding an initial termination condition: if $|\ce|\ge n\lambda$ or $|\td\ce|\ge n\lambda$, the algorithm terminates immediately. By the Chernoff bound, this modification increases the failure probability by $o(e^{-n})$.    
\end{remark}

\section{Proof of Proposition~\ref{th:TBAlign}}
\label{sec:prop1}

Proposition \ref{th:TBAlign} shows that \texttt{TBAlign} achieves almost exact recovery.  \texttt{TBAlign} is a variation of the alignment algorithm based on tree correlation tests in \cite{Luca24LRT}, which achieves partial recovery when $\lambda=
\Theta(1).$ The result in \cite{Luca24LRT} does not apply to diverging degrees $\lambda=(\log n)^{\alpha+o(1)}$ for two reasons: 
\begin{itemize}
    \item In~\cite{Luca24LRT}, tree correlation tests are analyzed in the constant-degree regime, where the mean offspring number, $\lambda,$ is a fixed constant, while the testing depth tends to infinity.
    This is different from the diverging-degree regime considered here, where both the mean offspring number $\lambda$ and the testing depth  grow jointly with $n$. Therefore, the asymptotic guarantees in~\cite{Luca24LRT} do not directly provide the required error bounds for the tests applied in \texttt{TBAlign}.
    \item The alignment algorithm in~\cite{Luca24LRT} uses local trees of depth $\Theta(\log n)$. This is specific to the constant-degree regime, where such neighborhoods can remain acyclic up to logarithmic depth. When $\lambda=(\log n)^{\alpha+o(1)}$, local neighborhoods are acyclic only up to depth of order $\log n/\log\log n$. In this work, we further choose the testing depth to be $d=(\log n)^\gamma$ for some $\gamma\in(1-\alpha,1)$, which is well below this scale. This choice keeps the tested trees of size $n^{o(1)}$, and is crucial for obtaining the $n^{o(1)}$ run-time bound for each likelihood ratio computation.
\end{itemize}

These two differences require a new analysis of tree correlation tests in the diverging-degree regime. We take an \emph{indirect route} in this analysis. Instead of analyzing the likelihood-ratio test directly, we first construct an auxiliary test that is weaker than the likelihood-ratio test but more statistically tractable. This auxiliary test is \emph{also computationally infeasible} and is used only as an analytical device. The weaker test consists of three steps: (i) baseline tree correlation test, (ii) iterative boosting, and (iii) maximum matching for the hypothesis testing. We note that the analysis of the baseline and iterative boosting follow the results established in \cite{Luca24Stat}. Based on this weaker correlation test, we prove the existence of the threshold of the likelihood ratio test based on the Neyman-Pearson lemma, summarized in the following theorem. Note that the result itself is a new result for tree correlation tests as existing results only apply to trees with constant degrees.

\begin{restatable}{theo}{RankLRT}
\label{prop:LRT}
Let $\alpha\in(0,1)$, $\gamma\in(1-\alpha,1)$, and $s\in(\sqrt{C_{\mathrm{Otter}}},1]$ be constants independent of $n$.
Assume that
\begin{align}
    \lambda=(\log n)^{\alpha+o(1)}
    \quad \text{and} \quad
    d=(\log n)^\gamma.
    \nonumber
\end{align}
Fix an arbitrary injective mapping $f:\cz_d\times\cz_d\rightarrow[0,1]$. 
Then there exist a likelihood-ratio threshold $\theta^*\in\mathbb{R}$, and a tie-breaking threshold $\kappa^*\in[0,1]$ such that the test $\Psi:\cz_d\times \cz_d\to\{0,1\}$ defined by
\begin{equation}
    \Psi(t,\td t)
    =
    \begin{cases}
    1, & \text{if } L_d( t,\td t)>\theta^*, \\
    1, & \text{if } L_d( t,\td t)=\theta^* \hbox{ and }f(t,\td t)>\kappa^*, \\
    0, & \text{otherwise},
    \end{cases}\label{eq:rank-LRT}
\end{equation}
satisfies
\begin{align}
    {\bq}\bigl(\Psi([\ch_d],[\td\ch_d])=1\bigr)
    &=
    O\left(n^{-\log\log n}\right),
    \label{eq:rank-LRT-typeI}
    \\
    {\bp}\bigl(\Psi([\ch_d],[\td\ch_d])=1\bigr)
    &=
    1-O\bigl((\log n)^{-\alpha/4}\bigr).\label{eq:rank-LRT-power}
\end{align}
\end{restatable}

The proof of this theorem is deferred to Appendix~\ref{appd:pf_thm_3}. 
We emphasize that Theorem~\ref{prop:LRT} is an \emph{existence result}. It establishes the existence of thresholds $\theta^*$ and $\kappa^*$ with the desired type-I error and power guarantees, but it does not provide explicit analytical formulas for these thresholds. Computing them by brute force is also infeasible, as it would require enumerating the space $\cz_d\times\cz_d$ of rooted tree pairs. Consequently, we view \texttt{TBAlign} as a stepping stone toward the analysis of \texttt{RBAlign}, rather than as an implementable algorithm.

With this theorem, we move on to prove Proposition~\ref{th:TBAlign}.
This proof is separated into two steps. In the first step, we analyze \texttt{TBAlign} with $\tau=\infty$, i.e., no degree restrictions are imposed on the local trees, and show that such an algorithm achieves almost exact recovery. In the second step, we show through degree concentration bounds in \erdos--\renyi graphs that adding the degree constraint $\tau=(\log n)^t$ only affects the matching of a vanishing fraction of the vertices, and hence almost exact recovery is still achieved.

\subsection[Step 1. Analysis with tau = infty]{Step 1. Analysis with $\tau=\infty$}
Let ${\bf M}_\mathrm{TB}^\infty$ denote the matching matrix under \texttt{TBAlign} with $\tau=\infty$,
and $\cw^\infty(u,\td v)$ denote the witness graph for the vertex pair $(u,\td v)$ under $\tau=\infty$.
Let $\bar{G}$ denote the union graph of $G$ and $\td G$. We use the bar notation to distinguish the notations associated with $\bar G$. More specifically, $\bar\cv$ and $\bar\ce$ are the vertex set and the edge set of $\bar G$ respectively. Moreover,  $u,$ $\tilde u$, and $\bar u$ represent the {\em same} node but in different graphs $G,$ $\tilde G,$ and $\bar G,$ respectively. 

For a vertex $u\in{\cv}$, we define the following two events: 
    \begin{align}
\ca_{1,u}:=\Bigl\{&
\nexists({\ct}^{\infty}_{w\setminus u},
{\td\ct}^{\infty}_{\td x\setminus \td v})
\in{\cc}^{\infty}_d\times {\td\cc}_d^{\infty}\text{ such that }
\Psi([{\ct}^{\infty}_{w\setminus u}],
[{\td\ct}^{\infty}_{\td x\setminus \td v}])=1
\text{ and } \notag\\
&\quad \bar{\cv}({\ct}^{\infty}_{w\setminus u})
\cap \bar{\cv}({\td\ct}^{\infty}_{\td x\setminus \td v})=\emptyset \Bigr\},\label{eq:E2}
\end{align}
and 
\begin{align}
    \ca_{2,u}&:=\left\{\texttt{GreedyMatching}({\cw}^{\infty}(u,\td u),\theta^*,3)=1\right\},
    \label{eq:E3u}
\end{align}

The event $\ca_{1,u}$ controls false positive matches with vertex $u$, as it ensures that no pair of disjoint local trees can pass the tree-correlation test. In contrast, the event $\ca_{2,u}$ ensures that the algorithm matches $u$ with its counterpart $\td u$ in $\td G$. The following two lemmas show that these two events occur with high probability. Their proofs are based on the type-I error bound~\eqref{eq:rank-LRT-typeI} and the power bound~\eqref{eq:rank-LRT-power} in Theorem~\ref{prop:LRT}, respectively, and are deferred to the end of this section.

 \begin{restatable}{lem}{TBAlignIdealEtwo}
\label{lem:TBAlignIdeal-E2}
Under the assumptions in Proposition~\ref{th:TBAlign}, we have
\begin{align}
    \P(\ca^c_{1,u})
    =
    \td O\bigl(n^{-1/3}\bigr).
    \nonumber
\end{align}
\end{restatable}
\begin{restatable}{lem}{TBAlignIdealEthree}
\label{lem:TBAlignIdeal-E3}
Under the assumptions in Proposition~\ref{th:TBAlign}, we have
\begin{align}
    \P(\ca^c_{2,u})
    =
    O\left(\exp\left(-(\log n)^{\alpha/2}\right)\right)
    \nonumber
\end{align}
\end{restatable}

For a vertex $u$, we also define the event
\begin{align}
        \ca_{3, u}
    &:=
    \left\{
        \text{the }2(d+1)\text{-hop neighborhood of }\bar u
        \text{ in }\bar G
        \text{ is a tree}
    \right\}.
    \label{eq:E1u}
    \end{align}
We show in Lemma~\ref{lem:TBAlignIdeal-E1} in Appendix~\ref{appd:rare} that $\ca_{3,u}$ occurs with high probability. We note that the 3-dangling-tree test applied in \texttt{TBAlign} was first proposed in~\cite{Luca24LRT}, where it was shown that, the event $\ca_{1,u}\cap\ca_{2,u}\cap\ca_{3,u}$ implies ${\bf M}^\infty_\mathrm{TB}(u,\td u) =1 \text{ and } {{\bf M}^\infty_\mathrm{TB}(u,\td v)}=0$, for all $\td v\neq \td u$. For completeness, we include this argument as Lemma~\ref{lem:sufficient} in Appendix~\ref{appd:rare}. Therefore, for any vertex $u\in \cv$, 
\begin{align}
    \P\Bigl(
        {\bf M}_\mathrm{TB}^\infty(u,\td u)=1
        \text{ and }
        \|{\bf M}_\mathrm{TB}^\infty(u,:)\|_1=1
    \Bigr)
    &\ge
    \P\bigl(
        \ca_{1,u}\cap \ca_{2,u}\cap \ca_{3,u}
    \bigr)\nonumber\\
    &\ge
    1-\P(\ca_{1,u}^c)-\P(\ca_{2,u}^c)-\P(\ca_{3,u}^c)
    \nonumber\\
    &=
    1-O\left(\exp\left(-(\log n)^{\alpha/2}\right)\right),
    \nonumber
\end{align}
where the last step uses Lemmas \ref{lem:TBAlignIdeal-E2}, \ref{lem:TBAlignIdeal-E3}, and \ref{lem:TBAlignIdeal-E1}.
By the linearity of expectation,
\begin{align}
    &\E\left[n- \sum_{u\in{\cv}} \mathbbm{1}\Bigl\{
            {\bf M}_\mathrm{TB}^\infty(u,\td u)=1 \text{ and } \|{\bf M}_\mathrm{TB}^\infty(u,:)\|_1=1 \Bigr\} \right]\\
    ={}&n-
    \sum_{u\in{\cv}}
    \P\Bigl(
        {\bf M}_\mathrm{TB}^\infty(u,\td u)= 1
        \text{ and }
        \|{\bf M}_\mathrm{TB}^\infty(u,:)\|_1 =1
    \Bigr)\\
    ={}&
    O\left(n\exp\left(-(\log n)^{\alpha/2}\right)\right).
    \nonumber
\end{align}
Hence, by Markov's inequality,
\begin{align}
    &\P\left( \sum_{u\in{\cv}} \mathbbm{1}\Bigl\{
            {\bf M}_\mathrm{TB}^\infty(u,\td u)=1 \text{ and } \|{\bf M}_\mathrm{TB}^\infty(u,:)\|_1=1 \Bigr\}< n-n\exp\left(-\frac12(\log n)^{\alpha/2}\right) \right)\\
    ={}&\P\left(  n-\sum_{u\in{\cv}}
        \mathbbm{1}\Bigl\{ {\bf M}_\mathrm{TB}^\infty(u,\td u)=1 \text{ and } \|{\bf M}_\mathrm{TB}^\infty(u,:)\|_1=1 \Bigr\}> n\exp\left(-\frac12(\log n)^{\alpha/2}\right) \right)\\
    \le{} & \frac{\E\left[n- \sum_{u\in{\cv}}
            \mathbbm{1}\Bigl\{ {\bf M}_\mathrm{TB}^\infty(u,\td u)=1 \text{ and } \|{\bf M}_\mathrm{TB}^\infty(u,:)\|_1=1
            \Bigr\} \right] }{n\exp\left(-\frac12(\log n)^{\alpha/2}\right)} \nonumber\\
    ={}& O\left(\exp\left(-\frac12(\log n)^{\alpha/2}\right)\right).
    \label{eq:ideal-count}
\end{align}

\subsection[Step 2. Analysis with tau=(log n) t]{Step 2. Analysis with $\tau=(\log n)^t$}
We now reduce $\tau$ to the actual value used under \texttt{TBAlign,} i.e., $\tau=(\log n)^{t}$. Recall that ${\bf M}_\mathrm{TB}(u,\tilde{v})$ is the matching matrix with $\tau=(\log n)^t.$ We define $\cb$ (resp. $\td\cb$) to be the set of vertices in $G$ (resp. $\td G$) whose $(d+1)$-hop neighborhood contains a vertex with degree larger than $(\log n)^t$, i.e.,
\begin{align}
    \cb
    &:=
    \left\{
    u\in {\cv}:
    \exists v\in {\cv}
    \text{ such that }
    \dist(u,v)\le d+1
    \text{ and }
    \deg(v)> (\log n)^t
    \right\},
    \label{eq:Bset}
    \\
    \td\cb
    &:=
    \left\{
    \td u\in \td{\cv}:
    \exists \td v\in \td{\cv}
    \text{ such that }
    \dist(\td u,\td v)\le d+1
    \text{ and }
    \deg(\td v)> (\log n)^t
    \right\}.
    \label{eq:tBset}
\end{align}
We note that
\begin{itemize}
    \item Given that $u\notin {\cb},$ a local tree ${\ct}_{w\setminus u},$ where $w\in{\cn}_u$, remains the same for $\tau=\infty$ and $\tau=(\log n)^t$ because (i) no vertex in the $(d+1)$-hop neighborhood of $u$ has degree larger than $(\log n)^t,$ and (ii) any vertex in the tree ${\ct}_{w\setminus u}$ is at most $(d+1)$ hops away from $u$. The same conclusion holds for a vertex $\td{v}\not\in \td \cb.$
Therefore, we have 
\begin{align}
    {\bf M}_\mathrm{TB}(u,\td v)= {\bf M}_\mathrm{TB}^\infty(u,\td v),&
    \qquad
    \forall (u,\td v)\in {\cb}^c\times \td{\cb}^c.
\end{align}

\item On the other hand, if $u\in {\cb},$ ${\ct}_{w\setminus u}$ may become an empty tree if it includes a vertex with degree $\geq(\log n)^t$, which would result in ${\bf L}({\ct}_{w\setminus u}, \td{\ct}_{\tilde{x}\setminus \tilde{v}})=0$ and reduce the number of edges with non-zero weights in the witness graph ${\cw}(u,\td v).$ When this happens, we may have ${\bf M}_\mathrm{TB}(u, \tilde{v})=0$ while ${\bf M}_\mathrm{TB}^\infty(u, \tilde{v})=1;$ also, ${\bf M}_\mathrm{TB}(u, \tilde{v})=0$ if ${\bf M}_\mathrm{TB}^\infty(u, \tilde{v})=0.$ In other words,  
\begin{align}
    {\bf M}_\mathrm{TB}(u,\td v)\leq {\bf M}_\mathrm{TB}^\infty(u,\td v),&
    \qquad
    \hbox{if } u\in {\cb} \hbox{ or } \tilde{v}\in\td{\cb}.
\end{align}
\end{itemize}
The results above imply that for $(u,\td u)\in {\cb}^c\times \td{\cb}^c,$ 
\begin{align*}
    \mathbbm{1}\Bigl\{
        {\bf M}_\mathrm{TB}(u,\td u)=1
        \text{ and }
        \|{\bf M}_\mathrm{TB}(u,:)\|_1=1
    \Bigr\} =1 \quad\hbox{if}\quad    \mathbbm{1}\Bigl\{
        {\bf M}_\mathrm{TB}^\infty(u,\td u)=1
        \text{ and }
        \|{\bf M}_\mathrm{TB}^\infty(u,:)\|_1=1
    \Bigr\} =1
\end{align*} so
\begin{align}
&\sum_{u\in{\cv}}
    \mathbbm{1}\Bigl\{
        {\bf M}_\mathrm{TB}(u,\td u)=1
        \text{ and }
        \|{\bf M}_\mathrm{TB}(u,:)\|_1=1
    \Bigr\}\nonumber\\
    {}\ge &
    \sum_{u\in{\cv}}
    \mathbbm{1}\Bigl\{
        {\bf M}_\mathrm{TB}^\infty(u,\td u)=1
        \text{ and }
        \|{\bf M}_\mathrm{TB}^\infty(u,:)\|_1=1  \hbox{ and } (u, \td u)\in{\cb}^c\times\td{\cb}^c
    \Bigr\}
    \nonumber\\
    {}\ge&
    \sum_{u\in{\cv}}
    \mathbbm{1}\Bigl\{
        {\bf M}_\mathrm{TB}^\infty(u,\td u)=1
        \text{ and }
        \|{\bf M}_\mathrm{TB}^\infty(u,:)\|_1=1
    \Bigr\}
    -
    \sum_{u\in{\cv}}
    \mathbbm{1}\{u\in\cb \text{ or } \td u\in\td\cb\}
    \nonumber\\
    {}\ge&
    \sum_{u\in{\cv}}
    \mathbbm{1}\Bigl\{
        {\bf M}_\mathrm{TB}^\infty(u,\td u)=1
        \text{ and }
        \|{\bf M}_\mathrm{TB}^\infty(u,:)\|_1=1
    \Bigr\}
    -
    |\cb|-|\td\cb|.
    \label{eq:reduction-count}
\end{align}
By \eqref{eq:ideal-count} and Lemma \ref{lem:TBAlignBadNodes} in Appendix~\ref{appd:bad-nodes}, in which we establish high-probability bounds on $|\cb|+|\td\cb|$, we obtain
\begin{align}
    &\sum_{u\in{\cv}}
    \mathbbm{1}\Bigl\{
        {\bf M}_\mathrm{TB}(u,\td u)=1
        \text{ and }
        \|{\bf M}_\mathrm{TB}(u,:)\|_1=1
    \Bigr\}\nonumber\\
    \ge{} & n-n\exp\left(-\frac12(\log n)^{\alpha/2}\right)-2n\exp\left(-\frac12(\log n)^{\alpha/2}\right)\nonumber\\
    ={}&
    n-3n\exp\left(-\frac12(\log n)^{\alpha/2}\right)
    \label{eq:TBAlign-final-count}
\end{align}
with probability $1-O(\exp(-\frac12(\log n)^{\alpha/2}))$. This completes the proof of Proposition~\ref{th:TBAlign}.

\subsection{Proof of Lemma~\ref{lem:TBAlignIdeal-E2}}
    Recall the definition of local tree $\ct_{w\setminus u}^\infty$ for a pair of adjacent vertices $w$ and $u$ in $G$.  In the following proof, we slightly extend this definition. 
    For a vertex $w\in\cv$, not necessarily adjacent to $u$ in $G$, we define $\ct^\infty_{w\setminus u}$ in the same way as before: it is the $d$-hop neighborhood of $w$ in the graph $G\setminus u$ if this neighborhood is a tree; otherwise, $\ct^\infty_{w\setminus u}$ is defined to be the empty tree. We extend the definition of $\td\ct^\infty_{\td x\setminus\td v}$ to the case of non-adjacent $(\td x,\td v)$ in the similar way. 
    
    With this extended definition, we define $\ca_{1,u}'$ as the event that there exist $w\in \cv$ and $(\td v,\td x)\in \td \cv\times \td\cv$ such that
    \begin{enumerate}
        \item $\Psi([{\ct}^{\infty}_{w\setminus u}],[{\td\ct}^{\infty}_{\td x\setminus \td v}])=1$;

        \item $\bar\cv(\ct^\infty_{w\setminus u})\cap \bar\cv(\td\ct^\infty_{\td x\setminus \td v})=\emptyset$. 
    \end{enumerate}
    By definition, we have $\ca_{1,u}^c\subseteq \ca_{1,u}'$, so to prove the lemma, it suffices to prove
    $\P(\ca_{1,u}')=\td O(n^{-1/3}).$

    We also define an auxiliary event 
    \[
    \cd:=\{|\cn(i,d)|< 8e^2\lambda^d\log n,\  \forall i\in{\cv}\}\bigcap\{|\td\cn(\td j,d)|< 8e^2\lambda^d\log n,\ \forall \tilde j \in \td{\cv}\},
    \]
    where $\cn(i,d)$ (resp. $\td\cn(\td j,d)$) denotes the set of vertices within distance $d$ from $i$ in $G$ (resp. $\td j$ in $\td G$). It follows by Lemma~\ref{lem:neighbor_size} in Appendix~\ref{appd:neighbor-size} and the union bound that
    \begin{equation}
    \label{eq:cd}
    \P(\cd^c)=\td O(n^{-1/3}).
    \end{equation}
    Then we have
    \begin{align}
        \P(\ca_{1,u}')&\le \P(\ca_{1,u}'\cap \cd)+\P(\cd^c)\nonumber\\
        &\le \P(\cd^c)+\sum_{\substack{
    w\in\cv\\
    (\td v,\td x)\in \td\cv\times\td\cv
    }}\P(\{\Psi([{\ct}^{\infty}_{w\setminus u}],[{\td\ct}^{\infty}_{\td x\setminus \td v}])=1\}\cap \{\bar\cv(\ct^\infty_{w\setminus u})\cap \bar\cv(\td\ct^\infty_{\td x\setminus \td v})=\emptyset\}\cap\cd)\label{eq:sum_pairs}
    \end{align}

    Now, we fix vertex pairs $w\in\cv$ and $(\td v,\td x)\in \td\cv\times\td\cv$. We omit the subscript and superscript from notations $\ct^\infty_{w\setminus u}$ and $\td\ct^\infty_{\td x\setminus \td v}$, and simply write $\ct$ and $\td\ct$. Recall that $\cz_d$ denote the set of all rooted unlabeled trees with depth up to $d$. We then have
    \begin{align}
        &\P(\{\Psi([{\ct}],[{\td\ct}])=1\}\cap \{\bar\cv(\ct)\cap \bar\cv(\td\ct)=\emptyset\}\cap\cd)\nonumber\\
        ={}& \sum_{\substack{
        (t,\td t)\in \cz_d\times\cz_d
        }} \P\big(\{[{\ct}]=t\}\cap\{[{\td\ct}]=\td t\}\cap \{\bar\cv(\ct)\cap \bar\cv(\td\ct)=\emptyset\}\cap\cd\big)\mathbbm{1}\{\Psi(t,\td t)=1\}\nonumber\\
        \stackrel{(a)}{\le} {} & \sum_{\substack{
    (t,\td t)\in \cz_d\times\cz_d:\\
    |\cv(t)|\le 8e^2\lambda^d\log n\\
    |\cv(\td t)|\le 8e^2\lambda^d\log n 
    }} \P\big(\{[{\ct}]=t\}\cap\{[{\td\ct}]=\td t\}\cap \{\bar\cv(\ct)\cap \bar\cv(\td\ct)=\emptyset\}\big)\mathbbm{1}\{\Psi(t,\td t)=1\}\label{eq:expend_t}
    \end{align}
    where (a) follows because both $\ct$ and $\td\ct$ have at most $8e^2\lambda^d\log n$ vertices under event $\cd$.

    Recall that $\bq(\ch,\td\ch)$ is the independent tree-pair distribution.
    By Lemma~\ref{lem:tree_pmf} in Appendix~\ref{appd:tree-pmf}, we know that for any pair of unlabeled trees $(t,\td t)\in \cz_d\times\cz_d$ satisfying $|\cv(t)|\le 8e^2\lambda^d\log n$ and  $|\cv(\td t)|\le 8e^2\lambda^d\log n$, 
    \begin{align}
        &\frac{\P\big(\{[{\ct}]=t\}\cap\{[{\td\ct}]=\td t\}\cap \{\bar\cv(\ct)\cap \bar\cv(\td\ct)=\emptyset\}\big)}{\bq([\ch_d]=t,[\td\ch_d]=\td t)}\nonumber\\
        \le{} & \exp\left(\frac{\lambda}{n}\left((8e^2\lambda^d\log n+1)8e^2\lambda^d\log n+8e^2\lambda^d\log n(1+8e^2\lambda^d\log n+8e^2\lambda^d\log n)\right)\right)\nonumber\\
        \le{} & \exp\left(\frac{\lambda}{n}320e^4\lambda^{2d}(\log n)^2\right)=O(1),\label{eq:const_ratio}
    \end{align}
    where the last equality follows because $\lambda^{2d}=n^{o(1)}$. Then~\eqref{eq:expend_t} and~\eqref{eq:const_ratio} together imply that
    \begin{align}
        &\P\big(\{[{\ct}]=t\}\cap\{[{\td\ct}]=\td t\}\cap \{\bar\cv(\ct)\cap \bar\cv(\td\ct)=\emptyset\}\big)\nonumber\\
        \le{}&  \exp\left(\frac{\lambda}{n}320e^4\lambda^{2d}(\log n)^2\right)\sum_{\substack{
    (t,\td t)\in \cz_d\times\cz_d\\
    |\cv(t)|\le 8e^2\lambda^d\log n\\
    |\cv(\td t)|\le 8e^2\lambda^d\log n 
    }}\bq([\ch_d]=t,[\td\ch_d]=\td t)\mathbbm{1}\{\Psi(t,\td t)=1\}\nonumber\\
    \le{}&  \exp\left(\frac{\lambda}{n}320e^4\lambda^{2d}(\log n)^2\right)\bq(\Psi([\ch_d],[\td\ch_d])=1)=O(n^{-\log\log n}),\label{eq:fix_uwvx}
    \end{align}
    where the last equality follows by~\eqref{eq:rank-LRT-typeI} in Theorem~\ref{prop:LRT}. Notice that~\eqref{eq:fix_uwvx} holds for any $w$ and $(\td v,\td x)$. Substituting~\eqref{eq:fix_uwvx} and~\eqref{eq:cd} into~\eqref{eq:sum_pairs} yields
    \begin{equation}
        \P(\ca_{1,u}')=\td O(n^{-1/3})+O(n^{3-\log\log n})=\td O(n^{-1/3}),
    \end{equation}
    which completes the proof.

\subsection{Proof of Lemma~\ref{lem:TBAlignIdeal-E3}}
    Recall that we defined 
    \[
    \ca_{2,u}
    :=
    \left\{
        \texttt{GreedyMatching}({\cw}^{\infty}(u,\td u),\theta^*,3)=1
    \right\}.
    \]
    Define $\ct$ as the local tree rooted at $u$ in $G$ such that $\ct$ is the $(d+1)$-hop-neighborhood of $u$ in $G$ if it is a tree; otherwise, $\ct$ is an empty tree. We similarly define $\td\ct$ as the local tree rooted at $\td u$ in $\td G$. By Lemma~\ref{lem:tree_tv} in Appendix~\ref{appd:tree-couple}, we can couple $(\ct,\tdct)$ and a pair of correlated Poisson Galton--Watson trees $(\ch,\td\ch)\sim\bp$ so that
    \begin{equation}
        \label{eq:tree_couple}
        \P([\ct]=[\ch_{d+1}],[\td\ct]=[\td\ch_{d+1}])= 1-O(n^{-1/2}),
    \end{equation}
    where $\ch_{d+1}$ is the truncation of $\ch$ up to depth $d+1$, and $\td\ch_{d+1}$ is the truncation of $\td\ch$ up to depth $d+1$.

    We now switch to study the output of the subroutine \texttt{GreedyMatching} when applied to
    the correlated tree pair $\ch_{d+1}$ and $\td\ch_{d+1}$
    More precisely, let $r$ and $\td r$ denote the root node of $\ch_{d+1}$ and $\td\ch_{d+1}$ respectively. Define $\sfch_r$ as the set of children of $r$ in $\ch_{d+1}$, and for each $i\in \sfch_r$, define $\ch_{d+1,i}$ to be the subtree rooted at $i$ in $\ch_{d+1}$. We also define $\wdtd\sfch_{\td r}$ and $\td\ch_{d+1,\td j}$ for each $\td j\in\wdtd\sfch_{\td r}$ similarly. The witness graph $\cw^\infty(r,\td r)$ is defined as the bipartite graph between vertex sets $\sfch_r$  and $\wdtd\sfch_{\td r}$ with the weight of edge $(i,\td j)$ being $L_{d}([\ch_{d+1,i}],[\td\ch_{d+1,\td j}])$. Notice that under the event $[\ct]=[\ch_{d+1}] $ and $ [\td\ct]=[\td\ch_{d+1}]$, the two witness graphs $\cw^\infty(u,\td u)$ and $\cw^\infty(r,\td r)$ are identical under a vertex relabeling. Therefore, we have 
    \[
    \{\texttt{GreedyMatching}({\cw}^{\infty}(u,\td u),\theta^*,3)=1\}\Leftrightarrow \{\texttt{GreedyMatching}({\cw}^{\infty}(r,\td r),\theta^*,3)=1\}.
    \]

    We define the random variable
    \[
        X
        :=
        \sum_{i\in \sfch_r^*}
        \Psi\bigl([\ch_{d+1,i}],[\td\ch_{d+1,\td i}]\bigr),
    \]
    where $\sfch_r^*$ denotes the set of children of the root in the intersection tree $\ch^*$. Equivalently, $X$ is the number of children $i$ of the root in $\ch^*$ for which either
    \[
        L_d([\ch_{d+1,i}],[\td\ch_{d+1,\td i}])>\theta^*
    \]
    or 
    \[
        L_d([\ch_{d+1,i}],[\td\ch_{d+1,\td i}])=\theta^*
        \quad\text{and}\quad
        f([\ch_{d+1,i}],[\td\ch_{d+1,\td i}])>\kappa^*.
    \]
    We claim that the event $X\ge 5$ implies $\texttt{GreedyMatching}({\cw}^{\infty}(r,\td r),\theta^*,3)=1$. Under the event $X\ge 5$, there exists five non-overlapping vertex pairs $(i_1,\td i_1),\ldots,(i_5,\td i_5)\in \sfch_r\times\wdtd\sfch_{\td r}$ such that for each $k\in [5]$, either 
    \[
        L_d([\ch_{d+1,i_k}],[\td\ch_{d+1,\td i_k}])>\theta^*
    \]
    or 
    \[
        L_d([\ch_{d+1,i_k}],[\td\ch_{d+1,\td i_k}])=\theta^*
        \quad\text{and}\quad
        f([\ch_{d+1,i_k}],[\td\ch_{d+1,\td i_k}])>\kappa^*.
    \]
    Towards a contradiction, suppose $\texttt{GreedyMatching}({\cw}^{\infty}(r,\td r),\theta^*,3)=0$. Then, the matching found by subroutine before it terminates has size less than or equal to $2$, which includes at most two vertices in $\sfch_r$  and at most two vertices in $\wdtd\sfch_{\td r}$.  This implies that there must exist some $k^*\in [5]$ such that the edge $(i_{k^*},\td i_{k^*})$ does not conflict with any edge in the matching. Therefore, the subroutine must not terminate at this point, which leads to a contradiction.

    Notice that the number of children in the set $\sfch^*_r$ follows distribution $\poi(\lambda s)$. Moreover, for each $i\in\sfch^*_r$, the event  
    $\Psi\bigl([\ch_{d+1,i}],[\td\ch_{d+1,\td i}]\bigr)=1$ happens with probability
    $\bp(\Psi([\ch_d],[\td\ch_d])=1)$, and these events are mutually independent among different $i$'s.
    It then follows by Poisson thinning that 
    $X\sim\poi\left(\lambda s\bp(\Psi([\ch_d],[\td\ch_d])=1)\right).$
    Therefore, we have
    \begin{align}
        \P(\texttt{GreedyMatching}({\cw}^{\infty}(r,\td r),\theta^*,3)=0)&\le \P(X\le 4)\nonumber\\
        &= \P\left(\poi\left(\lambda s\bp(\Psi([\ch_d],[\td\ch_d])=1)\right)\le 4\right)\nonumber.
    \end{align}
    Notice that for any $\mu\ge 1$, we have
    \[
    \P(\poi(\mu)\le 4)
    =e^{-\mu}\sum_{j=0}^4\frac{\mu^j}{j!}
    =O(\mu^4\exp(-\mu)).
    \]
    Then by~\eqref{eq:rank-LRT-power} in Theorem~\ref{prop:LRT} and the fact that $\lambda=(\log n)^{\alpha+o(1)}$, we have
    \begin{align}
        \P(\texttt{GreedyMatching}({\cw}^{\infty}(r,\td r),\theta^*,3)=0)=\exp\left(-(\log n)^{\alpha+o(1)}\right)=O\left(\exp\left(-(\log n)^{\alpha/2}\right)\right).\label{eq:tree_pass}
    \end{align}

    Finally, by equations~\eqref{eq:tree_couple} and~\eqref{eq:tree_pass}, we have
    \begin{align*}
    \P\left(\left\{
        \texttt{GreedyMatching}({\cw}^{\infty}(u,\td u),\theta^*,3)=1
    \right\}\right)&=1-O(n^{-1/2})-O\left(\exp\left(-(\log n)^{\alpha/2}\right)\right)\\
    &=1-O\left(\exp\left(-(\log n)^{\alpha/2}\right)\right),
    \end{align*}
    which completes the proof.
    
\section{Conclusion}
This paper studies graph matching under the correlated \erdos--\renyi graph pair model in the regime $\lambda=(\log n)^{\alpha+o(1)}$ for some $\alpha\in(0,1)$. We developed an algorithm with $n^{2+o(1)}$ time complexity and achieves almost exact recovery when the edge correlation parameter $s\in(\sqrt{C_{\mathrm{Otter}}},1]$. This improves over the previously known chandelier-counting algorithm of~\cite{Cheng26Chandelier}, whose running-time exponent depends on $s$ and goes to $\infty$ as $s$ approaches $\sqrt{C_{\mathrm{Otter}}}$ from above.

Algorithmically, our method is built upon local tree likelihood-ratio tests. 
Different from threshold-based likelihood-ratio tests in the literature, the proposed algorithm ranks the likelihood ratios of all observed local tree pairs and matches the vertex pairs by considering the local tree pairs sequentially according to their ranks, which avoids the calculation of an explicit threshold. 

\section*{Acknowledgments}
This work was supported in part by U.S. National Science Foundation (NSF) under grants 2207548 and 2324769; and AFOSR grant FA9550-24-1-0002. 

\newpage 
\bibliographystyle{IEEEtranN}
\bibliography{main3.bib}

\newpage 
\appendices
 The appendix is organized as follows: Appendix~\ref{appd:pf_thm_3} proves Theorem \ref{prop:LRT}, our guarantee for tree correlation testing in diverging degree regime. Appendix~\ref{appd:gw-lemmas} collects the technical lemmas for \GW trees, and Appendix~\ref{appd:rare} includes the additional lemmas needed for the proof of Proposition \ref{th:TBAlign}. Appendix~\ref{appd:er-properties} includes properties for \erdos--\renyi graphs, and Appendix~\ref{appd:concentration} includes the Poisson approximation results and concentration inequalities used in the proof. Finally, Appendix~\ref{appd:treecode} gives an explicit construction of the tie-breaking map~$f$.

\section{Proof of Theorem \ref{prop:LRT}: Tree Correlation Tests in the Diverging-Degree Regime}\label{appd:pf_thm_3}
Recall that $\cz_d$ denotes the set of rooted unlabeled trees of depth at most $d$. We consider the problem of deciding whether an observed tree pair $(t,\td t)\in\cz_d\times\cz_d$ is drawn from the correlated Galton--Watson distribution $\bp$ or from the independent Galton--Watson distribution $\bq$. Fix an arbitrary injective map $f:\cz_d\times\cz_d\to[0,1]$, which can be viewed as a tie-breaking function and gives a canonical ordering on the tree pairs in $\cz_d\times\cz_d$. The main goal of this section is to establish the guarantee for the likelihood-ratio test, which is restated as follows.

\RankLRT*

We prove Theorem~\ref{prop:LRT} via constructing an auxiliary correlation test $\Phi$. The following proposition characterizes the type-I error and power of $\Phi$, from which, Theorem~\ref{prop:LRT} follows by the Neyman--Pearson lemma.

\begin{restatable}{prop}{CorTest}
\label{prop:CorTest}
Assume the same setting as in Theorem \ref{prop:LRT}.
Then there exists a test $\Phi:\cz_d\times \cz_d\to\{0,1\}$
that satisfies
\begin{align}
    {\bq}\bigl(\Phi([\ch_d],[\td\ch_d])=1\bigr)
    &\le
    n^{-\log\log n},
    \label{eq:CorTest-typeI}
    \\
    {\bp}\bigl(\Phi([\ch_d],[\td\ch_d])=1\bigr)
    &=
    1-O\bigl((\log n)^{-\alpha/2}\bigr).
    \label{eq:CorTest-power}
\end{align}
for all sufficiently large $n$.
\end{restatable}

The remainder of this section is organized as follows. In Appendix~\ref{appd:construction}, we construct the test $\Phi$ which consists of three steps: initialization, iterative boosting, and bipartite maximum matching. The proof of Proposition~\ref{prop:CorTest} is presented in Appendix~\ref{appd:pf-cortest}, with supporting results in Appendices \ref{appd:fd-l}, \ref{appd:init}, and \ref{appd:boost}. Finally, in Appendix~\ref{appd:pf-lrt}, we prove Theorem~\ref{prop:LRT} using Proposition~\ref{prop:CorTest} and the Neyman--Pearson lemma.

\subsection[Construction of the test Phi]{Construction of the test $\Phi$}
\label{appd:construction}
We construct the test $\Phi$ in the following three steps.

\underline{\textbf{Step 1: Initialization}.} We initialize the construction at the base depth $\bar d$. Define the initial set $\cf_{\bar d}\subseteq\cz_{\bar d}\times\cz_{\bar d}$ by
\[
    \cf_{\bar d}:=\bigl\{(t, \td t) \in \cz_{\bar d} \times \cz_{\bar d} : L_{\bar d}(t, \td t) \ge C_1\bigr\},
\]
where $C_1$ and $\bar d$ are constants independent of $n$, which are specified later in Appendix~\ref{appd:init}.

\underline{\textbf{Step 2: Iterative boosting}.} Given $\cf_k$, we construct $\cf_{k+1}$ for each $k=\bar d,\ldots,d-l-1$, where $l=\beta\log\log n$ and $\beta>0$ is a constant independent of $n$, whose value will be specified in Appendix~\ref{appd:pf-cortest}.

For $t\in\cz_{k+1}$ and $h\in\cz_k$, let $N_h(t)\in\mathbb N$ denote the number of children $v$ of the root of $t$ such that the subtree of $t$ rooted at $v$ is isomorphic to $h$. We also define $N_{\td h}(\td t)$ analogously for $\td t\in\cz_{k+1}$ and $\td h\in\cz_k$. Let
\[
    p_h := \gw([\ch_k] = h),
\]
where $\ch$ is a \GW tree with $\mathrm{Poi}(\lambda)$ offsprings.
We then construct the set 
\begin{equation}
\cf_{k+1}:=\left\{(t, \td t) \in \cz_{k+1} \times \cz_{k+1} :\sum_{(h, \td h) \in \cf_k}\bigl(N_h(t) - \lambda p_h\bigr)\bigl(N_{\td h}(\td t) - \lambda p_{\td h}\bigr)> \sigma_k\right\},\label{eq:def_iter}
\end{equation}
where the threshold $\sigma_k$ is chosen as
\begin{align}
    \sigma_k:=\max\left\{\frac1{10}\lambda e^{\frac{1}{4}(\bar{d}-k)}, 3\lambda^{3/4}\right\}.\label{eq:s2th}
\end{align}
Iterating this procedure for $k=\bar d,\ldots,d-l-1$ yields the set $\cf_{d-l}$.

\underline{\textbf{Step 3: Bipartite matching}.} Finally, we use $\cf_{d-l}$ to define the test $\Phi$ via the maximum matching over a bipartite graph $\cw$ defined below. Given a pair of rooted unlabeled trees
$(t,\td t)\in\cz_d\times\cz_d$, we first assign arbitrary labels to the
vertices in the $l$-th generations of $t$ and $\td t$, and denote the resulting
sets by $\cs_l(t)$ and $\td\cs_l(\td t)$, respectively.
We then construct a bipartite graph $\cw$ between $\cs_{l}(t)$ and $\td\cs_{l}(\td t)$ such that an edge exists between $u\in \cs_l(t)$ and $\td v \in \td\cs_l(\td t)$ if $(t_u,\td t_{\td v})\in \cf_{d-l}, $ where $t_u$ and $\td t_{\td v}$ are the
unlabeled subtrees of $t$ and $\td t$ rooted at $u$ and $\td v$, respectively. An illustration of this construction is provided in Figure~\ref{fig:gw-cortest-construction}.

\begin{figure}[htbp]
    \centering 
        \includegraphics[width=0.85\textwidth]{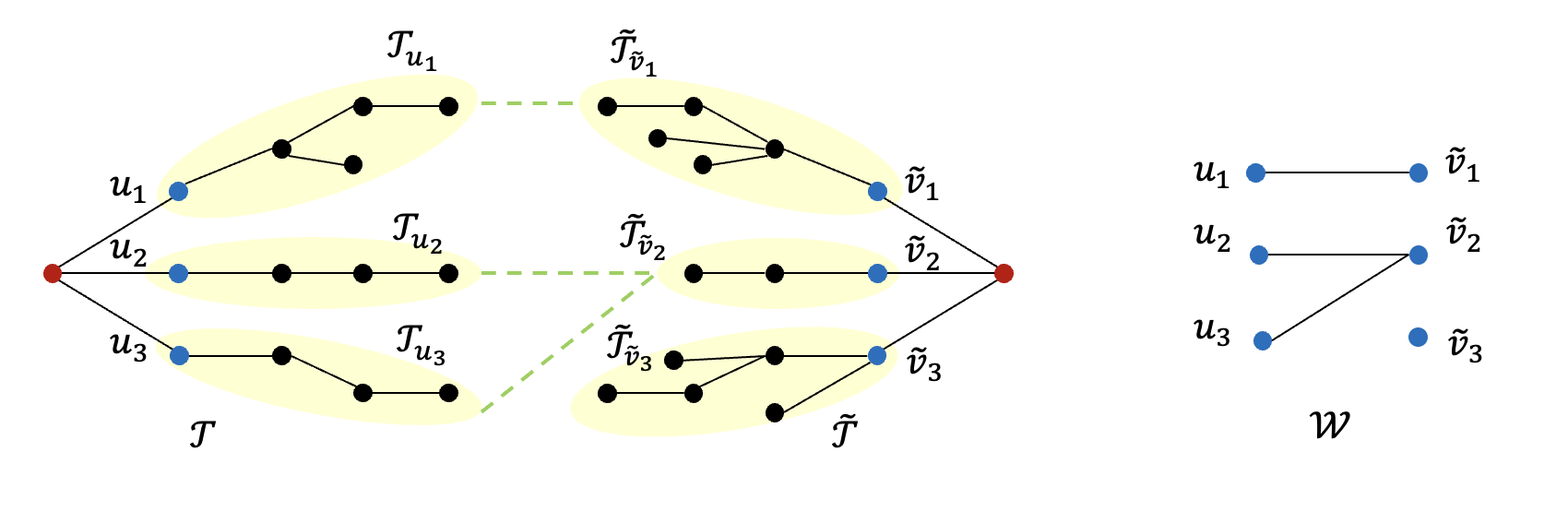}
        
        \caption{An example of constructing bipartite graph $\cw$ from trees $(t,\td t)\in\cz_4\times\cz_4$ when $l=1$. The red nodes in both tree $t$ and $\td t$ are roots, while the blue nodes in $t$ and $\td t$ denote the two parts in $\cw$. The green dashed line between two subtrees denotes that they appear in set $\cf_{3}$, which corresponds to the edges in the bipartite graph $\cw$.}
        \label{fig:gw-cortest-construction}
\end{figure}

We now define the test $\Phi$ such that it declares that $(t,\td t)$ is drawn from the
correlated model $\bp$ when the maximum matching size of $\cw$ is at least
$\frac{1}{2C_2}(\lambda s)^l$. That is,
    \begin{align}
    \Phi(t,\td t)
    =
    \mathbbm{1}
    \left\{
        \texttt{MaxMatching}\bigl(\cw\bigr)
        \ge
        \frac{1}{2C_2}(\lambda s)^{l}
    \right\},
    \label{eq:Phi_def}
\end{align}
where $\texttt{MaxMatching}(\cw)$ denotes the maximum matching size of $\cw$ and $C_2>1$ is a constant independent of $n$ that will be specified later in Appendix~\ref{appd:pf-cortest}.

\subsection{Proof of Proposition~\ref{prop:CorTest}}
\label{appd:pf-cortest}

Step~3 of the test $\Phi$ relies on set $\cf_{d-l}$ constructed in Steps~1 and~2. The key properties of $\cf_{d-l}$ are summarized in the following proposition. Its proof is deferred to Appendix~\ref{appd:fd-l} and is based on the analysis of
Steps~1 and~2 in Appendices~\ref{appd:init}
and~\ref{appd:boost}, respectively.

\begin{restatable}{prop}{lemmatching}
\label{lem:matching}
Let $\alpha\in(0,1)$, $\beta>0$, $\gamma\in(1-\alpha,1)$, and $s\in(\sqrt{C_{\mathrm{Otter}}},1]$ be constants. Suppose that
\begin{align}
    \lambda=(\log n)^{\alpha+o(1)},
    \qquad
    d=(\log n)^\gamma, \qquad l=\beta\log\log n.
    \nonumber
\end{align}
Then there exist constants $C_3>0$ and $C_4>0$ independent of $n$, such that the set $\cf_{d-l}\subseteq \cz_{d-l}\times\cz_{d-l}$ constructed in Appendix~\ref{appd:construction} satisfies
\begin{align}
    \bq&\bigl(([\ch_{d-l}],[\td\ch_{d-l}])\in\cf_{d-l}\bigr)\leq C_3e^{-d+l}, \label{eq:L2_Type1}\\
    \bp&\bigl(([\ch_{d-l}],[\td\ch_{d-l}])\in\cf_{d-l}\bigr)\geq  1-C_4 \lambda^{-1}, \label{eq:L2_Power}
\end{align}
for all large enough $n$.
\end{restatable}

We next prove Proposition~\ref{prop:CorTest} based on Proposition~\ref{lem:matching}.

\begin{proof}
Recall the test $\Phi$ constructed in Appendix~\ref{appd:construction} from the set $\cf_{d-l}$ of Proposition~\ref{lem:matching}. We prove that $\Phi$ satisfies the type-I error bound~\eqref{eq:CorTest-typeI} and the power bound~\eqref{eq:CorTest-power} in the following two parts.

\underline{\textbf{Part 1: Proof of \eqref{eq:CorTest-typeI}}.} Suppose $(\ch,\td\ch)\sim\bq$. Throughout the proof, let $\cs_l = \cs_l(\ch)$ and $\td\cs_l = \cs_l(\td\ch)$ be the level-$l$ vertex sets of $\ch$ and $\td\ch$, respectively. For each $u\in\cs_l$ and $\td v\in\td\cs_l$, let
$\ch_{d,u}$ and $\td\ch_{d,\td v}$ denote the subtrees of
$\ch_d$ and $\td\ch_d$ rooted at $u$ and $\td v$, respectively.
We write $\cs_l=\{1, 2, \ldots, m\}$, $\td\cs_l=\{1, 2, \ldots, \tilde m\}$, and assume without loss of generality that $m\ge \td m$. Given the corresponding bipartite graph $\cw$, for each injective mapping $\td \pi:\td\cs_l\rightarrow \cs_l$,
we can define an associated matching: $${\cm}_{\td\pi}({\cw})=\{(u, \td v): u=\td\pi(\td v)\hbox{ and } (u, \td v)\in{\cw}\}.$$
It is easy to see that the MaxMatching must be the one of them because any matching can be viewed as a partial injection from $\td \cs_l$ to $\cs_l$.

Now given $\cs_l=\{1, 2, \cdots, m\}$, $\td\cs_l=\{1, 2, \cdots, \tilde m\},$ and a fixed injection $\td\pi$, we regard the bipartite graph $\cw$ as random, and bound
$\bq\left(
        |{\cm}_{\td\pi}(\cw)|\ge \frac{1}{2C_2}(\lambda s)^l\right).$
Note that events $\{(u, \td\pi^{-1}(\td u))\in {\cw}\}$ and $\{(v, \td\pi^{-1}(\td v))\in {\cw}\}$ are i.i.d. given $u\not=v$ and $\td\pi$ is an injection, where $\tilde \pi^{-1}$ is the inverse of $\td \pi$. Furthermore, $\{(u, \td\pi^{-1}(\td u))\in {\cw}\}$ if and only if $([\ch_{d,u}],[\td\ch_{d, \td\pi^{-1}(\td u)}])\in\cf_{d-l},$ and, by \eqref{eq:L2_Type1} in Proposition~\ref{lem:matching}, we have $\bq\bigl(([\ch_{d,u}],[\td\ch_{d, \td\pi^{-1}(\td u)}])\in\cf_{d-l}\bigr)\leq C_3 e^{-d+l}$ for all sufficiently large $n$, where $C_3$ is a constant independent with $n$. Therefore,
\begin{align}
    |{\cm}_{\td\pi}({\cw})|\stackrel{\text{sto.}}{\leq}\mathrm{Bin}\bigl(\td m,C_3e^{-d+l}\bigr).
    \nonumber
\end{align}
By Chernoff bound \eqref{eq:CB_Upper} in Lemma \ref{lem:CB},
\begin{align}
    &\bq\left(
        |{\cm}_{\td\pi}({\cw})|\ge \frac{1}{2C_2}(\lambda s)^l \,\middle|\,  \max\{|{\cs}_l|,|{\td\cs}_l|\} \leq (C_2\lambda)^l\right)\nonumber \\ 
        \le{}& \mathsf P\left( \mathrm{Bin}\bigl((C_2\lambda)^l,C_3e^{-d+l}\bigr)\ge \frac{1}{2C_2}(\lambda s)^l \right) \nonumber\\
    \le {}&  \exp\left( \frac{1}{2C_2}(\lambda s)^l \Bigl[ -d+l(1+\log C_2-\log s)+\log(2eC_2C_3) \Bigr] \right) \nonumber\\
    \le{} & \exp\left(  -\frac{1}{4C_2}(\lambda s)^ld \right),
    \label{eq:typeI-fixed-cn}
\end{align}
where the last inequality holds since
\begin{align}
    -\frac{d}{2} \ge -d+l(1+\log C_2-\log s)+\log(2eC_2C_3) \nonumber
\end{align}
for all sufficiently large $n$.
Applying the union bound on all possible injections $\td\pi$, we have 
\begin{align}
    &\bq\left(
        \texttt{MaxMatching}(\cw)\ge \frac{1}{2C_2}(\lambda s)^l \,\middle|\, \max\{|{\cs}_l|,|{\td\cs}_l|\} \leq (C_2\lambda)^l \right)\nonumber\\
    \le{}& \bigl((C_2\lambda)^l\bigr)! \cdot
     \exp\left( -\frac{1}{4C_2}(\lambda s)^ld \right) \nonumber\\
    \le{}& \exp\left( (\lambda s)^l \left[ 2l\left(\frac{C_2}{s}\right)^l\log \lambda  -\frac{d}{4C_2}\right]\right).
    \label{eq:s3curub}
\end{align}
where \eqref{eq:s3curub} holds because $((C_2\lambda)^l)!\leq \exp{(l(C_2\lambda)^l\log(C_2\lambda))}$ and $\log(C_2\lambda)\le 2\log \lambda$ for large enough $n$.
Now define $\ca_5$ to be the event
\begin{align}
    \ca_5 := \left\{ \max\{\lvert\cs_l\rvert,\lvert\td\cs_l\rvert\}\le (C_2\lambda)^l
    \right\}. \label{eq:def-A5}
\end{align}
By \eqref{eq:gen_LB} in Lemma \ref{lem:s3bos},
\begin{align}
    \bq(\ca_5^c)
    &\le  \P\left(\lvert\cs_l\rvert\ge (C_2\lambda)^l\right) +\P\left(\lvert\td\cs_l\rvert\ge (C_2\lambda)^l\right) \nonumber\\
    &\le 2\exp\left( -\tfrac13\lambda C_2^l+\lambda(e^{2/3}-1) \right).\label{eq:ca5-prob}
\end{align}
Combining~\eqref{eq:ca5-prob} and~\eqref{eq:s3curub}, we conclude
\begin{align}
&    {\bq}\bigl(\Phi([\ch_d],[\td\ch_d])=1\bigr)=\bq\left(
        \texttt{MaxMatching}(\cw)\ge \frac{1}{2C_2}(\lambda s)^l\right)\nonumber\\
    \le& \bq\left(
        \texttt{MaxMatching}(\cw)\ge \frac{1}{2C_2}(\lambda s)^l
        \,\middle|\,\ca_5\right)
    + \bq(\ca_5^c)\nonumber\\
    \le & \exp\left( (\lambda s)^l \left[ 2l\left(\frac{C_2}{s}\right)^l\log \lambda  -\frac{d}{4C_2}\right]\right)+2\exp\left( -\tfrac13\lambda C_2^l+\lambda(e^{2/3}-1) \right).\label{eq:two_terms}
\end{align}

To further bound~\eqref{eq:two_terms}, we choose constants $\beta$ and $C_2$ as follows. First, let $\eta=\frac12(\gamma+\alpha-1)$. By our assumptions that $\alpha\in (0,1)$ and $\gamma\in (1-\alpha,1)$, we have $\eta\in (0,\gamma+\alpha-1)$. We then choose 
\[
\beta=\begin{cases}
    -\frac{\eta}{\log s}\,\quad &\text{  when }\,s<1\\
    \eta\,&\text{  when }\, s=1
\end{cases},
\]
and
\[
C_2=\exp\left(\frac{\gamma-\alpha-\eta+1}{2\beta}\right).
\]
The choice of $\beta$ and $C_2$ satisfies the following three properties that are crucial for establishing the desired bound:
\begin{itemize}
    \item \textbf{Property 1: $C_2>1$.} 
    
    To see this, notice that $\beta>0$ by its definition, and we have
    \[
    \gamma-\alpha-\eta+1=\frac{\gamma-3\alpha+3}{2}>0,
    \]
    where the last inequality follows because $\alpha<1$.
    \item \textbf{Property 2: $\beta\log\frac{C_2}{s}<\gamma$.} 
    
    To see this, we separately consider the cases of $s<1$ and $s=1$. First, assume $s<1$. Then we have
    \[
    \beta\log\frac{C_2}{s}=\frac{\gamma-\alpha-\eta+1}{2}+\eta=\frac{\gamma-\alpha+\eta+1}{2}<\gamma,
    \]
    where the last equality follows because $\eta<\gamma+\alpha-1$. Now assume $s=1$. Then we have
    \[
    \beta\log\frac{C_2}{s}=\beta\log C_2=\frac{\gamma-\alpha-\eta+1}{2}<\gamma.
    \]
    \item \textbf{Property 3: $\alpha+\beta\log C_2>1$.} 
    
    This property follows because
    \[
    \alpha+\beta\log C_2=\alpha+\frac{\gamma-\alpha-\eta+1}{2}=\frac{\gamma+\alpha-\eta+1}{2}>1,
    \]
    where the last inequality follows because $\eta<\gamma+\alpha-1$.
\end{itemize}
We now use these properties to bound the two terms in~\eqref{eq:two_terms}. First, we have
\[
2l\left(\frac{C_2}{s}\right)^l\log \lambda-\frac{d}{4C_2}=\exp{\Bigl((1+o(1))\beta\log\frac{C_2}{s}\log\log n\Bigr)}-\frac{1}{4C_2}\exp{(\gamma \log\log n)}.
\]
Property 2 implies that the negative term on the right-hand side of the above equation dominates. Then we know that, for all large enough $n$,
\begin{align}
    \exp\left( (\lambda s)^l \left[ 2l\left(\frac{C_2}{s}\right)^l\log \lambda  -\frac{d}{4C_2}\right]\right)&\le \exp\left(-\frac{(\lambda s)^ld}{8C_2}\right)=n^{-\frac{(\lambda s)^ld}{8C_2\log n}}\le \frac12 n^{-\log\log n},
    \label{eq:first-term}
\end{align}
where the last inequality follows because $(\lambda s)^l=\exp((\alpha\beta+o(1))(\log\log n)^2)=\omega(\log n\log\log n)$.

Second, we can write
\begin{align*}
    &2\exp\left( -\tfrac13\lambda C_2^l+\lambda(e^{2/3}-1) \right)\nonumber\\
    ={}&2\exp\left( -\tfrac13 (\log n)^{\alpha+\beta \log C_2+o(1)}+(e^{2/3}-1)(\log n)^{\alpha +o(1)} \right)\nonumber\\
    ={}&2\exp\left(-\log n\left(\tfrac13 (\log n)^{\alpha+\beta \log C_2-1+o(1)}-(e^{2/3}-1)(\log n)^{\alpha -1+o(1)}\right)\right)\nonumber.
\end{align*}
Property 3 implies that
\[
\tfrac13 (\log n)^{\alpha+\beta \log C_2-1+o(1)}-(e^{2/3}-1)(\log n)^{\alpha-1 +o(1)}=\omega(\log\log n). 
\]
It then follows that 
\begin{equation}
2\exp\left( -\tfrac13\lambda C_2^l+\lambda(e^{2/3}-1) \right)\le \frac12n^{-\log\log n},
\label{eq:2nd-term}
\end{equation}
for all large enough $n$.

Substituting~\eqref{eq:first-term} and~\eqref{eq:2nd-term} into~\eqref{eq:two_terms} yields that for all large enough $n$, 
\[
{\bq}\bigl(\Phi([\ch_d],[\td\ch_d])=1\bigr)\le n^{-\log\log n},
\]
which completes the proof of~\eqref{eq:CorTest-typeI}.

\underline{\textbf{Part 2: Proof of \eqref{eq:CorTest-power}}.} 
Suppose $(\ch,\td\ch)\sim\bp$, and let $\ch^*$ be the intersection tree in the construction of the correlated trees. We write $\cs_l^*$ for the set of vertices in the $l$-th generation of $\ch^*$. We now show that for the same constant $\beta>0$ and $C_2>1$ specified in the proof of \eqref{eq:CorTest-typeI}, the power of the test $\Phi$ defined in \eqref{eq:Phi_def} satisfies \eqref{eq:CorTest-power}.

For each $u^*\in\cs_l^*$, let $u$ and $\td u$ denote the corresponding nodes in $\ch_d$ and $\td\ch_d$ at level $l$, respectively. If $([\ch_{d,u}],[\td\ch_{d,\td u}])\in\cf_{d-l}$, then $(u,\td u)\in \ce(\cw)$ by the definition of $\cw$. Therefore, the set of pairs $(u,\td u)$ such that $([\ch_{d,u}],[\td\ch_{d,\td u}])\in\cf_{d-l}$ forms a matching of $\cw$, and we have
\[
\texttt{MaxMatching}(\cw)\geq \sum_{u^*\in \cs^*_l}\mathbbm{1}\{([\ch_{d,u}],[\td\ch_{d,\td u}])\in\cf_{d-l}\}.
\]
Therefore,
\begin{align}
    \bp\left(
\texttt{MaxMatching}(\cw)
\le \frac{1}{2C_2}(\lambda s)^{l}
\right)
\le 
\bp\left(
\sum_{u^*\in \cs^*_l}\mathbbm{1}\{([\ch_{d,u}],[\td\ch_{d,\td u}])\in\cf_{d-l}\} \le \frac{1}{2C_2}(\lambda s)^{l}
\right).\label{eq:s3_bd1}
\end{align}
Next, we define a typical event
\begin{align*}
    \mathcal{A}_4 
&:= \Bigl\{ |\cs_l^*| \ge C_2^{-1}(\lambda s)^{l} \Bigr\}.
\end{align*}
The probability $\bp(\ca_4)$ can be bounded by Lemma \ref{lem:s3bos}. First, since $\ch^*$ is a Galton--Watson tree with offspring distribution $\mathrm{Poi}(\lambda s)$, we have
\begin{align}
    l=\beta\log\log n\le \frac{s}{2}(\log n)^{\alpha+o(1)}=\frac{\lambda s}{2}
    \nonumber
\end{align}
for all sufficiently large $n$, which satisfies the condition of Lemma \ref{lem:s3bos}. Therefore by \eqref{eq:gen_UB},
\begin{align}
    \mathsf{P}(\ca_4)\geq  1-\frac{C_2^2}{(C_2-1)^2}\cdot\frac{1}{\lambda s-1}. \label{eq:ca4-prob}
\end{align}
Next, since conditioned on $|\cs^*_l|$, the subtree pairs $\{([\ch_{d,u}],[\td\ch_{d,\td u}]):u^*\in\cs_l^*\}$ are independent and identically distributed with the correlated model $\bp$. Therefore,
\begin{align}
    \sum_{u^*\in \cs^*_l}\mathbbm{1}\{([\ch_{d,u}],[\td\ch_{d,\td u}])\in\cf_{d-l}\}\sim \mathrm{Bin}\left(
|\cs^*_l|,\, \bp\bigl(([\ch_{d-l}],[\td\ch_{d-l}])\in\cf_{d-l}\bigr)
\right). \nonumber
\end{align}

Conditioned on event $\ca_4$, we have $\lvert\cs^*_l\rvert\ge C_2^{-1}(\lambda s)^{l}$. Moreover, by~\eqref{eq:L2_Power} in Proposition~\ref{lem:matching}, we have $\bp\bigl(([\ch_{d-l}],[\td\ch_{d-l}])\in\cf_{d-l}\bigr)\geq  1-C_4 \lambda^{-1}$ for all sufficiently large $n$, where $C_4$ is a constant independent with $n$. Therefore,
\begin{align}
    \sum_{u^*\in \cs^*_l}\mathbbm{1}\{([\ch_{d,u}],[\td\ch_{d,\td u}])\in\cf_{d-l}\}\stackrel{\text{sto.}}{\geq}\mathrm{Bin}\left(\frac{1}{C_2}(\lambda s)^{l},\, 1-C_4\lambda^{-1}\right).\label{eq:s3_power_bd2}
\end{align}
By Chernoff bound \eqref{eq:CB_Lower} in Lemma \ref{lem:CB}, we have
\begin{align} 
&\bp\left(
\sum_{u^*\in \cs^*_l}\mathbbm{1}\{([\ch_{d,u}],[\td\ch_{d,\td u}])\in\cf_{d-l}\} \le \frac{1}{2C_2}(\lambda s)^{l}\,\middle|\, \ca_4
\right)\nonumber\\
\leq{} & \P\left( \textrm{Bin}\left(
\frac{1}{C_2}(\lambda s)^{l},\, 1-C_4\lambda^{-1}
\right)\le \frac{1}{2C_2}(\lambda s)^{l}
\right)\nonumber\\
\leq{} &
\exp\left(
-\frac{1}{8C_2}(\lambda s)^{l}
\cdot
\frac{(1-2C_4\lambda^{-1})^2}{1-C_4\lambda^{-1}}
\right)\nonumber\\
={} & O\left(n^{-\log\log n}\right),\label{s3_bd2}
\end{align}
where the last equality follows because $(\lambda s)^l=\omega(\log n\log\log n)$.
Therefore, combining \eqref{eq:ca4-prob} and \eqref{s3_bd2} gives
\begin{align*}
    &\bp\left(
\texttt{MaxMatching}(\cw)
\geq \frac{1}{2C_2}(\lambda s)^{l}
\right) \\\geq{}& 1- \bp\left(
|\{u^*\in \cs^*_l: ([\ch_{d,u}],[\td\ch_{d,\td u}])\in\cf_{d-l}\}| \le \frac{1}{2C_2}(\lambda s)^{l}\,\middle|\, \ca_4
\right)- \bp(\ca_4^c)\\
\geq{} & 1- O\left(n^{-\log\log n}\right)-\frac{C_2^2}{(C_2-1)^2}\cdot\frac{1}{\lambda s-1}\\
={}& 1-O\left(\frac{1}{(\log n)^{\alpha/2}}\right),
\end{align*}
which completes the proof of \eqref{eq:CorTest-power}.
\end{proof}

\subsection{Proof of Proposition \ref{lem:matching}}
\label{appd:fd-l}
Recall the set $\cf_{d-l}$ and the constants $C_1,\bar d,\bar\lambda$ from the construction in Appendix~\ref{appd:construction}. We prove this proposition by induction that for each $k\in \{\bar d,\ldots,d-l-1\}$, the set $\cf_{k+1}$ constructed in Appendix~\ref{appd:construction} satisfies
\begin{equation*}
\bq\bigl(([\ch_{k+1}], [\td\ch_{k+1}]) \in \cf_{k+1}\bigr)\le e^{\bar d-k-1}C_1^{-1},
\end{equation*}
and
\[
\bp\bigl(([\ch_{k+1}], [\td\ch_{k+1}]) \in \cf_{k+1}\bigr)
    \ge \max\left\{ 1 - \frac{10}{\lambda s} - \frac{50}{s^2} \bq\bigl(([\ch_k], [\td\ch_k]) \in \cf_k\bigr), 0.4 \right\}.
\]
First consider the base case with $k=\bar d$. Recall that $\lambda=(\log n)^{\alpha+o(1)}$. Suppose $n$ is large enough so that $\lambda\ge \max\left\{\bar\lambda,\frac{15^4}{C_{\mathrm{Otter}}^2}\right\}$, then it follows by Lemma~\ref{prop:initialization} in Appendix~\ref{appd:init} that
\[
\bq\bigl(([\ch_{\bar d}], [\td\ch_{\bar d}]) \in \cf_{\bar d}\bigr) \le C_1^{-1},
\]
and
\[
\bp\bigl(([\ch_{\bar d}], [\td\ch_{\bar d}]) \in \cf_{\bar d}\bigr) \ge 0.4.
\]
This shows the initial set $\cf_{\bar d}$ satisfies condition~\eqref{eq:typei-k} and~\eqref{eq:power-k} of Lemma~\ref{prop:iteration} in Appendix~\ref{appd:boost}, so we can apply the lemma to obtain 
\begin{equation*}
\bq\bigl(([\ch_{\bar d+1}], [\td\ch_{\bar d+1}]) \in \cf_{\bar d+1}\bigr)
    \le e^{-1}C_1^{-1},
\end{equation*}
and
\[
\bp\bigl(([\ch_{\bar d+1}], [\td\ch_{\bar d+1}]) \in \cf_{\bar d+1}\bigr)
    \ge \max\left\{ 1 - \frac{10}{\lambda s} - \frac{50}{s^2} \bq\bigl(([\ch_{\bar d}], [\td\ch_{\bar d}]) \in \cf_{\bar d}\bigr), 0.4 \right\},
\]
which completes the proof for the base case.

Now suppose the induction hypothesis holds for the set $\cf_k$.
We can again apply Lemma~\ref{prop:iteration} to obtain
\begin{align*}
    \bq\bigl(([\ch_{k+1}], [\td\ch_{k+1}]) \in \cf_{k+1}\bigr)
    \le e^{\bar d-k-1}C_1^{-1},
\end{align*}
and 
\[
\bp\bigl(([\ch_{k+1}], [\td\ch_{k+1}]) \in \cf_{k+1}\bigr)
    \ge
    \max\left\{
        1 - \frac{10}{\lambda s} - \frac{50}{s^2} \bq\bigl(([\ch_k], [\td\ch_k]) \in \cf_k\bigr), 0.4\right\},
\]
which completes the proof of the induction step.

By choosing $C_3:=C_1^{-1}e^{\bar d}$, we have
\[
\bq\bigl(([\ch_{d-l}], [\td\ch_{d-l}]) \in \cf_{d-l}\bigr)
    \le
    C_1^{-1} e^{-(d-l-\bar d)}
    =
    C_3 e^{-(d-l)},
\]
which completes the proof of~\eqref{eq:L2_Type1}.

On the other hand, we have
\begin{align*}
\bp\bigl(([\ch_{d-l}], [\td\ch_{d-l}]) \in \cf_{d-l}\bigr)
    &\ge
    1 - \frac{10}{\lambda s}
    - \frac{50}{s^2} \bq\bigl(([\ch_{d-l-1}], [\td\ch_{d-l-1}]) \in \cf_{d-l-1}\bigr)\\
    &\ge     1 - \frac{10}{\lambda s}
    - \frac{50}{s^2}e^{-(d-1-\bar d-l)}C_1^{-1}.
\end{align*}
Suppose $n$ is large enough so that $e^{-(d-l)}\le \lambda^{-1}$. By choosing $C_4=\frac{10}{s}+\frac{50C_1^{-1}}{s^2}e^{\bar d+1}$, we have
\[
\bp\bigl(([\ch_{d-l}], [\td\ch_{d-l}]) \in \cf_{d-l}\bigr)\ge 1-C_4\lambda^{-1},
\]
which completes the proof of~\eqref{eq:L2_Power}.

\subsection{Guarantee for the initialization}
\label{appd:init}
In this and the following subsections, we establish the two key lemmas used to prove Proposition \ref{lem:matching}. We use the notation $\mu$ instead of $\lambda$ to denote the mean of the offspring in the independent and correlated Galton--Watson tree distributions $\bq$ and $\bp$. This is to emphasize that the results in these two subsections hold for an arbitrary offspring mean $\mu$, and are not tied to the specific choice $\lambda=(\log n)^{\alpha+o(1)}$ used in the graph model considered in this paper.

\begin{lem}\label{prop:initialization}
Assume $s\in(\sqrt{C_{\mathrm{Otter}}},1]$. Let $C_1>1$ be a constant.
Let $\ch$ and $\td\ch$ denote two \GW trees with $\mathrm{Poi}(\mu)$ offsprings. 
Then there exist constants $\bar d\in\mathbb{N}$ and $\bar\lambda>1$, depending only on $s$ and $C_1$, such that for all $\mu\ge \bar\lambda$,
\begin{align}
    \bq\bigl(L_{\bar d}([\ch_{\bar d}],[\td\ch_{\bar d}])\ge C_1\bigr)
    &\le C_1^{-1},\label{eq:init-typeI}\\
    \bp\bigl(L_{\bar d}([\ch_{\bar d}],[\td\ch_{\bar d}])\ge C_1\bigr)
    &\ge 0.4 .\label{eq:init-power}
\end{align}
\end{lem}
\begin{proof}
We first prove the bound~\eqref{eq:init-typeI}.
Let $k\ge 1$ be an integer. For two trees $(t,\td t)\in\cz_k\times\cz_k$, the likelihood ratio $L_k(t,\td t)$ is defined as
\[
L_k(t,\td t)=\frac{\bp([\ch_k]=t,[\td \ch_k]=\td t)}{\bq ([\ch_k]=t,[\td\ch_k]=\td t)},
\]
and we have $$\E_\bq\left[L_k([\ch_{k}],[\td\ch_{k}])\right]=1.$$
By Markov's inequality, for any $C_1>0$,
\begin{align}
    \bq\bigl(L_k([\ch_k],[\td\ch_k])\ge C_1\bigr) \le C_1^{-1}\E_\bq\left[L_k([\ch_k],[\td\ch_k])\right]
    =C_1^{-1}.
\end{align}
Therefore, the desired bound~\eqref{eq:init-typeI} holds for any $\bar d\ge 1$ and $\mu>1$.

Next, we prove \eqref{eq:init-power} by considering two separate cases: $s\in(\sqrt{C_{\mathrm{Otter}}},1)$ and $s=1$.

\textbf{\underline{Case $1$: $s\in\left(\sqrt{C_{\mathrm{Otter}}},1\right)$}.}
To find constants $\bar d$ and $\bar \lambda$ such that~\eqref{eq:init-power} holds, we first establish the following lower bound for general $k\ge 1$ and $\mu>1$:
\begin{align}
    \bp(L_k([\ch_{k}],[\td\ch_{k}])\ge C_1)\ge
    \frac{\be_\bp[\log L_k([\ch_k],[\td\ch_k])]-C_1^{-1}-\log C_1}{\log\be_\bp[ L_k([\ch_k],[\td\ch_k])]}.
    \label{eq:power-basic-bound}
\end{align}

Since $L_k(t,\td t)>0$ for any $(t,\td t)\in \cz_k\times\cz_k$, we have
\begin{align}
    \log L_k(t,\td t)\le \log L_k(t,\td t)\indi\{L_k(t,\td t)\ge 1\}= \int_1^\infty\frac{1}{u}\indi\left\{
        L_k(t,\td t)\ge u\right\} du,\label{eq:kl-trunc}
\end{align}
where the inequality in \eqref{eq:kl-trunc} holds since $\log L_k(t,\td t)$ is negative when $0<L_k(t,\td t)<1$, and the equality in \eqref{eq:kl-trunc} holds since if $L_k(t,\td t)<1$, both sides are zero; and if $L_k(t,\td t)\ge 1$, the right-hand side equals to $\int_1^{L_k(t,\td t)} u^{-1}du=\log L_k(t,\td t)$. Taking expectation over $([\ch_k],[\td\ch_k])$ under $\bp$, we have
\begin{align}
    \be_\bp[\log L_k([\ch_k],[\td\ch_k])]\le& 
    \be_\bp\left[\int_1^\infty\frac{1}{u}\indi(L_k([\ch_k],[\td\ch_k])\ge u)\,du\right]\\
    =& \int_1^\infty\frac{1}{u}\bp(L_k([\ch_k],[\td\ch_k])\ge u)\,du.
    \label{eq:kl-integral}
\end{align}
where the last equality holds using Tonelli's theorem, which applies because the integrand is nonnegative.
By Lemma~\ref{lem:init-first-moment}, we know that $\be_\bp \left[L_k([\ch_k],[\td\ch_k]) \right]> 1$. We now split the integral on the right-hand side of \eqref{eq:kl-integral} into intervals
\[
[1, C_1),\quad 
[C_1, C_1\be_\bp [L_k([\ch_k],[\td\ch_k])])\quad\hbox{and}\quad 
[C_1\be_\bp [L_k([\ch_k],[\td\ch_k]) ],\infty),
\]
and have
\begin{align}
    &\int_1^\infty\frac{1}{u}\bp(L_k([\ch_k],[\td\ch_k])\ge u)\,du\nonumber\\
    ={}&
    \int_1^{C_1}\frac{1}{u}\bp(L_k([\ch_k],[\td\ch_k])\ge u)\,du+\int_{C_1}^{C_1\be_\bp [L_k([\ch_k],[\td\ch_k]) ]}\frac{1}{u}\bp(L_k([\ch_k],[\td\ch_k])\ge u)\,du
    \nonumber\\
    &+\int_{C_1\be_\bp [L_k([\ch_k],[\td\ch_k]) ]}^\infty\frac{1}{u}\bp(L_k([\ch_k],[\td\ch_k])\ge u)\,du,.
    \label{eq:split}
\end{align}
Next, we bound these three terms separately. For the first term, since $\bp(L_k([\ch_k],[\td\ch_k])\geq u)\leq 1$, 
\begin{equation}
    \int_{1}^{C_1}\frac{1}{u}\bp\left(L_k([\ch_k],[\td\ch_k])\ge u\right) du\le \int_{1}^{C_1}\frac{1}{u}du=\log C_1.\label{eq:first-range-bound}
\end{equation}
For the second term, since $\bp(L_k([\ch_k],[\td\ch_k])\geq u)$ is non-increasing with $u$,
\begin{align}
    &\int_{C_1}^{C_1\be_\bp [L_k([\ch_k],[\td\ch_k]) ]}\frac{1}{u}\bp\left(L_k([\ch_k],[\td\ch_k])\ge u\right) du\nonumber\\
    \le{}& \bp \left(L_k([\ch_k],[\td\ch_k])\ge C_1\right)
    \int_{C_1}^{C_1\be_\bp [L_k([\ch_k],[\td\ch_k]) ]}\frac{1}{u}\,du \nonumber\\
    ={}&\bp\left(L_k([\ch_k],[\td\ch_k])\ge C_1\right)\log\be_\bp [L_k([\ch_k],[\td\ch_k]) ].
    \label{eq:second-range-bound}
\end{align}
For the third term, using $u^{-1}\le (C_1\be_\bp[L_k([\ch_k],[\td\ch_k]) ])^{-1}$, we have
\begin{align}
    &\int_{C_1\be_\bp[L_k([\ch_k],[\td\ch_k]) ]}^\infty\frac{1}{u}\bp(L_k([\ch_k],[\td\ch_k])\ge u)\,du\nonumber\\
    \le{}&
    (C_1\be_\bp[L_k([\ch_k],[\td\ch_k]) ])^{-1}\int_{C_1\be_\bp[L_k([\ch_k],[\td\ch_k]) ]}^\infty\bp(L_k([\ch_k],[\td\ch_k])\ge u)\,du\nonumber\\
    \stackrel{(a)}{=}{}&(C_1\be_\bp[L_k([\ch_k],[\td\ch_k]) ])^{-1}\be_\bp\left[\max\{(L_k([\ch_k],[\td\ch_k])-C_1\be_\bp[L_k([\ch_k],[\td\ch_k]) ]),0\}\right]\nonumber\\
    \le{}&
    (C_1\be_\bp[L_k([\ch_k],[\td\ch_k]) ])^{-1}\be_\bp[L_k([\ch_k],[\td\ch_k]) ]=C_1^{-1},
    \label{eq:third-range-bound}
\end{align}
where (a) follows because for any real valued random variable $X$ and real number $a$, $\be[\max\{X-a,0\}]=\int_a^\infty\P(X\ge u)du$.

Substituting~\eqref{eq:first-range-bound},~\eqref{eq:second-range-bound} and~\eqref{eq:third-range-bound} into~\eqref{eq:split} yields 
\begin{align}
    &\int_1^\infty\frac{1}{u}\bp(L_k([\ch_k],[\td\ch_k])\ge u)\,du\nonumber\\
    \le{}& \log C_1+\bp\left(L_k([\ch_k],[\td\ch_k])\ge C_1\right)\log\be_\bp [L_k([\ch_k],[\td\ch_k]) ]+C_1^{-1},\label{eq:integral}
\end{align}
and the bound~\eqref{eq:power-basic-bound} follows by~\eqref{eq:kl-integral} and~\eqref{eq:integral}.

With~\eqref{eq:power-basic-bound}, we move on to choose $\bar d$ and $\bar \lambda$ so that~\eqref{eq:init-power} holds. Towards this goal, define
\begin{equation}
    \chi_k:=\frac{1}{2}\log\sum_{t\in\cz_k}s^{2|\cv(t)|-2},
    \label{eq:chisq-def}
\end{equation}
for each $k\ge 1$. It follows by Lemmas~\ref{lem:init-first-moment} and~\ref{lem:init-kl-lb} in Appendix~\ref{appd:lr_moments} that
\begin{equation}
    \be_\bp\left[L_k([\ch_k],[\td\ch_k])
    \right] =\exp(2\chi_k)
    \label{eq:EP}
\end{equation}
and
\begin{equation}
    \liminf_{\mu\to\infty}
    \be_\bp\left[\log L_k([\ch_k],[\td\ch_k])\right]\ge\chi_k.
\end{equation}

Now, we choose $\bar d$ so that it satisfies
\begin{align}
\frac{0.99\chi_{\bar d}-C_1^{-1}-\log C_1}{2\chi_{\bar d}}\ge 0.4. \label{eq:bard-choice}
\end{align}    
The existence of such $\bar d$ is guaranteed by Lemma~\ref{lem:init-chi-prop} and the fact that~\eqref{eq:bard-choice} holds when $\chi_{\bar d}\ge 0.19(C_1^{-1}+\log C_1)$. It also follows from Lemma~\ref{lem:init-chi-prop} that $\chi_{\bar d}$ is finite.

With this choice of $\bar d$ in hand, we further choose $\bar \lambda$ such that for any $\mu\ge \bar\lambda$, 
\begin{equation}
    \label{eq:barlambda}
    \be_\bp\left[\log L_{\bar d}([\ch_{\bar d}],[\td\ch_{\bar d}])\right]\ge0.99\chi_{\bar d}.
\end{equation}

Plugging~\eqref{eq:EP} and~\eqref{eq:barlambda} into~\eqref{eq:power-basic-bound} yields
\begin{align*}
    \bp(L_{\bar d}([\ch_{{\bar d}}],[\td\ch_{{\bar d}}])\ge C_1)&\ge \frac{0.99\chi_{\bar d}-C_1^{-1}-\log C_1}{2\chi_{\bar d}}\ge 0.4,
\end{align*}
which completes the proof in this case.

\textbf{\underline{Case $2$: $s=1$}.} Let $\bar d$ be the same as the one chosen in the previous case. Since $s=1$, the two trees satisfy $[\ch_{\bar d}]=[\td\ch_{\bar d}]$
almost surely under $\bp$, and for each rooted unlabeled tree pairs $( t ,\td t )\in\cz_{\bar d}$, 
\begin{align}
    L_{\bar d}( t ,\td t )
    =\begin{cases}
        \frac{\bp\bigl([\ch_{\bar d}]= t ,[\td\ch_{\bar d}]= t \bigr)}
         {\bq\bigl([\ch_{\bar d}]= t ,[\td\ch_{\bar d}]= t \bigr)}=\frac{1}{\bp([\ch_{\bar d}]= t)},\quad& \,\text{ if }\,  t =\td t ,\\
         0, & \,\text{otherwise}\,.
    \end{cases}\\
    \label{eq:L1-star}
\end{align}
Moreover, for the case of  $ t =\td t $, let $c$ denote the root degree of $t$. Then we have 
\begin{align}
    L_{\bar d}( t , t )=\frac{1}{\bp([\ch_{\bar d}]= t)}\ge\frac{1}{\P(\poi(\mu)=c)}\geq \frac{1}{\max_{c'\ge0}\P(\poi(\mu)=c')}\geq \sqrt{\pi}\mu^{1/2}.
\end{align}
where the last inequality holds for any $\mu>1$ by \cite[Eq.~(4.46)]{johnson1992univariate}. Let $\bar \lambda=C_1^2>1$, we have $L_{\bar d}( t , t )\ge C_1$ for any $ t \in\cz_{\bar d}$
given that $\mu\ge\bar \lambda$.  Consequently,
\[
    \bp\bigl(L([\ch_{\bar d}],[\td\ch_{\bar d}])\ge C_1\bigr)=1,
\]
which proves~\eqref{eq:init-power} at $s=1$.

\end{proof}

\subsection{Guarantee for the iterative boosting}
\label{appd:boost}
We next establish the guarantee for the iterative boosting, which propagates the type-I error and power bounds from depth $k$ to depth $k+1$.

\begin{lem}[Iterative boosting]\label{prop:iteration}
Let $C_1=2\times 10^6$, and $\bar d$ and $\bar \lambda$ be chosen as in Lemma~\ref{prop:initialization}. Assume  $s \in (\sqrt{C_{\mathrm{Otter}}}, 1]$, $k\ge \bar d$ and $\mu\ge\max\left\{\bar\lambda,\frac{15^4}{C_{\mathrm{Otter}}^2}\right\}$. Suppose $\cf_{k}$ is a subset of $\cz_k\times \cz_k$ such that
\begin{align}
    \bq\bigl(([\ch_{k}], [\td\ch_{k}]) \in \cf_{k}\bigr)
    &\le e^{\bar d-k} C_1^{-1},\label{eq:typei-k}\\
    \bp\bigl(([\ch_{k}], [\td\ch_{k}]) \in \cf_k\bigr)
    &\ge 0.4.\label{eq:power-k}
\end{align}
Define the set
\begin{equation}
\cf_{k+1}:=\left\{(t, \td t) \in \cz_{k+1} \times \cz_{k+1} :\sum_{(h, \td h) \in \cf_k}\bigl(N_h(t) - \mu p_h\bigr)\bigl(N_{\td h}(\td t) - \mu p_{\td h}\bigr)\geq \sigma_k\right\},\label{eq:def_iter_2}
\end{equation}
where $\sigma_k=\max\left\{\frac1{10}\mu e^{\frac{1}{4}(\bar{d}-k)}, 3\mu^{3/4}\right\}$, $p_{h}=\mathrm{GW}([\ch_k]=h)$, and $N_h(t)$ is the number of children of the root of $t$ whose descendant subtree is isomorphic to $h$; $p_{\td h}$ and $N_{\td h}(\td t)$ are defined analogously. Then we have
\begin{align}
    \bq\bigl(([\ch_{k+1}], [\td\ch_{k+1}]) \in \cf_{k+1}\bigr)
    &\le e^{\bar d-k-1}C_1^{-1},\label{eq:iter-typeI}\\
    \bp\bigl(([\ch_{k+1}], [\td\ch_{k+1}]) \in \cf_{k+1}\bigr)
    &\ge\max\left\{  1 - \frac{10}{\mu s} - \frac{50}{s^2} \bq\bigl(([\ch_k], [\td\ch_k]) \in \cf_k\bigr),\,0.4
    \right\}.\label{eq:iter-power}
\end{align}
\end{lem}

\begin{proof}
    Throughout the proof, for each $(t, \td t) \in \cz_{k+1} \times \cz_{k+1}$, we define $X_{k+1}(t,\td t)$ by
    \begin{align}
        X_{k+1}(t,\td t):=\sum_{(h, \td h) \in \cf_k}\bigl(N_h(t) - \mu p_{h}\bigr)\bigl(N_{\td h}(\td t) - \mu p_{\td h}\bigr),\label{eq:X_def}
    \end{align}
    then by \eqref{eq:def_iter_2}, $\cf_{k+1}=\{(t,\td t)\in\cz_{k+1}\times\cz_{k+1}: X_{k+1}(t,\td t)\geq\sigma_k\}$.
    We now prove \eqref{eq:iter-typeI} and \eqref{eq:iter-power} separately.
    
    \textbf{\underline{Proof of \eqref{eq:iter-typeI}}.} Suppose $(\ch,\td\ch)\sim \bq$. By Markov's inequality on the non-negative random variable $X_{k+1}^4([\ch_{k+1}],[\td\ch_{k+1}])$, we have
\begin{align}
\bq \bigl(X_{k+1}([\ch_{k+1}],[\td\ch_{k+1}])\ge \sigma_k\bigr)&\le\bq \bigl(X^4_{k+1}([\ch_{k+1}],[\td\ch_{k+1}])\ge \sigma_k^4\bigr)\nonumber\\
&\le \frac{1}{\sigma_k^{4}}
\be_{\bq} \bigl[ X_{k+1}([\ch_{k+1}],[\td\ch_{k+1}])^{4} \bigr]\nonumber\\
&\le \frac{1}{\sigma_k^{4}}
\Bigl(36\mu^4\bq(([\ch_k], [\td\ch_k]) \in \cf_k)^{2}
+13 \mu^{3} \bq(([\ch_k], [\td\ch_k]) \in \cf_k)\Bigr), \label{eq:s2qb1}
\end{align}
where the last inequality holds by~\eqref{eq:q5mom} in Lemma~\ref{lem:p1mom}.
By the choice of $\sigma_k$, we have
\begin{align}
    \sigma^4_k=\max\left\{\frac{1}{10^4}\mu^4e^{\bar d-k},81\mu^{3}\right\}\geq \max\{72e\mu^4\bq\bigl(([\ch_{k}],[\td\ch_{k}])\in \cf_{k}\bigr),26e\mu^3\},\label{eq:sigma4_def}
\end{align}
where the last inequality follows because $26e<81$ and the fact that 
\[
72e\bq\bigl(([\ch_{k}],[\td\ch_{k}])\in\cf_{k}\bigr)\le 72eC_1^{-1}e^{\bar d-k}<\frac{1}{10^4}e^{\bar d-k}
\]
under assumption~\eqref{eq:typei-k}.
Substituting~\eqref{eq:sigma4_def} into~\eqref{eq:s2qb1} yields
\begin{align*}
    &\bq \bigl(X_{k+1}([\ch_{k+1}],[\td\ch_{k+1}])\ge \sigma_k\bigr)\nonumber\\
    \le{}& \frac{36\mu^4\bq\bigl(([\ch_k], [\td\ch_k]) \in \cf_k\bigr)^{2}
+13 \mu^{3} \bq\bigl(([\ch_k], [\td\ch_k]) \in \cf_k\bigr)}{\max\{72e\mu^4\bq\bigl(([\ch_{k}],[\td\ch_{k}])\in \cf_{k}\bigr),26e\mu^3\}}\\
\le{}& \frac{2\max\{36\mu^4\bq\bigl(([\ch_k], [\td\ch_k]) \in \cf_k\bigr)^{2}
,13 \mu^{3} \bq\bigl(([\ch_k], [\td\ch_k]) \in \cf_k\bigr)\}}{\max\{72e\mu^4\bq\bigl(([\ch_{k}],[\td\ch_{k}])\in \cf_{k}\bigr),26e\mu^3\}}\\
\le{}& \frac1e \bq\bigl(([\ch_{k}],[\td\ch_{k}])\in \cf_{k}\bigr)\\
\le{}& e^{\bar d-k-1}C_1^{-1}
\end{align*}
which completes the proof of~\eqref{eq:iter-typeI}.

\textbf{\underline{Proof of \eqref{eq:iter-power}}.}
By assumption~\eqref{eq:power-k}, we have
\begin{equation}
\frac12 s\bp\bigl(([\ch_k],[\td\ch_k])\in \cf_k\bigr)> 0.2\sqrt{C_\mathrm{Otter}}>0.1.
\label{eq:sigmak1}
\end{equation}
Because $\mu\ge\frac{15^4}{C_\mathrm{Otter}^2}$, we have
\begin{equation}
\frac{\frac12\mu s\bp\bigl(([\ch_k],[\td\ch_k])\in \cf_k\bigr)}{3\mu^{3/4}}=\frac{\mu^{1/4}s\bp\bigl(([\ch_k],[\td\ch_k])\in \cf_k\bigr)}{6}\ge 1.\label{eq:sigmak2}
\end{equation}
Recall that $\sigma_k=\max\left\{\frac1{10}\mu e^{-\frac{1}{4}(k-\bar{d})}, 3\mu^{3/4}\right\}$. By~\eqref{eq:sigmak1} and~\eqref{eq:sigmak2}, we have
\begin{align}
\sigma_{k}\leq\max\left\{\frac1{10}\mu,3\mu^{3/4}\right\} \le \frac12\mu s\bp\bigl(([\ch_k],[\td\ch_k])\in \cf_k\bigr)=\frac12 \be_{\bp}\bigl[X_{k+1}([\ch_{k+1}],[\td\ch_{k+1}])\bigr],\label{eq:s2sigma}
\end{align}
where the last equality holds by \eqref{eq:p1mom} in Lemma \ref{lem:p1mom}.

By Chebyshev's inequality,
\begin{align}
&\bp\bigl(X_{k+1}([\ch_{k+1}],[\td\ch_{k+1}])< \sigma_k \bigr)\nonumber\\
\le{}& \bp\Bigl(
X_{k+1}([\ch_{k+1}],[\td\ch_{k+1}]) \le \tfrac12 \be_{\bp}\bigl[ X_{k+1}([\ch_{k+1}],[\td\ch_{k+1}]) \bigr] \Bigr) \nonumber\\
\le{}& \bp\Bigl( \bigl| X_{k+1}([\ch_{k+1}],[\td\ch_{k+1}]) - \be_{\bp}\bigl[ X_{k+1}([\ch_{k+1}],[\td\ch_{k+1}]) \bigr] \bigr| \ge\tfrac12\be_{\bp}\bigl[X_{k+1}([\ch_{k+1}],[\td\ch_{k+1}])\bigr]\Bigr)\nonumber\\
\le{}& \frac{ 4\operatorname{Var}_{\bp}\bigl( X_{k+1}([\ch_{k+1}],[\td\ch_{k+1}])\bigr)}{\bigl(\be_{\bp}\bigl[X_{k+1}([\ch_{k+1}],[\td\ch_{k+1}])\bigr]\bigr)^2}.
\label{eq:s2pb1}
\end{align}

By equations~\eqref{eq:p1mom} and~\eqref{eq:p2mom} in Lemma~\ref{lem:p1mom},
we obtain
\begin{align}
&\bp\bigl(X_{k+1}([\ch_{k+1}],[\td\ch_{k+1}])< \sigma_k \bigr)\nonumber\\
\le{}& 4\frac{\mu s \bp\bigl(([\ch_k],[\td\ch_k])\in \cf_k\bigr)+\mu^2(1+s^2)\bq\bigl(([\ch_k],[\td\ch_k])\in \cf_k\bigr)}{\mu^2s^2\bigl(\bp\bigl(([\ch_k],[\td\ch_k])\in \cf_k\bigr)\bigr)^2}\nonumber\\
\le{}&\frac{4}{\mu s \bp\bigl(([\ch_k],[\td\ch_k])\in \cf_k\bigr)}+\frac{8\bq\bigl(([\ch_k],[\td\ch_k])\in \cf_k\bigr)}{s^2\bigl(\bp\bigl(([\ch_k],[\td\ch_k])\in \cf_k\bigr)\bigr)^2}\nonumber\\
\le{}&\frac{10}{\mu s}+\frac{50}{s^2}\bq\bigl(([\ch_k],[\td\ch_k])\in \cf_k\bigr),\label{eq::pf2power}
\end{align}
where the last inequality follows by $\bp\bigl(([\ch_k],[\td\ch_k])\in \cf_k\bigr)\ge 0.4$.

Since $\mu\geq \frac{15^4}{C_{\mathrm{Otter}}^2}$ and $C_1=2\times 10^6> \frac{50}{0.6s^2-\frac{10sC_\mathrm{Otter}^2}{15^4}}$, we have 
\begin{align}
    \frac{10}{\mu s}+\frac{50}{s^2}\bq\bigl(([\ch_k],[\td\ch_k])\in \cf_k\bigr)\leq \frac{10}{\mu s}+\frac{50}{s^2}C_1^{-1}< \frac{10C_\mathrm{Otter}^2}{15^4s}+0.6-\frac{10C_\mathrm{Otter}^2}{15^4s}= 0.6\label{eq:pf2power2}
\end{align}
where the first inequality holds by \eqref{eq:typei-k}. Combining \eqref{eq::pf2power} and \eqref{eq:pf2power2}, we have
\begin{align}
\bp\bigl(([\ch_{k+1}], [\td\ch_{k+1}]) \in \cf_{k+1}\bigr)
    &=\bp\bigl(X_{k+1}([\ch_{k+1}],[\td\ch_{k+1}])\ge \sigma_k
\bigr)\nonumber\\
    &=1-\bp\bigl(X_{k+1}([\ch_{k+1}],[\td\ch_{k+1}])<\sigma_k
\bigr)\nonumber\\
    &\ge\max\left\{ 1 - \frac{10}{\mu s} - \frac{50}{s^2}  \bq\bigl(([\ch_k], [\td\ch_k]) \in \cf_k\bigr),\,0.4 \right\},\nonumber
\end{align}
which completes the proof of \eqref{eq:iter-power}.
\end{proof}

\subsection{Proof of Theorem~\ref{prop:LRT}: Neyman--Pearson Lemma}
\label{appd:pf-lrt}

We now complete the proof of Theorem~\ref{prop:LRT} based on Proposition~\ref{prop:CorTest}. By the Neyman--Pearson lemma \cite[Theorem~3.2.1]{lehmann2005testing}, Proposition \ref{prop:CorTest} implies that there exists a likelihood-ratio test with likelihood-ratio threshold $\theta^*,$ tie-breaking threshold $\kappa^*$ and {randomization} probability $p$ such that it is a most powerful test with  type-I error being equal to ${\bq}\bigl(\Phi([\ch_d],[\td\ch_d])=1\bigr)$, and only makes a randomized decision on a single tree pair $(t',\td t')$, which satisfies $L(t',\td t')=\theta^*$ and $f(t',\td t')=\kappa^*$. Here $p$ is the probability that the test outputs $1$ when the input is $(t',\td t')$.

Notice that the only difference between the test $\Psi$ in~\eqref{eq:rank-LRT} and this most powerful test is that $\Psi$ deterministically outputs $0$ at the input $(t',\td t')$ instead of giving a randomized output.
Therefore, we have 
\begin{align}
    {\bq}\bigl(\Psi([\ch_d],[\td\ch_d])=1\bigr)
    =&\ {\bq}\bigl(\Phi([\ch_d],[\td\ch_d])=1\bigr)-p{\bq}\bigl([\ch_d]=t',[\td\ch_d]=\td t'\bigr)\nonumber\\
    \le &\ {\bq}\bigl(\Phi([\ch_d],[\td\ch_d])=1\bigr)\nonumber\\
    =&\ O(n^{-\log\log n}),\label{eq:rank_LRT_Type1_Proof}
\end{align}
and
\begin{align}
    {\bp}\bigl(\Psi([\ch_d],[\td\ch_d])=1\bigr)
    \ge
    {\bp}\bigl(\Phi([\ch_d],[\td\ch_d])=1\bigr)-p{\bp}\bigl([\ch_d]=t',[\td\ch_d]=\td t'\bigr).
    \nonumber
\end{align}
Given that $(\ch,\td\ch)\sim \bp$, $\ch$ follows a \GW process with $\poi(\lambda)$ offspring distribution. Then for any fixed $(t,\td t)\in \cz_d\times\cz_d$,

\[
{\bp}\bigl([\ch_d]=t,[\td\ch_d]=\td t)\bigr)\le {\bp}\bigl([\ch_d]=t\bigr)\le \mathsf{P}(\poi(\lambda)=c_r),
\]
where $c_r$ denotes the root degree of $t$. Therefore, under the assumption that $\lambda=(\log n)^{\alpha+o(1)}$, we have 
\begin{align}
    \sup_{(t,\td t)\in\cz\times\cz}{\bp}\bigl([\ch_d]=t,[\td\ch_d]=\td t\bigr)\le \sup_{c\ge 0} \mathsf{P}(\poi(\lambda)=c)=O(\lambda^{-1/2}) = O\bigl((\log n)^{-(\alpha+o(1))/2}\bigr),
    \nonumber
\end{align}
where the $O(\lambda^{-1/2})$ bound for the maximum point probability of the $\mathrm{Poi}(\lambda)$ distribution is given by \cite[Eq.~(4.46)]{johnson1992univariate}. Combining the bounds above with \eqref{eq:CorTest-power} in Proposition \ref{prop:CorTest}, we obtain
\begin{align}
    {\bp}\bigl(\Psi([\ch_d],[\td\ch_d])=1\bigr)
    &\ge
    {\bp}\bigl(\Phi([\ch_d],[\td\ch_d])=1\bigr)
    -O\bigl((\log n)^{-\alpha/2}\bigr)
    \nonumber\\
    &=
    1-O\bigl((\log n)^{-\alpha/2}\bigr)-O\bigl((\log n)^{-(\alpha+o(1))/2}\bigr)\nonumber\\
    &=
    1-O\bigl((\log n)^{-\alpha/4}\bigr),
    \nonumber
\end{align}
which completes the proof of Theorem \ref{prop:LRT}.

\section{Technical lemmas for Poisson Galton--Watson trees}\label{appd:gw-lemmas}
\subsection{Counting of unlabeled trees}
\label{appd:counting}
\begin{lem}[Asymptotics of the number of unlabeled trees~\cite{Richard48TreeCount}]
\label{lem:otter}
One has 
\[
\lim_{k\rightarrow\infty}\frac{\bigl|\bigl\{ t \in\cz:| \cv(t)|=k\bigr\}\bigr|}{\frac{C_5}{k^{3/2}}C_\mathrm{Otter}^{-k}}=1,
\]
for some constant $C_5>0$, where $C_\mathrm{Otter}\approx0.338$ is Otter's tree counting constant. In particular, there exists a constant $k_0>0$ such that for all $k\ge k_0$,
\[
\bigl|\bigl\{ t \in\cz:| \cv(t)|=k\bigr\}\bigr|\ge \frac{C_5}{2k^{3/2}}C_{\mathrm{Otter}}^{-k}.
\]
\end{lem}

\subsection{Tree generation size bounds}
\label{appd:generation_size}
\begin{lem}\label{lem:s3bos}
Let $\mu>1$, and let $\ch$ be a Poisson Galton--Watson tree with offspring mean $\mu$. 
Let $\cs_l(\ch)$ be the set of nodes in the $l$-th generation of $\ch$ for $l \leq \frac{\mu}{2}$ and $l\in\mathbb{N}$.
Then for any constant $C_2 > 1$,
\begin{align}
\P\bigl(\lvert \cs_{l}(\ch) \rvert \ge (C_2 \mu)^l \bigr)
&\le \exp\left(-\tfrac{1}{3}\mu C_2^l + \mu (e^{2/3}-1)\right),\label{eq:gen_LB} \\
\P\bigl(\lvert \cs_{l}(\ch) \rvert\le C_2^{-1}\mu^l \bigr)
&\le \frac{C_2^2}{(C_2-1)^2}\cdot\frac{1}{\mu-1}.\label{eq:gen_UB}
\end{align}
\end{lem}
\begin{proof}
First, for any fixed $k\in[l]$, each node in $\cs_{k-1}(\ch)$ has an independent $\mathrm{Poi}(\mu)$ number of children. Hence, conditioned on $\lvert \cs_{k-1}(\ch) \rvert$, by the Poisson summation property, 
\begin{align}
    \lvert \cs_{k}(\ch) \rvert\sim \mathrm{Poi}(\lvert \cs_{k-1}(\ch) \rvert\mu).\label{eq:gen_size_diistribution}
\end{align}
Define the log-moment generating function $
\varphi_k(t) := \log \E\left[e^{t \lvert \cs_{k}(\ch) \rvert}\right]$ for all $k\in[l]$.
Since $|\cs_{0}(\ch)|=1$ and $|\cs_{1}(\ch)|\sim \mathrm{Poi}(\mu)$, we have
\begin{align}
    \varphi_0(t) =& \log \E[e^{t |\cs_{0}(\ch)|}] = t, \nonumber\\ \varphi_1(t) =& \log \E[e^{t |\cs_{1}(\ch)|}]
= \mu(e^t - 1)
= \mu\bigl(e^{\varphi_0(t)} - 1\bigr),\nonumber
\end{align}
and for general $1 \le k\le l$, the tower property and \eqref{eq:gen_size_diistribution} give
\begin{align}
\varphi_k(t)
&= \log \E\left[e^{t \lvert \cs_{k}(\ch) \rvert}\right]
= \log \E \left[\E \left[e^{t \lvert \cs_{k}(\ch) \rvert}\,\middle|\, \lvert \cs_{k-1}(\ch) \rvert\right]\right] \nonumber\\
&= \log \E \left[\exp \bigl(\mu\cdot \lvert \cs_{k-1}(\ch) \rvert\cdot(e^t-1)\bigr)\right] \nonumber\\
&= \varphi_{k-1} \bigl(\mu(e^t-1)\bigr).\label{eq:s3iterpsi}
\end{align}
Applying \eqref{eq:s3iterpsi} for all $k=l,\cdots,2$, we have
\begin{align}
    \varphi_l(t)=\varphi_{l-1} \bigl(\mu(e^t-1)\bigr)=\varphi_{l-1}\circ\varphi_1(t)=\cdots=\underbrace{\varphi_1\circ\varphi_1\circ\cdots\circ\varphi_1}_{l\ \mathrm{times}}(t).
    \label{eq:s3iterpsi3}
\end{align}
where $\circ$ denotes the composition of functions, i.e., $\varphi_{l-1}\circ\varphi_1(t)=\varphi_{l-1}(\varphi_1(t))$. Thus, the last expression means that $\varphi_1$ is
applied to $t$ repeatedly $l$ times. Applying the same
recursive argument to $\varphi_{l-1}$ yields $\varphi_{l-1}=\underbrace{\varphi_1\circ\varphi_1\circ\cdots\circ\varphi_1}_{(l-1)\ \mathrm{times}}.$ Hence, separating the leftmost $\varphi_1$ in \eqref{eq:s3iterpsi3}, we obtain
\begin{align}
    \varphi_l(t)=\varphi_1\bigl(\varphi_{l-1}(t)\bigr)
    = \mu\bigl(e^{\varphi_{l-1}(t)}-1\bigr).
    \label{eq:s3iterpsi2}
\end{align}

By the Chernoff bound, for any $t>0$,
\begin{align}
    \P\left(\lvert \cs_{l}(\ch) \rvert \ge (C_2 \mu)^l\right)
\le \exp\left(-t (C_2 \mu)^l + \varphi_l(t)\right)= \exp\left(-t (C_2 \mu)^l + \mu\bigl(e^{\varphi_{l-1}(t)} - 1\bigr)\right).\label{eq:child_size_Chernoff}
\end{align}
By choosing  $t = \frac{1}{3\mu^{l-1}}$, we have 
\begin{align}
    \P\left(\lvert \cs_{l}(\ch) \rvert \ge (C_2 \mu)^l\right)
\le & \exp \left(-\frac{1}{3\mu^{l-1}} (C_2 \mu)^l + \mu\bigl(e^{\varphi_{l-1}\left(\frac{1}{3\mu^{l-1}}\right)} - 1\bigr)\right)\nonumber\\
\le & \exp \left(-\frac{1}{3}\mu C_2^{l}
+ \mu(e^{\frac23}-1)\right),\label{mgb}
\end{align} 
which proves \eqref{eq:gen_LB}.
We note that \eqref{mgb} holds because  we have 
\begin{equation}
\varphi_{l-1}\left(\frac{1}{3\mu^{l-1}}\right)
\le \frac{1}{3} + \frac{l-1}{3\mu}\leq \frac{2}{3}\label{eq:logMGF_Iteration}
\end{equation} by following the proof below. 

\underline{\textbf{Proof of~\eqref{eq:logMGF_Iteration}:}} We prove \eqref{eq:logMGF_Iteration} by induction on $k =0 ,\ldots, l-1$. For the base case $k=0$, we have $|\cs_0(\ch)|=1$, so $\varphi_0\Bigl(\frac{1}{3\mu^{l-1}}\Bigr) = \frac{1}{3\mu^{l-1}}$, i.e., \eqref{eq:logMGF_Iteration} holds. Next, assume the hypothesis holds for $k-1$, such that
    \begin{align}
        \varphi_{k-1}\left(\frac{1}{3\mu^{l-1}}\right)
\le \frac{1}{3\mu^{l-k}} + \frac{k-1}{3\mu^{l-k+1}},\label{eq:s3indpsi}
    \end{align}
First, the right-hand side of \eqref{eq:s3indpsi} is smaller than $1$ for all $k \le l-1$, since $\frac{1}{3\mu^{l-k}} \le \frac13$ and $\frac{k-1}{3\mu^{l-k+1}} \le \frac{l-2}{3\mu^2} <\frac23$. Then by \eqref{eq:s3iterpsi}, we have
    \begin{align}
        \varphi_k\Bigl(\frac{1}{3\mu^{l-1}}\Bigr)=\mu\left(e^{\varphi_{k-1}\left(\frac{1}{3\mu^{l-1}}\right)}-1\right)
        \leq  \mu \varphi_{k-1}\left(\frac{1}{3\mu^{l-1}}\right)+\mu \left(\varphi_{k-1}\left(\frac{1}{3\mu^{l-1}}\right)\right)^2\nonumber
    \end{align}
    where the inequality holds by noticing that $e^x-1\leq x+x^2$ for any $x\in(0,1)$ and by letting $x:=\varphi_{k-1}\left(\frac{1}{3\mu^{l-1}}\right)<1$. By \eqref{eq:s3indpsi} and noticing $\mu\geq 2l\geq 2(k+1)$, we have
    \begin{align}
        \varphi_k\Bigl(\frac{1}{3\mu^{l-1}}\Bigr)
        \leq & \frac{1}{3\mu^{l-k-1}}+\frac{k-1}{3\mu^{l-k}}+\frac{1}{9\mu^{2l-2k-1}}+\frac{(k-1)^2}{9\mu^{2l-2k+1}}+\frac{2(k-1)}{9\mu^{2l-2k}}\nonumber\\
        \leq & \frac{1}{3\mu^{l-k-1}}+\frac{k-1}{3\mu^{l-k}}+\frac{1}{9\mu^{2l-2k-1}}+\frac{1}{36\mu^{2l-2k-1}}+\frac{1}{9\mu^{2l-2k-1}}\nonumber\\
        \leq & \frac{1}{3\mu^{l-k-1}}+\frac{k}{3\mu^{l-k}},\nonumber
    \end{align}
    where the last inequality holds since $2l-2k-1\geq l-k$ for all $k=0,\cdots,l-1$. By the induction, \eqref{eq:logMGF_Iteration} is proved.

Next, since $\ch$ is a Poisson Galton--Watson tree with offspring mean $\mu$, we have $\E[\lvert \cs_{l}(\ch) \rvert] = \mu^l$ and $\mathrm{Var}(|\cs_0(\ch)|)=0$. Moreover, since for each $k\in[l]$,
\begin{align*}
\mathrm{Var}(\lvert \cs_{k}(\ch) \rvert)
= &\E\left[\mathrm{Var}(|\cs_k(\ch)| \,\bigm|\,  \lvert \cs_{k-1}(\ch) \rvert)\right]
   + \mathrm{Var}(\E\left[\lvert \cs_{k}(\ch) \rvert\,\bigm|\, \lvert \cs_{k-1}(\ch) \rvert\right]) \\
=& \mu \E\left[\lvert \cs_{k-1}(\ch) \rvert\right] + \mu^2 \mathrm{Var}(\lvert \cs_{k-1}(\ch) \rvert),
\end{align*}
induction on $k\in[l]$ gives $\mathrm{Var}(\lvert \cs_{l}(\ch) \rvert)
= \mu^l \cdot \frac{\mu^l-1}{\mu-1}.$ Hence, by Chebyshev's inequality,
\begin{align}
\P(\lvert \cs_{l}(\ch) \rvert \le C_2^{-1}\mu^l)
&= \P\left(\lvert \cs_{l}(\ch) \rvert \le C_2^{-1}\E[\lvert \cs_{l}(\ch) \rvert]\right)\nonumber \\
&\le \P \left(|\lvert \cs_{l}(\ch) \rvert-\E[\lvert \cs_{l}(\ch) \rvert]|
\ge (1-C_2^{-1})\E[\lvert \cs_{l}(\ch) \rvert]\right)\nonumber \\
&\le \frac{\mathrm{Var}(\lvert \cs_{l}(\ch) \rvert)}
{(1-C_2^{-1})^2(\E[\lvert \cs_{l}(\ch) \rvert])^2}\nonumber \\
&\leq \frac{1}{(1-C_2^{-1})^2}
\cdot \frac{1}{\mu-1}.\nonumber
\end{align}
This completes the proof of \eqref{eq:gen_UB}.
\end{proof}

\subsection{Moments of the likelihood ratio}
\label{appd:lr_moments}
The proof of Lemma~\ref{prop:initialization} uses the following lemmas.
Recall that for each $k\ge 1$, we defined
\begin{equation}
\chi_k:=\frac{1}{2}\log\sum_{t\in\cz_k}s^{2|\cv(t)|-2}.\label{eq:defn-chik}
\end{equation}

\begin{lem}[Properties of $\chi_k$]
\label{lem:init-chi-prop}
Assume $s\in(0,1)$. For any $k\ge 1$, $\chi_k$ satisfies the following three properties:
\begin{enumerate}
    \item $\chi_k$ is strictly positive;
    \item $\chi_k$ is finite;
    \item If we further assume that $s^2>C_\mathrm{Otter}$, then for any constant $C_6>0$, there exists an integer $\bar d$ such that $\chi_{\bar d}>C_6$. In particular, $\bar d$ is given by
    \[
    \bar d=\left\lceil\max\left\{k_0-1,\left(\frac{3}{\log\frac{s^2}{C_\mathrm{Otter}}}\right)^2,\frac{2}{\log\frac{s^2}{C_\mathrm{Otter}}}\left(2C_6 + \log\frac{2s^2}{C_5}\right)\right\}\right\rceil,
    \]
    where $k_0$ and $C_5$ are constants indicated by Lemma~\ref{lem:otter}.
\end{enumerate}
\end{lem}

\begin{proof} 

\textbf{\underline{Proof of Property 1:}}
For every $k\geq 1$, the set $\cz_k$ at least contains both the single-vertex tree and the tree consisting of a root with one single child. Consequently, the definition of $\chi_k$ gives
\begin{align}
    \chi_k
    =\frac12\log\sum_{t\in\cz_k}s^{2|\cv(t)|-2}
    \geq \frac12\log(1+s^2)>0.\nonumber
\end{align}

\textbf{\underline{Proof of Property 2:}}
We prove by induction that $\chi_k$ is finite for every finite $k$. For the base case $k=1$, every tree in $\cz_1$ is determined by the number of children of its root. If the root has $r$ children, then the tree has
$r+1$ vertices. Hence
\begin{align}
    \exp(2\chi_1)
    =\sum_{r=0}^{\infty}s^{2r}
    =\frac{1}{1-s^2},\nonumber
\end{align}
which is finite for every $s\in(0,1)$.

Now suppose by induction hypothesis that $\chi_{k-1}$ is finite for some $k\geq 2$, and we prove that $\chi_k$ is also finite. This induction step utilizes the following variational characterization of $\chi_k$, which is first proven in~\cite{Richard48TreeCount}:
\begin{align}
    \chi_k=\frac12\sum_{h \in \cz_{k-1}}
\log\Bigl(\frac{1}{1-s^{2|\cv(h)|}}\Bigr).\label{eq:chi-alt}
\end{align}
With \eqref{eq:chi-alt}, we have
\begin{align}
    \chi_k&\leq
    \frac12\sum_{h\in\cz_{k-1}}\frac{s^{2|\cv(h)|}}{1-s^{2|\cv(h)|}}
    \le  \frac12\sum_{h\in\cz_{k-1}}\frac{s^{2|\cv(h)|}}{1-s^{2}}\nonumber\\
    &\leq \frac{s^2}{2(1-s^2)}\sum_{h\in\cz_{k-1}}s^{2|\cv(h)|-2}
    =\frac{s^2}{2(1-s^2)}\exp(2\chi_{k-1}),\label{eq:chi-rec-bound}
\end{align}
where the first inequality holds because for any real number $x\in(0,1)$, $\log\left(\frac{1}{1-x}\right)\leq \frac{x}{1-x}$, the second inequality holds because $|\cv(h)|\geq 1$ for every $h\in\cz_{k-1}$, and the equality follows by the definition~\eqref{eq:defn-chik}. Thus, given that $\chi_{k-1}$ is finite, the right-hand side of \eqref{eq:chi-rec-bound} is finite, and hence $\chi_k$ is finite. 

We now prove~\eqref{eq:chi-alt}. First, using the geometric series expansion, for each ($h\in\cz_{k-1}$) we have
\[
(1-s^{2|\cv(h)|})^{-1}=\sum_{l=0}^{\infty}s^{2l|\cv(h)|}.
\]
Therefore, the right-hand side of~\eqref{eq:chi-alt} can be rewritten as
\begin{align}
    \frac12\sum_{h \in \cz_{k-1}}
\log\Bigl(\frac{1}{1-s^{2|\cv(h)|}}\Bigr)
=\frac12\sum_{h \in \cz_{k-1}}\log\left(\sum_{l=0}^{\infty} s^{2l|\cv(h)|}\right)
=\frac12\log\left(\prod_{h \in \cz_{k-1}}\sum_{l=0}^{\infty} s^{2l|\cv(h)|}\right).\label{eq:l9eq1}
\end{align}
We next expand the product inside the logarithm. For each factor indexed by $h\in\cz_{k-1}$, one needs to select one term from the sum $\sum_{l=0}^{\infty} s^{2l|\cv(h)|}.$ Such a selection is equivalently described by a function $\zeta:\cz_{k-1}\to \mathbb {Z}_{\ge 0},$ where $\mathbb{Z}_{\ge 0}:=\{0,1,2,\ldots\}$ and $\zeta(h)$ specifies that the term $s^{2\zeta(h)|\cv(h)|}$ is selected from the factor corresponding to $h$. Since all terms are nonnegative, we may expand the product as
\[
\prod_{h \in \cz_{k-1}}
\left(\sum_{l=0}^{\infty} s^{2l|\cv(h)|}\right)
=\sum_{\zeta:\cz_{k-1} \to\mathbb {Z}_{\ge 0}}s^{2\sum_{h \in \cz_{k-1}} \zeta(h)|\cv(h)|}.
\]

We now show that the functions $\zeta:\cz_{k-1}\to \mathbb Z_{\ge 0}$ are in one-to-one correspondence with the trees in $\cz_k$. Given such a function $\zeta$, construct a rooted tree $t_\zeta$ as follows. Start with a root vertex. For each $h\in\cz_{k-1}$, attach the $\zeta(h)$ children to this root and let the rooted descendant subtree below each of these children be a copy of $h$. Since $h$ has depth at most $k-1$, the resulting tree has depth at most $k$, and hence $t_\zeta\in\cz_k$.

Conversely, given any tree $t\in\cz_k$, each child of the root of $t$ has a rooted descendant subtree of depth at most $k-1$, and hence its isomorphism class belongs to $\cz_{k-1}$. Define $\zeta_t(h)$ to be the number of children of the root of $t$ whose rooted descendant subtree is isomorphic to $h$. Then $\zeta_t$ is a mapping from $\cz_{k-1}$ to $\mathbb{Z}_{\ge 0}$.

These two constructions are inverse to each other. Indeed, starting from a function $\zeta$, constructing the tree $t_\zeta$, and then counting the rooted descendant subtree isomorphism classes below the root recovers the same function $\zeta$. Conversely, starting from a tree $t$, forming the function $\zeta_t$, and then reconstructing a tree from $\zeta_t$ gives back the same rooted tree up to isomorphism. Therefore, there is a one-to-one correspondence between functions $\zeta:\cz_{k-1}\rightarrow\mathbb{Z}_{\ge 0}$ and trees $t\in\cz_{k}$.

Under this correspondence, for each $\zeta$, the number of vertices of $t_\zeta$ satisfies
\[
|\cv(t_\zeta)|=1 + \sum_{h \in \cz_{k-1}} \zeta(h)|\cv(h)|.
\]
Here the extra $1$ accounts for the root, while the summation counts all vertices in the subtrees attached below the root. Therefore, we have
\[
s^{2\sum_{h \in \cz_{k-1}} \zeta(h)|\cv(h)|}=s^{2|\cv(t_\zeta)|-2},
\]
which further implies that 
\[
\prod_{h \in \cz_{k-1}}
\left(\sum_{l=0}^{\infty} s^{2l|\cv(h)|}\right)
=\sum_{\zeta:\cz_{k-1} \to\mathbb {Z}_{\ge 0}}s^{2\sum_{h \in \cz_{k-1}} \zeta(h)|\cv(h)|}=\sum_{t\in \cz_k}s^{2|\cv(t)|-2}.
\]
Substituting this into the right-hand side of~\eqref{eq:l9eq1} yields 
\begin{align*}
\frac12\sum_{h \in \cz_{k-1}}
\log\Bigl(\frac{1}{1-s^{2|\cv(h)|}}\Bigr)=\frac12\log\left(\sum_{t\in \cz_k}s^{2|\cv(t)|-2}\right)=\chi_k,
\end{align*}
which completes the proof of~\eqref{eq:chi-alt}.

\textbf{\underline{Proof of Property 3:}}
Note that every rooted tree with $k$ vertices has depth at most $k-1$, and hence
\[
\{t\in\cz:|\cv(t)|=k\}=
\bigl\{ t \in\cz_{k-1}:| \cv(t)|=k\bigr\}.
\]
By Lemma~\ref{lem:otter}, there exist constants $C_5>0$ and $k_0\in\mathbb{N}$ such that for all $k\ge k_0$,
\begin{align}
    \bigl|\bigl\{ t \in\cz:| \cv(t)|=k\bigr\}\bigr|\;\ge\;
    \frac{C_5}{2k^{3/2}}C_{\mathrm{Otter}}^{-k}.\label{eq:otter_ineq}
\end{align}
Since $s^2>C_{\mathrm{Otter}}$, we have $\log\frac{s^2}{C_\mathrm{Otter}}>0$. We choose
\begin{align}
    \bar d= \left\lceil\max\left\{k_0-1,\left(\frac{3}{\log\frac{s^2}{C_\mathrm{Otter}}}\right)^2,\frac{2}{\log\frac{s^2}{C_\mathrm{Otter}}}\left(2C_6 + \log\frac{2s^2}{C_5}\right)\right\}\right\rceil.\nonumber
\end{align}
Using the definition~\eqref{eq:defn-chik} and restricting the summation to trees with exactly $\bar d+1$ vertices, we obtain
\begin{align}
\chi_{\bar d}=\frac12\log\sum_{ t \in\cz_{\bar d}}s^{2| \cv(t) |-2}\geq \frac12\log\sum_{\substack{ t \in\cz_{\bar d}, \\
     | \cv(t) |=\bar d+1 }}s^{2\bar d}
     =\frac12\log\left(s^{2\bar d}\left|
\left\{t\in\cz_{\bar d}: |\cv(t)|=\bar d+1\right\}\right|\right)\nonumber
\end{align}
Using~\eqref{eq:otter_ineq} with $k=\bar d+1\ge k_0$, we have
\begin{align}
    \chi_{\bar d}&\geq\frac12
    \log
    \left(
        s^{2\bar d}
        \frac{C_5}{2(\bar d+1)^{3/2}}
        C_{\mathrm{Otter}}^{-(\bar d+1)}
    \right)=\frac12 \left((\bar d+1)\log\frac{s^2}{C_\mathrm{Otter}}+\log\frac{C_5}{2s^2}-\frac32\log(\bar d+1) \right).\label{eq:chid_bound}
\end{align}
Note that $\sqrt{x+1}>\log(x+1)$ for any $x>-1$. We have
\[
\log(\bar d+1)< (\bar d+1)^{1/2}< \frac{\log\frac{s^2}{C_\mathrm{Otter}}}{3}(\bar d+1),
\]
where the last inequality follows because $\bar d+1>(\frac{3}{\log\frac{s^2}{C_\mathrm{Otter}}})^2$.
By \eqref{eq:chid_bound}, we have
\begin{align}
\chi_{\bar d}>\frac12 \left(\left(1-\frac32\times\frac13\right)(\bar d+1)\log\frac{s^2}{C_\mathrm{Otter}}+\log\frac{C_5}{2s^2}\right)= \frac14(\bar d+1)\log\frac{s^2}{C_\mathrm{Otter}}+\frac12\log\frac{C_5}{2s^2}>C_6,\nonumber
\end{align}
where the last inequality holds because $\bar d+1>\frac{2}{\log\frac{s^2}{C_\mathrm{Otter}}}\left(2C_6 + \log\frac{2s^2}{C_5}\right)$. This completes the proof.

\end{proof}

\begin{lem}[First moment identity~\texorpdfstring{\cite[Corollary~1]{Luca24Stat}}{}]
\label{lem:init-first-moment}
Fix $s\in(0,1)$ and $k\ge1$.  Let $\ch$ and $\td\ch$ be Poisson
Galton--Watson trees with offspring mean $\mu$, and let $\bp$ denote the
correlated tree-pair law.  Then
\begin{equation}
    \label{eq:init-first-moment-sample}
    \be_\bp\left[L_k([\ch_k],[\td\ch_k])
    \right] =\exp(2\chi_k).
\end{equation}
Since $\chi_k>0$ for any $k\ge 1$, this also implies that $\be_\bp\left[L_k([\ch_k],[\td\ch_k])
    \right]> 1$ for any $k\ge 1$.
\end{lem}

\begin{lem}[Log moment lower bound\texorpdfstring{\cite[Proposition~3.1]{Luca24Stat}}{}]
\label{lem:init-kl-lb}
Fix $s\in(0,1)$ and $k\ge1$.  Let $\ch$ and $\td\ch$ be Poisson
Galton--Watson trees with offspring mean $\mu$, and let $\bp$ denote the
correlated tree-pair law.  Then
\begin{equation}
    \label{eq:init-kl-lb-sample}
    \liminf_{\mu\to\infty}
    \be_\bp\left[\log L_k([\ch_k],[\td\ch_k])\right]\ge\chi_k.
\end{equation}
\end{lem}

\subsection{Moment bounds of the boosting statistics}

\begin{lem}[Moment bounds~\texorpdfstring{\cite[Lemma~B.1]{Luca24Stat}}{}]\label{lem:p1mom}
For any $k\geq 1$ and $\mu>1$, let $\ch$ and $\td\ch$ be two \GW trees with $\mathrm{Poi}(\mu)$ offsprings, and $X_{k+1}(t,\td t)$ is defined in \eqref{eq:X_def}. Then
\begin{align}
    \E_{\bp} \left[ X_{k+1}([\ch_{k+1}],[\td\ch_{k+1}]) \right] &=\mu s \bp(([\ch_k], [\td\ch_k]) \in \cf_k),\label{eq:p1mom} \\
\operatorname{Var}_{\bp} \left(X_{k+1}([\ch_{k+1}],[\td\ch_{k+1}])\right)
&\leq \mu s\bp(([\ch_k], [\td\ch_k]) \in \cf_k)
+ \mu^{2}(1+s^{2})\bq(([\ch_k], [\td\ch_k]) \in \cf_k).\label{eq:p2mom}\\
        \E_{\bq} \left[ X^4_{k+1}([\ch_{k+1}],[\td\ch_{k+1}]) \right]
&\leq 36\mu^4\bigl(\bq(([\ch_k], [\td\ch_k]) \in \cf_k)\bigr)^{2}
+13 \mu^{3} \bq(([\ch_k], [\td\ch_k]) \in \cf_k). \label{eq:q5mom}
    \end{align}
\end{lem}

\section{Additional Lemmas for Proving Proposition \ref{th:TBAlign}}
\label{appd:rare}
\subsection[Event A3u]{Event $\ca_{3,u}$}
\begin{restatable}{lem}{TBAlignIdealEone}
\label{lem:TBAlignIdeal-E1}
Under the assumptions in Proposition~\ref{th:TBAlign}, we have
\begin{align}
    \P(\ca^c_{3,u})
    =
    n^{-1+o(1)}.
    \nonumber
\end{align}
\end{restatable}

\begin{proof}
Recall that we defined 
\[
\ca_{3,u}:=\{\text{the }2(d+1)\text{-hop neighborhood of vertex $u$ in  }\bar{G} \text{ is a tree}\}.
\]
For each $0\le k\le 2(d+1)$, we define $\bar\cs(u,k)$ as the set of vertices at distance $k$ from $u$ in $\bar G$. 
We also define the following three events for each $1\le k\le 2(d+1)$:
\begin{itemize}
    \item $\cb_k$ is the event that there exists a vertex in $\bar\cs(u,k)$ that is adjacent to two distinct vertices in $\bar\cs(u,k-1)$.
    \item $\cc_k$ is the event that there exist two vertices in $\bar\cs(u,k)$ that are adjacent in $\bar G$.
    \item $\cd_k$ is the event that $|\bar\cs(u,k)|\ge (\lambda(4-2s))^k\log n$.
\end{itemize}

We first show that
    \[
        \bigcap_{k=1}^{2(d+1)}(\cb_k^c\cap\cc^c_k)
        \subseteq \ca_{3,u}.
    \]
Indeed, consider the subgraph of $\bar G$ induced by the vertices within distance $2(d+1)$ from $u$. For every $k\ge 1$, each vertex in $\bar \cs(u,k)$ has at least one neighbor in $\bar \cs(u,k-1)$ by the definition of graph distance. On the event $\cb_k^c$, no vertex in $\bar \cs(u,k)$ is adjacent to two distinct vertices in $\bar \cs(u,k-1)$. Hence every vertex in $\bar \cs(u,k)$ has a unique parent in the previous layer. Moreover, on the event $\cc_k^c$, there are no edges between two vertices in the same layer $\bar \cs(u,k)$. Together, these imply that the $2(d+1)$-hop neighborhood of $u$ is a tree. We note that event $\cd_k$ is not necessary for establishing $\ca_{3,u}$, but it controls the size of each layer of the neighborhood. These size bounds will later help bounding the error events.

From the result above, we have
\begin{align}
    \P(\ca_{3,u}^c)&\le \P\left(\bigcup_{k=1}^{2(d+1)}\cb_k\cup\cc_k\cup\cd_k\right)\nonumber\\
    &\le \sum_{k=1}^{2(d+1)}\P\left(\cb_k\cup\cc_k\cup\cd_k\,\bigg|\,\bigcap_{i=1}^{k-1}\cb_i^c\cap\cc^c_i\cap\cd^c_i\right)\label{eq:tree_chain}.
\end{align}
Next, we will focus on bounding each term in this summation. 

Conditioned on event $\cd_{k-1}$, there are at most $(\lambda(4-2s))^{k-1}\log n$ nodes in the set $\bar\cs(u,k-1)$.
This bound is also satisfied at $k-1=0$ as $|\bar\cs(u,0)|=1\le \log n$.
Therefore, there are at most $n(\lambda(4-2s))^{2k-2}\log^2 n$ ways to pick two distinct vertices from $\bar\cs(u,k-1)$ and a vertex that is outside the $(k-1)$-hop neighborhood of $u$. In the union graph $\bar G$, the probability that the two picked vertices both have an edge with the vertex in the unexplored part is $\frac{\lambda^2(2-s)^2}{n^2}$. We can apply a union bound to obtain
\begin{equation}
    \label{eq:bkcond}
    \P\left(\cb_k\,\bigg|\, \bigcap_{i=1}^{k-1}\cb_i^c\cap\cc^c_i\cap\cd^c_i\right)=O\left(\frac{(\lambda(4-2s))^{2k}\log^2 n}{n}\right).
\end{equation}

Given that $|\bar\cs(u,k-1)|\le (\lambda(4-2s))^{k-1}\log n$, we have
\[
|\bar\cs(u,k)|\sto \bin\left(n,\frac{\lambda(2-s)}{n}(\lambda(4-2s))^{k-1}\log n\right).
\]
It follows by the Chernoff bound (equation~\eqref{eq:CB_upper2} in Lemma~\ref{lem:CB}) that
\begin{align}
    \P\left(\cd_k\,\bigg|\, \bigcap_{i=1}^{k-1}\cb_i^c\cap\cc^c_i\cap\cd^c_i\right)&\le \P\left(\bin\left(n,\frac{\lambda(2-s)}{n}(\lambda(4-2s))^{k-1}\log n\right)\ge (\lambda(4-2s))^{k}\log n\right)\nonumber\\
    &=O\left(\exp\left(-\frac{\lambda \log n}{3}\right)\right)=n^{-\omega(1)}.\label{eq:dkcond}
\end{align}

Now further condition on event $\cd_k^c$, we know that there are at most $(\lambda(4-2s))^{2k}\log ^2 n$ ways to choose two distinct vertices from $\bar\cs(u,k)$. By the union bound, we get
\begin{align}
    \P\left(\cc_k\,\bigg|\, \cd_k^c\cap\bigcap_{i=1}^{k-1}\cb_i^c\cap\cc^c_i\cap\cd^c_i\right)=O\left(\frac{\lambda (\lambda(4-2s))^{2k}\log ^2 n}{n}\right).\label{eq:ckcond}
\end{align}
By equations~\eqref{eq:bkcond},~\eqref{eq:dkcond} and~\eqref{eq:ckcond}, we have
\begin{align}
    &\P\left(\cb_k\cup\cc_k\cup\cd_k\,\bigg|\,\bigcap_{i=1}^{k-1}\cb_i^c\cap\cc^c_i\cap\cd^c_i\right)\nonumber\\
    &\le \P\left(\cb_k\,\bigg|\, \bigcap_{i=1}^{k-1}\cb_i^c\cap\cc^c_i\cap\cd^c_i\right)+\P\left(\cd_k\,\bigg|\, \bigcap_{i=1}^{k-1}\cb_i^c\cap\cc^c_i\cap\cd^c_i\right)+\P\left(\cc_k\,\bigg|\, \cd_k^c\cap\bigcap_{i=1}^{k-1}\cb_i^c\cap\cc^c_i\cap\cd^c_i\right)\nonumber\\
    &=O\left(\frac{\lambda (\lambda(4-2s))^{2k}\log ^2 n}{n}\right)\label{eq:cond_chain}.
\end{align}
Substituting~\eqref{eq:cond_chain} into~\eqref{eq:tree_chain} yields
\begin{align*}
    \P(\ca_{3,u}^c)=O\left(\frac{(2d+2)\lambda (\lambda(4-2s))^{4d+4}\log ^2 n}{n}\right)=n^{-1+o(1)},
\end{align*}
where the last equality follows because $(2d+2)\lambda (\lambda(4-2s))^{4d+4}\log ^2 n=n^{o(1)}$ by our assumption that $\lambda=(\log n)^{\alpha+o(1)}$ and $d=(\log n)^\gamma$, where constants $\alpha$ and $\gamma$ are strictly less than $1$. This completes the proof.
\end{proof}

\subsection{Sufficient condition for correct matching}
\begin{lem}
\label{lem:sufficient}
    The event $\ca_{1,u}\cap\ca_{2,u}\cap\ca_{3,u}$ implies that ${\bf M}^\infty_\mathrm{TB}(u,\td u) =1 \text{ and } {{\bf M}^\infty_\mathrm{TB}(u,\td v)}=0$, for all $\td v\neq \td u$.
\end{lem}

\begin{proof}
    In this proof, we separately show that $\ca_{1,u}\cap \ca_{3,u}$ implies ${\bf M}_\mathrm{TB}^\infty(u,\td v)=0$ for all $\td v\neq \td u$ and $\ca_{2,u}$ implies ${\bf M}^\infty_\mathrm{TB}(u,\td u) =1$.
    
    \underline{\bf Conditioned on $\ca_{1,u}\cap \ca_{3,u}$:} We first prove using contradiction that when event $\ca_{1,u}\cap \ca_{3,u}$ occurs, ${\bf M}_\mathrm{TB}^\infty(u,\td v)=0$ for all $\td v\neq \td u$. Suppose that there exists some $\td v\neq \td u$ such that
\begin{align}
    {\bf M}_\mathrm{TB}^\infty(u,\td v)=\texttt{GreedyMatching}({\cw}^{\infty}(u,\td v),\theta^*,3)=1.
    \nonumber
\end{align}
Then there exist {\em at least} three distinct neighbors $u_1,u_2,u_3$ of $u$ and three distinct neighbors $\td v_1,\td v_2,\td v_3$ of $\td v$ such that
\begin{align}
    \Psi([{\ct}^{\infty}_{u_i\setminus u}],[{\td\ct}^{\infty}_{\td v_i\setminus \td v}])
     = 1,
    \qquad
    i=1,2,3.
    \label{eq:wrong-three-witnesses}
\end{align}
Conditioned on $\ca_{1,u}$,  the local tree pair $({\ct}^{\infty}_{u_i\setminus u},{\td\ct}^{\infty}_{\td v_i\setminus \td v})$ has nonempty intersection in $\bar G$. In other words, there exists
\begin{align}
    \bar{x}_i\in
    \bar{\cv}({\ct}^{\infty}_{u_i\setminus u})
    \cap
    \bar{\cv}({\td\ct}^{\infty}_{\td v_i\setminus \td v}).
    \nonumber
\end{align}
By definition of the local trees, we have the following statements:
\begin{enumerate}
    \item there is a path $(y_i^{(0)},y_i^{(1)},\ldots,y_i^{(t)})$ in $G$ from $u$ to $x_i$ with $y_i^{(0)}=u$, $y_i^{(1)}=u_i$, $y_i^{(t)}=x_i$, and $t\le d+1$,
    \item there is a path $(\td z_i^{(0)},\td z_i^{(1)},\ldots,\td z_i^{(\td t)})$ in $\td G$ from $\td v$ to $x_i$ with $\td z_i^{(0)}=\td v$, $\td z_i^{(1)}=\td v_i$, $\td z_i^{(\td t)}=x_i$, and $\td t\le d+1$.
\end{enumerate}
 Passing to the union graph, we obtain two paths
\begin{align}
    (\bar y_i^{(0)},\bar y_i^{(1)},\ldots,\bar y_i^{(t)})
    \qquad\text{and}\qquad
    (\bar z_i^{(0)},\bar z_i^{(1)},\ldots,\bar z_i^{(\td t)}),
    \nonumber
\end{align}
where $\bar y_i^{(0)}=\bar u$, $\bar y_i^{(1)}=\bar u_i$, $\bar z_i^{(0)}=\bar v$, $\bar z_i^{(1)}=\bar v_i$, and $\bar y_i^{(t)}=\bar z_i^{(\td t)}=\bar x_i$.

We now claim that there exists a path $\bar{\cp}_i=(\bar p_i^{(0)},\bar p_i^{(1)},\ldots,\bar p_i^{(\bar t)})$
from $\bar u$ to $\bar v$ in $\bar G$, with $\bar t\le 2(d+1)$, such that either $\bar p_i^{(1)}=\bar u_i$, or $\bar p_i^{(\bar t-1)}=\bar v_i$. We consider three cases. 
\begin{itemize}
\item {\bf{Case (i):}} $\bar v=\bar y_i^{(\iota)}$ for some $\iota\in\{0,1,\ldots,t\}$. Since $\bar v\neq \bar u$, we have $\iota\neq 0$. We can construct the desired path $\bar{\cp}_i :=(\bar y_i^{(0)},\bar y_i^{(1)},\ldots,\bar y_i^{(\iota)})$. Since $\bar y_i^{(0)}=\bar u$, $\bar y_i^{(1)}=\bar u_i$ and $\bar y_i^{(\iota)}=\bar v$, $\bar{\cp}_i$
is a path from $\bar u$ to $\bar v$ of length $\bar t=\iota\le d+1$, and $\bar p_i^{(1)}=\bar y_i^{(1)}=\bar u_i$.

\item {\bf{Case (ii):}} $\bar u=\bar z_i^{(\iota)}$ for some $\iota\in\{0,1,\ldots,\td t\}$. Since $\bar v\neq \bar u$, we have $\iota\neq 0$. We can construct the desired path $\bar{\cp}_i:=(\bar z_i^{(\iota)},\bar z_i^{(\iota-1)},\ldots,\bar z_i^{(1)},\bar z_i^{(0)})$. Since $\bar z_i^{(\iota)}=\bar u$, $\td z_i^{(1)}=\td v_i$ and $\td z_i^{(0)}=\td v$, $\bar{\cp}_i$
is a path from $\bar u$ to $\bar v$ of length $\bar t=\iota\le d+1$, and $\bar p_i^{(\bar t-1)}=\bar z_i^{(1)}=\bar v_i$.

\item {\bf{Case (iii):}} $\bar v\neq \bar y_i^{(\iota)}$ for all $\iota\in\{0,1,\ldots,t\}$ and $\bar u\neq \bar z_i^{(\iota)}$ for all $\iota\in\{0,1,\ldots,\td t\}$. Since $\bar y_i^{(t)}=\bar z_i^{(\td t)}=\bar x_i$, the two paths intersect in $\bar G$. Let $\bar w_i$ be the first vertex on the path $(\bar y_i^{(0)},\bar y_i^{(1)},\ldots,\bar y_i^{(t)})$ that also lies on $(\bar z_i^{(0)},\bar z_i^{(1)},\ldots,\bar z_i^{(\td t)})$. Then the initial segment from $\bar u$ to $\bar w_i$ along the $\bar y$-path and the initial segment from $\bar v$ to $\bar w_i$ along the $\bar z$-path are disjoint except at $\bar w_i$. Hence we construct the desired path
\begin{align}
    \bar{\cp}_i:=(\bar y_i^{(0)},\bar y_i^{(1)},\ldots,\bar w_i,\ldots,\bar z_i^{(1)},\bar z_i^{(0)}).
    \nonumber
\end{align}
Since $\bar y_i^{(0)}=\bar u$ and $\td z_i^{(0)}=\td v$, $\bar{\cp}_i$ is a path from $\bar u$ to $\bar v$ in $\bar G$. Its length satisfies $\bar t\le t+\td t\le 2(d+1)$, and it has both $\bar p_i^{(1)}=\bar y_i^{(1)}=\bar u_i$ and $\bar p_i^{(\bar t-1)}=\bar z_i^{(1)}=\bar v_i$.
\end{itemize}
Some illustrating examples for all the $3$ cases of constructing path $\bar{\mathcal{P}}_i$ are shown in Figure \ref{fig:3cases}.

\begin{figure}
    \centering 
        \includegraphics[width=0.6\textwidth]{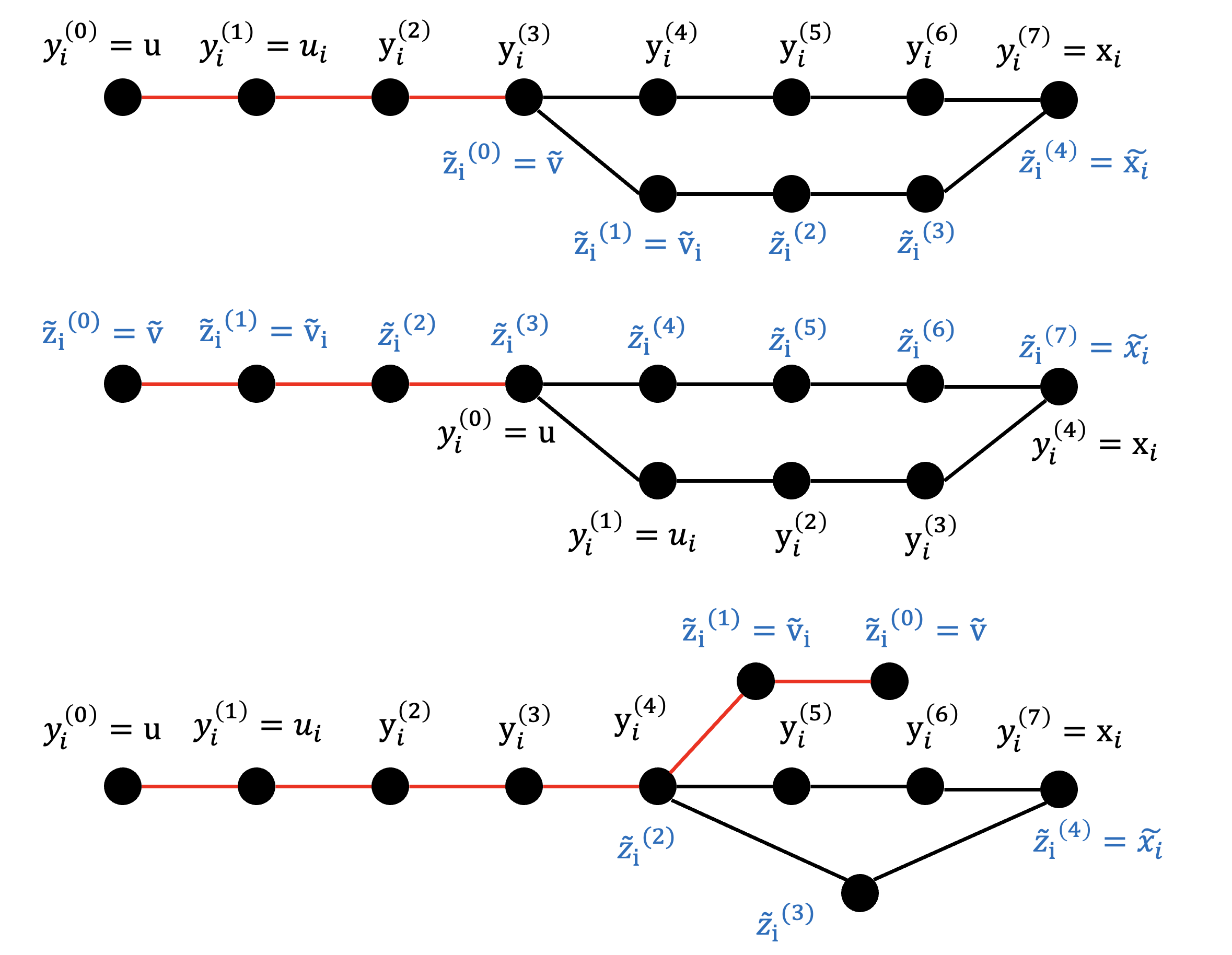}
        \caption{Examples of constructing path $\bar{\mathcal{P}}_i$ in Case (i) (top), Case (ii) (middle) and Case (iii) (bottom). In each case, the red edges denotes path $\bar{\mathcal{P}}_i$. The black notations next to each node represents the nodes of path $(y_i^{(0)},y_i^{(1)},\ldots,y_i^{(t)})$ in $G$, while the blue notations next to each node represents the nodes of path $(\td z_i^{(0)},\td z_i^{(1)},\ldots,\td z_i^{(\td t)})$ in $\td G$.} 
        \label{fig:3cases}
\end{figure}

Applying this construction for $i=1,2,3$, we obtain three paths $\bar{\cp}_1,\bar{\cp}_2,\bar{\cp}_3$ from $\bar u$ to $\bar v$ in $\bar G$, each of length at most $2(d+1)$. Since $\bar u_1,\bar u_2,\bar u_3$ are distinct and $\bar v_1,\bar v_2,\bar v_3$ are distinct, at least two of these three paths are non-identical paths from $\bar u$ to $\bar v$ in $\bar G$. Therefore, the union of two distinct paths from $\bar u$ to $\bar v$ contains a cycle. Since each path has length at most $2(d+1)$, every vertex on each path lies in the $2(d+1)$-hop neighborhood of $\bar u$ in $\bar G$. Therefore, the $2(d+1)$-hop neighborhood of $\bar u$ in $\bar G$ contains a cycle, contradicting the condition that event $\ca_{3,u}$ occurs. This contradiction shows that ${\bf M}_\mathrm{TB}^\infty(u,\td v)=0$ for every $\td v\neq \td u$.

\underline{\bf Conditioned on $\ca_{2,u}$:}  When event $\ca_{2,u}$ occurs, under \texttt{TBAlign}, 
\begin{align}
    {\bf M}_\mathrm{TB}^\infty(u,\td u)=\texttt{GreedyMatching}({\cw}^{\infty}(u,\td u),\theta^*,3)=1.
    \nonumber
\end{align}
\end{proof}

\subsection[Bounding the sizes of B and tilde B]{Bounding the sizes of $\cb$ and $\td\cb$}
\label{appd:bad-nodes}
\begin{lem}\label{lem:TBAlignBadNodes}
Suppose the same assumptions as in Proposition \ref{th:TBAlign}. For the sets $\cb$ and $\td\cb$ defined in \eqref{eq:Bset} and \eqref{eq:tBset}, we have
\begin{align}
    |\cb|+|\td\cb|
    \le 2n\exp\left(-\frac12(\log n)^{\alpha/2}\right)
    \label{eq:bad-node-bound}
\end{align}
with probability $1-O(\exp(-\frac12(\log n)^{\alpha/2}))$.
\end{lem}

\begin{proof}
    In this proof, we will focus on showing that 
    \begin{equation}
        \label{eq:eb}
        \E[|\cb|]=O\left(n\exp\left(-\frac14 (\log n)^t\right)\right).
    \end{equation}
    With~\eqref{eq:eb} in hand, we can obtain by symmetry that 
    \begin{equation}
        \label{eq:etdb}
        \E[|\td\cb|]=O\left(n\exp\left(-\frac14 (\log n)^t\right)\right).
    \end{equation}
    It follows by~\eqref{eq:eb} and~\eqref{eq:etdb} that
    \begin{align*}
        &\P\left(|\cb|+|\td\cb|\ge 2n\exp\left(-\frac12(\log n)^{\alpha/2}\right)\right)\\
        \le{} & \P\left(|\cb|\ge n\exp\left(-\frac12(\log n)^{\alpha/2}\right)\right)+\P\left(|\td\cb|\ge n\exp\left(-\frac12(\log n)^{\alpha/2}\right)\right)\\
        \stackrel{(a)}{\le}{}&
        2\exp\left(\frac12(\log n)^{\alpha/2}-\frac14(\log n)^t\right)\\
        \stackrel{(b)}{=} {}&O\left(\exp\left(-\frac12(\log n)^{\alpha/2}\right)\right),
    \end{align*}
    where the (a) follows by Markov's inequality and (b) follows by the assumption that $\alpha<t$.

    \textbf{\underline{Proof of~\eqref{eq:eb}:}}
    Fix a vertex $u\in \cv$. Define event
    \[
    \mathsf{BAD}_u:=\{\exists v\in \cv\text{ such that }
    \dist(u,v)\le d+1
    \text{ and }
    \deg(v)> (\log n)^t\}.
    \]
    Notice that for event $\mathsf{BAD}_u$ to happen, we either have $\deg(u)>(\log n)^t$ or there exists a path $(u,v_1,v_2,\ldots,v_l)$ of length $1\le l\le d+1$ such that $\deg(v_l)>(\log n)^t$. 
    Because $\deg(u)\sim\bin(n-1,\lambda/n)$, we can apply the Chernoff bound to get
    \begin{equation}
        \label{eq:degu}
        \P(\deg(u)>(\log n)^t)=O\left(\exp\left(-\frac12(\log n)^t\right)\right).
    \end{equation}
    Fix distinct vertices $v_1,\ldots,v_l$ in $\cv$. Because each edge appears in the graph $G$ with probability $\lambda/n$, the path $(u,v_1,v_2,\ldots,v_l)$ appears in $G$ with probability $\frac{\lambda^l}{n^l}$. Moreover, conditioned on the path $(u,v_1,v_2,\ldots,v_l)$, we have
    \[
    \deg(v_l)\sim 1+\bin(n-2,\lambda/n).
    \]
    This is because the path implies that $v_{l-1}$ is a neighbor of $v_l$, while the edge between $v_l$ and any other vertex in the graph appears independently with probability $\lambda/n$. From this argument, we know that
    \begin{align}
        &\P(\text{the path }(u,v_1,\ldots,v_l) \text{ appears in $G$ and }\deg(v_l)>(\log n)^t)\nonumber\nonumber\\
        ={}&\frac{\lambda^l}{n^l}\P\left(\bin\left(n-2,\frac{\lambda}{n}\right)>(\log n)^t-1\right)\nonumber\\
        ={}&O\left(\frac{\lambda^l}{n^l}\exp\left(-\frac12(\log n)^t\right)\right).\label{eq:degv}
    \end{align}
    For each fixed $l$, there are at most $n^l$ ways to choose the vertices $v_1,\ldots,v_l$ from the graph. Therefore, we can apply the union bound to obtain
    \begin{align}
        \P(\mathsf{BAD}_u)&\le
        O\left(\exp\left(-\frac12(\log n)^t\right)\right)+
        \sum_{l=1}^{d+1}n^l O\left(\frac{\lambda^l}{n^l}\exp\left(-\frac12(\log n)^t\right)\right)\nonumber\\
        &=O\left(\lambda^{d+1}\exp\left(-\frac12(\log n)^t\right)\right)\nonumber\\
        &=O\left(\exp\left(((\log n)^\gamma+1)\log \lambda-\frac12(\log n)^t\right)\right)\nonumber\\
        &=O\left(\exp\left(-\frac14(\log n)^t\right)\right),\label{eq:badu}
    \end{align}
    where the last equality follows because $\gamma<t$ and $\log \lambda=\Theta(\log\log n)$. Finally~\eqref{eq:badu} implies~\eqref{eq:eb} because $\E[|\cb|]=\sum_{u\in\cv}\P(\mathsf{BAD}_u)$.
\end{proof}

\section{General properties of correlated \erdos--\renyi graph pairs}\label{appd:er-properties}
\subsection{Neighborhood size}
\label{appd:neighbor-size}
\begin{lem}
\label{lem:neighbor_size}
    Consider an \erdos--\renyi random graph $G(\cv,\ce)\sim\mathrm{ER}(n,\frac{\lambda}{n})$, where $\lambda=(\log n)^{\alpha+o(1)}$ for some constant $\alpha\in (0,1)$. Let $d=(\log n)^\gamma$ for some constant $\gamma\in (0,1)$, and for each vertex $i$ in $G$, let $\cn(i,d)$ denote the set of vertices within distance $d$ from $i$ in $G$. Then, we have
    \begin{equation}
        \label{eq:neighbor_size}
        \P\left(\exists i\in \cv:|\cn(i,d)|\ge 8e^2\lambda^d\log n\right)= \td{O}(n^{-1/3}).
    \end{equation}
\end{lem}

\begin{proof}
    Fix a vertex $i\in\cv$.
    We will show that
    \begin{equation*}
        \P\left(|\cn(i,d)|\ge 8e^2\lambda^d\log n\right)= \td{O}(n^{-4/3}).
    \end{equation*} 
    Then~\eqref{eq:neighbor_size} then follows by applying a union bound on the $n$ vertices in $\cv$.
    
    For each $k\in \{0,\ldots,d\}$, let $\cs(i,k)$ denote the set of vertices exactly  $k$-hop away from vertex $i$. It follows from this definition that 
    \[
    |\cn (i,d)|=\sum_{k=0}^d|\cs (i,k)|. 
    \]
    Define event 
    \[
    \cb_k:=\left\{|\cs (i,k)|<4(\log n)\lambda^k\prod_{h=0}^k(1+\lambda^{-h/2})\right\}.
    \]
    Notice that 
    \begin{align*}
        \sum_{k=0}^d4(\log n)\lambda^k\prod_{h=0}^k(1+\lambda^{-h/2})&\le 4(\log n)\prod_{h=0}^\infty(1+\lambda^{-h/2})\sum_{k=0}^d\lambda^k\\
        &\stackrel{(a)}{\le}4(\log n)e^2\frac{(\lambda^{d+1}-1)}{\lambda-1}\\
        &\le 8e^2\lambda^d\log n,
    \end{align*}
    where (a) follows because 
    \begin{align*}
        \prod_{h=0}^\infty(1+\lambda^{-h/2})\le \exp\left(\sum_{h=0}^\infty\lambda^{-h/2}\right)\le e^2
    \end{align*}
    for all large enough $n$.
    Therefore, it suffices to show that $\P(\cap_{k=0}^d\cb_k)= 1-\td{O}(n^{-4/3})$. We rewrite the probability of this intersection event as
    \begin{align}
       \P(\cap_{k=0}^d\cb_k)&=\P(\cb_0)-\sum_{k=1}^{d}\P(\cb_k^c\cap(\cap_{h=0}^{k-1}\cb_{h}))\nonumber\\
       &\ge \P(\cb_0)-\sum_{k=1}^{d}\P(\cb_k^c|\cap_{h=0}^{k-1}\cb_{h}).\label{eq:chain}
    \end{align}
    Notice that $\cb_0$ holds almost surely, and it suffices to bound each term in the summation in~\eqref{eq:chain}. Fix some $k\in\{0,\ldots,d\}$. Conditioned on events $\sum_{l=0}^{k-1}|\cs (i,l)|=m$ and $|\cs (i,k-1)|=m'$, we have
    \[
    |\cs (i,k)|\sim \bin(n-m,1-(1-\lambda/n)^{m'}).
    \]
    Notice that $1-(1-\lambda/n)^{m'}\le \lambda m'/n$. Then conditioned on event $\cap_{l=0}^{k-1}\cb_{l}$, we have
    \[
    |\cs (i,k)|\stackrel{\mathrm{sto.}}{\le}\bin\left(n,\frac{4(\log n)\lambda^{k}\prod_{h=0}^{k-1}(1+\lambda^{-h/2})}{n}\right),
    \]
    where $\sto$ is the stochastic dominance operator. By the Bennett's inequality stated in Lemma~\ref{lem:bennett}, we have
    \begin{align}
        &\P(\cb_k^c|\cap_{h=0}^{k-1}\cb_{h})\nonumber\\
        \le{}& \P\left(\bin\left(n,\frac{4(\log n)\lambda^{k}\prod_{h=0}^{k-1}(1+\lambda^{-h/2})}{n}\right)\ge 4(\log n)\lambda^{k}\prod_{h=0}^{k}(1+\lambda^{-h/2})\right)\nonumber\\
        \le{}& \exp\left(-4(\log n)\lambda^{k}\left(\prod_{h=0}^{k-1}(1+\lambda^{-h/2})\right)\left(1-\frac{4(\log n)\lambda^{k}\prod_{h=0}^{k-1}(1+\lambda^{-h/2})}{n}\right) \phi(\lambda^{-k/2})\right),\label{eq:size_bennett}
    \end{align}
    where $\phi(x):=(1+x)\log(1+x)-x$.
    We now bound each term in this exponent. 
    Notice that $\phi(x)\ge \frac{x^2}{3}$ for all $x\in [0,1]$, and it follows that $$\phi(\lambda^{-k/2})\ge \frac{\lambda^{-k}}{3}.$$ 
    We also have 
    $$\prod_{h=0}^{k-1}(1+\lambda^{-h/2})\ge 1+\lambda^{0}=2.$$
    Under the assumptions $\lambda=(\log n)^{\alpha+o(1)}$ and $d=(\log n)^\gamma$, we have $\lambda^d=n^{o(1)}$. Therefore, 
    \begin{align*}
        1-\frac{4(\log n)\lambda^{k}\prod_{h=0}^{k-1}(1+\lambda^{-h/2})}{n}\ge 1-\frac{4e^2(\log n)\lambda^k}{n}\ge\frac12
    \end{align*}
    for all large enough $n$. Substituting these three inequalities into~\eqref{eq:size_bennett} yields
    \begin{align*}
        P(\cb_k^c|\cap_{h=0}^{k-1}\cb_{h})&\le\exp\left(-4(\log n)\lambda^k\times 2\times\frac12\times \frac{\lambda^{-k}}{3}\right)\\
        &\le \exp\left(-\frac43\log n\right)=n^{-4/3}.
    \end{align*}

    Notice that the above inequality holds for every $k\in [d]$. Therefore, we have 
    \[
    \P(\cap_{k=0}^d\cb_k)\ge 1-dn^{-4/3}=1-\td{O}(n^{-4/3}),
    \]
    which completes the proof.
\end{proof}
\subsection{Local tree PMF ratio}
\label{appd:tree-pmf}
For a pair of vertices $(u,w)\in \cv\times\cv$, $\ct$ denotes the local tree rooted at $w$ in $G\setminus u$ with depth up to $d$. For a pair of vertices $(\td v,\td x)\in \td\cv\times\td\cv$, $\td\ct$ denotes the local tree rooted at $\td x$ in $\td G\setminus \td v$ with depth up to $d$.
\begin{lem}
    \label{lem:tree_pmf}
    Suppose $(G,\td G)\sim \mathrm{CER}(n,\lambda,s)$.
    For any pair of unlabeled rooted tree $(t,\td t)\in \cz\times\cz$, both with depth within $d$, we have
    \begin{align}
        &\frac{\P\big(\{[\ct]=t\}\cap\{[{\td\ct}]=\td t\}\cap \{\bar\cv(\ct)\cap \bar\cv(\td\ct)=\emptyset\}\big)}{\bq([\ch_d]=t,[\td\ch_d]=\td t)}\nonumber\\
        &\le \exp\left(\frac{\lambda}{n}\left((|\cv(t)|+1)|\cv(t)|+|\cv(\td t)|(1+|\cv(t)|+|\cv(\td t)|)\right)\right)\label{eq:tree_pmf}
    \end{align}
    
\end{lem}

\begin{proof}
    To begin with, we define a mapping $\ell$ that relabels a local tree in an \erdos--\renyi graph into a Galton--Watson tree that is labeled by the children orders. Consider a non-empty local tree $\ct$. The image $\ell(\ct)$ has the same underlying rooted tree structure as $\ct$, but its vertex labels are replaced by order labels as follows. First, the root of $\ell(\ct)$ is assigned the label $(1)$. For each vertex $v\in\cv(\ct)$, we order the children of $v$ according to their indices in $\ct$: if the children of $v$ have indices
    \[
    i_1<i_2<\cdots<i_{c_v},
    \]
    then the child with index $i_j$ is assigned order $j$ among the children of $v$. Given these sibling orders, the label of each vertex in $\ell(\ct)$ is obtained by concatenating the label of its parent and its order.

    With the definition of $\ell$, we claim that for each $(h,\td h)\in \cx\times\cx$, both with depth up to $d$,
    \begin{align}
        &\frac{\P\big(\{\ell(\ct)=h\}\cap\{\ell(\td\ct)=\td h\}\cap \{\bar\cv(\ct)\cap \bar\cv(\td\ct)=\emptyset\}\big)}{\bq(\ch_d=h,\td\ch_d=\td h)}\nonumber\\
        \le{}& \exp\left(\frac{\lambda}{n}\left((|\cv(h)|+1)|\cv(h)|+|\cv(\td h)|(1+|\cv(h)|+|\cv(\td h)|)\right)\right).\label{eq:cxd}
    \end{align}

    Since the mapping $\ell$ preserves the tree structure, we know that $[\ct]=t$ only if $\ell(\ct)=h$ for some $h\in\cx$ that satisfies $[h]=t$. Therefore, we have for each $(t,\td t)\in\cz\times\cz$,
    \begin{align*}
        &\frac{\P\big(\{[{\ct}]=t\}\cap\{[{\td\ct}]=\td t\}\cap \{\bar\cv(\ct)\cap \bar\cv(\td\ct)=\emptyset\}\big)}{\bq([\ch_d]=t,[\td\ch_d]=\td t)}\\
        ={}& \frac{\sum_{\substack{(h,\td h)\in \cx\times\cx:\\
        [h]=t,[\td h]=\td t
        }}\P\big(\{\ell(\ct)=h\}\cap\{\ell(\td\ct)=\td h\}\cap \{\bar\cv(\ct)\cap \bar\cv(\td\ct)=\emptyset\}\big)}{\sum_{\substack{(h,\td h)\in \cx\times\cx:\\
        [h]=t,[\td h]=\td t
        }}\bq(\ch_d=h,\td\ch_d=\td h)}\\
        \le{}& \exp\left(\frac{\lambda}{n}\left((|\cv(t)|+1)|\cv(t)|+|\cv(\td t)|(1+|\cv(t)|+|\cv(\td t)|)\right)\right),
    \end{align*}
    where the last inequality follows by~\eqref{eq:cxd} and the fact that $|\cv(t)|=|\cv(h)|$ given that $[h]=t$.
    This completes the proof of~\eqref{eq:tree_pmf}. In the rest of the proof, we focus on proving the claimed inequality~\eqref{eq:cxd}.

        By the definition of the independent tree distribution $\bq(\ch,\td \ch)$, we have
        \[
        \bq(\ch_d=h,\td\ch_d=\td h)=\gw(\ch_d=h)\gw(\td\ch_d=\td h).
        \]
        
        In the following, we will focus on proving
        \begin{equation}
            \label{eq:t}
            \frac{\P(\ell(\ct)=h)}{\gw(\ch_d=h)}\le \exp\left(\frac{\lambda}{n}(|\cv(h)|+1)|\cv(h)|\right),
        \end{equation}
        and 
        \begin{equation}
            \label{eq:tdt}
            \frac{\P(\{\ell(\td\ct)=\td h\}\cap \{\bar\cv(\ct)\cap \bar\cv(\td\ct)=\emptyset\}\,|\,\ell(\ct)=h)}{\gw(\td\ch_d=\td h)}\le \exp\left(\frac{\lambda}{n}|\cv(\td h)|(1+|\cv(h)|+|\cv(\td h)|)\right).
        \end{equation}
        The desired inequality~\eqref{eq:cxd} then follows by multiplying~\eqref{eq:t} and~\eqref{eq:tdt}.

        \underline{\textbf{Proof of~\eqref{eq:t}:}} 
        For each $k\in\{0,\ldots,d-1\}$, we define a total ordering for the vertices at depth $k$ in the order-labeled rooted tree $h$: a vertex is ranked higher if its parent is ranked higher among vertices at depth $k-1$; among vertices sharing the same parent, the order is determined by the ordering of siblings. Recall that $\cv(h,k)$ denotes the set of vertices at depth $k$ in graph $h$. For each $k$ and $r\in [|\cv(h,k)|]$, let $c_{k,r}$ denote the number of children of the $r$-th ordered vertex at depth $k$. Then by the definition of the $\gw$ distribution, we have

        \begin{equation}
            \label{eq:gw_t}
            \gw(\ch_d=h)=\prod_{k=0}^{d-1}\prod_{r=1}^{|\cv(h,k)|}\P(\poi(\lambda)=c_{k,r}),
        \end{equation}
        where we define the product $\prod_{r=1}^{|\cv(h,k)|}\P(\poi(\lambda)=c_{k,r})$ to be $1$ when $|\cv(h,k)|=0$.

        To bound the numerator $\P(\ell(\ct)=h)$, we describe a breadth-first exploration procedure for constructing $\ct$. The exploration is performed in the graph $G\setminus u$ and is rooted at $w$. We reveal the vertices of $\ct$ level by level, following the ordering of vertices in the target tree $h$ introduced above.  More precisely, among vertices at the same depth, a vertex is explored earlier if its parent is explored earlier, and among siblings, the order is determined by their vertex labels in $G$. 

        At the beginning of the exploration, the root $w$ is the only discovered vertex. Since the vertex $u$ is removed, the number of neighbors of $w$ in $G\setminus u$ has distribution
        \[
        \bin\left(n-2,\frac{\lambda}{n}\right).
        \]
        To obtain a local tree isomorphic to $h$, this number must be equal to $c_{0,1}$.

        More generally, suppose that the exploration is about to reveal the children of the $r$-th ordered vertex at depth $k$ in the breadth-first order, where $k\in\{0,\ldots,d-1\}$ and $r\in[|\cv(h,k)|]$.
        These children must satisfy:
        \begin{itemize}
            \item The total number of these children is exactly $c_{k,r}$. We also note that the number of these children is distributed as 
            \[
            \bin\left(n-1-x_{k,r},\frac{\lambda}{n}\right),
            \]
            where $x_{k,r}$ denotes the number of vertices that have already been discovered before this step in the breadth-first exploration.
            \item In addition to matching the prescribed offspring numbers, the exploration must not create any cycle. Equivalently, every newly discovered vertex must be connected to the previously explored part only through its parent, and there must be no edges among vertices discovered at the same step. If any newly discovered vertex has an additional edge to an already discovered vertex, or if there is an edge among two newly discovered vertices, then the explored $d$-neighborhood is not a tree and $\ct$ becomes empty.
        \end{itemize}

        Through this breadth-first exploration, we have
        \begin{align}
            &\P(\ell(\ct)=h)\nonumber\\
            ={}&\prod_{k=0}^{d-1}\prod_{r=1}^{|\cv(h,k)|}\P(\text{the number of children is $c_{k,r}$ and the new children do not create any cycle})\nonumber\\
            \le{}& \prod_{k=0}^{d-1}\prod_{r=1}^{|\cv(h,k)|}\P\left(\bin\left(n-1-x_{k,r},\frac{\lambda}{n}\right)=c_{k,r}\right)\nonumber\\
            \stackrel{(a)}{\le}{}&\prod_{k=0}^{d-1}\prod_{r=1}^{|\cv(h,k)|}\exp\left(\frac{\lambda }{n}(1+x_{k,r}+c_{k,r})\right)\P(\poi(\lambda)=c_{k,r})\nonumber\\
            \stackrel{(b)}{\le}{}&\exp\left(\frac{\lambda}{n}|\cv(h)|(1+|\cv(h)|)\right)\prod_{k=0}^{d-1}\prod_{r=1}^{|\cv(h,k)|}\P(\poi(\lambda)=c_{k,r}).\label{eq:t_bfs}
        \end{align}
        Here (a) follows by Lemma~\ref{lem:binom-poi-gap}, and (b) follows because $c_{k,r}+x_{k,r}\le |\cv(h)|$ throughout the exploration process and there are at most $|\cv(h)|$ terms in the double product. This completes the proof of~\eqref{eq:t}.

        \underline{\textbf{Proof of~\eqref{eq:tdt}:}} For the unlabeled rooted tree $\td h$, we similarly define the notations $\td c_{k,r}$ and $\cv(\td h,k)$. Then we have
        \begin{equation}
            \label{eq:gw_tdt}
            \gw(\td\ch_d=\td h)=\prod_{k=0}^{d-1}\prod_{r=1}^{|\cv(\td h,k)|}\P(\poi(\lambda)=\td c_{k,r}).
        \end{equation}

        Now we move on to bound the numerator. Notice that
        \begin{align}
            &\P(\{\ell(\td\ct)=\td h\}\cap \{\bar\cv(\ct)\cap \bar\cv(\td\ct)=\emptyset\}\,|\,\ell(\ct)=h)\nonumber\\
            \le{}& \max_{\substack{\bar\cs\subseteq \bar\cv:\\
            \bar w\in \bar\cs,\\
            \bar x\notin\bar\cs,\\
            |\bar\cs|=|\cv(h)|}}\P(\{\ell(\td\ct)=\td h\}\cap \{\bar\cv(\ct)\cap \bar\cv(\td\ct)=\emptyset\}\,|\,\{\ell(\ct)=h\}\cap \{\bar{\cv}(\ct)=\bar\cs\})\label{eq:max_cond}
        \end{align}
        Fix an arbitrary subset $\bar\cs\subseteq \bar\cv$ that satisfies the criteria in~\eqref{eq:max_cond}. We focus on bounding the conditional probability 
        \[
        \P(\{\ell(\td\ct)=\td h\}\cap \{\bar\cv(\ct)\cap \bar\cv(\td\ct)=\emptyset\}\,|\,\{\ell(\ct)=h\}\cap \{\bar{\cv}(\ct)=\bar\cs\}).
        \]

        We again consider the breadth-first exploration procedure to construct $\td\ct$. To guarantee the desired event $\{\ell(\td\ct)=\td h\}\cap \{\bar\cv(\ct)\cap \bar\cv(\td\ct)=\emptyset\}$, when we explore the children of the $r$-th vertex at depth $d$, we have to ensure
        \begin{itemize}
            \item None of the children are corresponding to the vertices in $\bar\cs$. Otherwise, $\td\ct$ has a non-empty overlap with the tree $\ct$;
            \item The number of children is exactly $c_{k,r}$. Given the event that the vertex has no children corresponding to $\bar\cs$, the distribution of the number of children at this step is given by 
            \[
            \bin\left(n-1-|\cv(\td h)|-\td x_{k,r},\frac{\lambda}{n}\right),
            \]
            where $\td x_{k,r}$ is the number of vertices discovered before this step. This is because conditioning on the event $\{\ell(\ct)=h\}\cap \{\bar{\cv}(\ct)=\bar\cs\}$ does not reveal any information about any edge without an endpoint in $\bar{\cs}$.
            \item The newly discovered children are connected to the previously explored part only through its parent, and there are no edges among them.
        \end{itemize}
        By this exploration process, we have
        \begin{align*}
            &\P(\{\ell(\td\ct)=\td h\}\cap \{\bar\cv(\ct)\cap \bar\cv(\td\ct)=\emptyset\}\,|\,\{\ell(\ct)=h\}\cap \{\bar{\cv}(\ct)=\bar\cs\})\\
            ={}&\prod_{k=0}^{d-1}\prod_{r=1}^{|\cv(\td h,k)|}\P(\text{no children corresponding to any vertex in }\bar\cs\text{ and}\\
            &\quad\quad\quad\quad\quad\quad\text{the number of new children is $\td c_{k,r}$ and}\\
            &\quad\quad\quad\quad\quad\quad\text{the new children do not create any cycle})\\
            \le{}& \prod_{k=0}^{d-1}\prod_{r=1}^{|\cv(\td h,k)|} \P\left(\bin\left(n-1-|\cv(\td h)|-\td x_{k,r},\frac{\lambda}{n}\right)=\td c_{k,r}\right)\\
            \stackrel{(c)}{\le}{}&\prod_{k=0}^{d-1}\prod_{r=1}^{|\cv(\td h,k)|}\exp\left(\frac{\lambda}{n}(1+\td x_{k,r}+\td c_{k,r}+|\cv(\td h)|)\right)\P(\poi(\lambda)=\td c_{k,r})\\
            \stackrel{(d)}{\le}{}&\exp\left(\frac{\lambda}{n}|\cv(\td h)|(1+|\cv(h)|+|\cv(\td h)|)\right)\prod_{k=0}^{d-1}\prod_{r=1}^{|\cv(\td h,k)|}\P(\poi(\lambda)=\td c_{k,r}).
        \end{align*}
        Here (c) and (d) follows by the similar reason as (a) and (b) in~\eqref{eq:t_bfs}. This completes the proof of~\eqref{eq:tdt}.
    \end{proof}

\subsection{Coupling correlated \erdos--\renyi neighborhoods with correlated Galton--Watson trees}
\label{appd:tree-couple}
Let $d\ge 0$ be an integer, and $u$ and $\td u$ be two corresponding vertices in $G$ and $\td G$. Define $\ct$ to be the local tree rooted at $u$ in $G$ such that $\ct$ is the $d$-hop-neighborhood of $u$ in $G$ if it is a tree; otherwise, $\ct$ is an empty tree. Similarly, define $\td\ct$ as the local tree rooted at $\td u$ in $\td G$.
\begin{lem}
\label{lem:tree_tv}
    Suppose that $\lambda=(\log n)^{\alpha+o(1)}$ for some constant $\alpha\in (0,1)$ and $s\in (\sqrt{C_\mathrm{Otter}},1]$ be a constant. Assume $d=(\log n)^{\gamma}$ for some constant $\gamma\in (0,1)$. Then there exists a coupling  between $(\ct,\td \ct)$ and $(\ch,\td\ch)\sim \bp$ such that
    \begin{equation}
        \label{eq:dtv}
        \P([\ct]=[\ch_d],[\td\ct]=[\td\ch_d])= 1-O(n^{-1/2}).
    \end{equation}
    
\end{lem}

\begin{proof}
    We first define a mapping $\ell$ that relabels a pair of local trees to the format of a pair of Galton--Watson trees. We first describe the relabeling of $\ct$. If $\ct$ is an empty tree, then the output $\ell(\ct)$ is also empty. Now suppose $\ct$ is non-empty. The root of $\ell(\ct)$ is assigned the label $(1)$. For a vertex $v\in\cv(\ct)$, if the corresponding vertex $\td v$ also belongs to $\cv(\td\ct)$, then the children of $v$ are ordered in two blocks. First, we list the common children, namely those children $y$ of $v$ such that the corresponding vertex $\td y$ is a child of $\td v$ in $\td\ct$. These common children are ordered increasingly according to their indices. If there are $c^*_v$ such children, they receive orders $1,\ldots,c^*_v$. Next, we list the remaining children of $v$. These children are again ordered increasingly according to their indices and receive orders $c^*_v+1,\ldots,c^*_v+c_v$, where $c_v$ is the number of such remaining children.
    If, on the other hand, the corresponding vertex $\td v$ does not belong to $\cv(\td\ct)$, then all children of $v$ are ordered increasingly according to their indices and receive orders $1,\ldots,c_v$, where $c_v$ is the number of children of $v$. After this ordering of the children has been specified for every vertex in $\ell(\ct)$, the order labels are assigned recursively: if a vertex has label $(a_0,\ldots,a_k)$ and its $i$-th ordered child is considered, then this child is assigned the label $(a_0,\ldots,a_k,i)$. The relabeling of $\td\ct$ is defined analogously. An example of this relabeling is presented in Figure~\ref{fig:relabel}.

    \begin{figure}[htbp]
    \centering
    \includegraphics[width=0.7\linewidth]{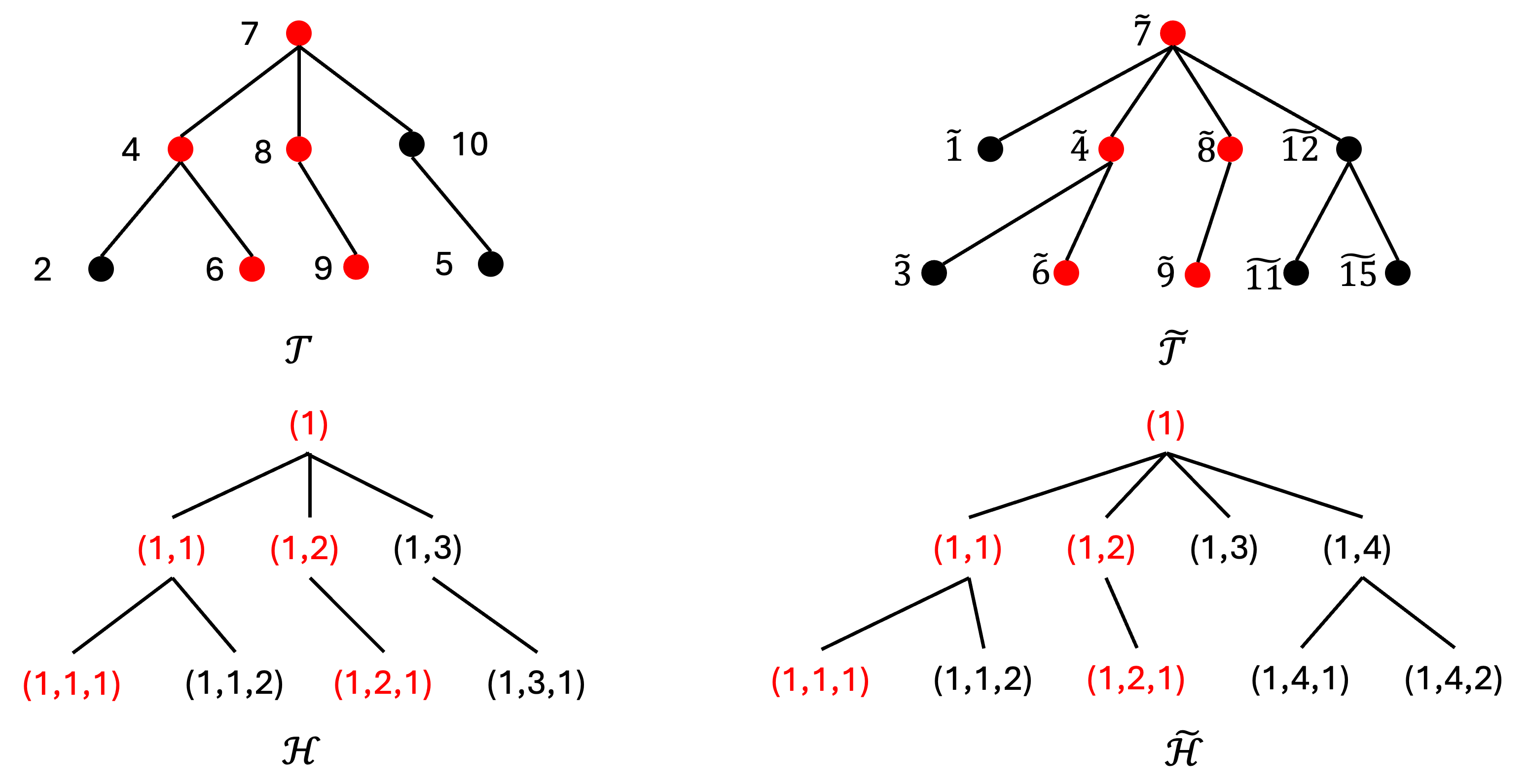} 
    \caption{An illustration of the relabeling mapping $\ell$.
    The common nodes in the trees are marked in red.
    We have $\ell(\ct,\td\ct)=(\ch,\td\ch)$ in this example.}
    \label{fig:relabel}
\end{figure}

    In the rest of this proof, we will focus on constructing a coupling such that
    \begin{equation}
        \label{eq:couple}
        \P(\ell(\ct,\td \ct)=(\ch_d,\td\ch_d))=1-O(n^{-1/2}).
    \end{equation}
    This implies the lemma statement~\eqref{eq:dtv} because $\ell(\ct,\td\ct)=(\ch_d,\td\ch_d)$ implies $[\ct]=[\ch_d]$ and $[\td\ct]=[\td\ch_d]$. Therefore, we are left to construct a coupling so that~\eqref{eq:couple} holds.

    The coupling is constructed recursively over depths. For each $k\in [d]$, let $\ct_k$ denote the rooted tree obtained by truncating $\ct$ at depth $k$. More specifically, $\ct_k$ is the subgraph of $\ct$ induced by the set of vertices within depth $k$ from the root.
    We note that we have $\ct=\ct_d$ under this definition.
    The truncated trees $\td\ct_k$, $\ch_k$ and $\td\ch_k$ are defined similarly.  By these definitions, we have
    \[
    \{\ell(\ct,\td\ct)=(\ch_d,\td\ch_d)\}=\cap_{k=1}^d \{\ell(\ct_k,\td\ct_k)=(\ch_k,\td\ch_k)\}.
    \]

    \textbf{\underline{Step 1. Coupling $(\ct_1,\td\ct_1)$ with $(\ch_1,\td\ch_1)$:}} Recall that $u$ and $\td u$ are the roots of $\ct$ and $\td\ct$, respectively. Under the mapping $\ell$, both roots are relabeled by $(1)$, which agrees with the root labels of $\ch$ and $\td\ch$. 
    Let $\sfch_u^*$ denote the set of common children of $u$ and $\td u$ in $\ct$ and $\td\ct$. More precisely, $\sfch_u^*$ 
    consists of all indices $i\in[n]$ such that $i$ is a child of $u$ in $\ct$ and $\td i$ is a child of $\td u$ in $\td\ct$. Let $\sfch_u$ denote the set of children of $u$ in $\ct$ whose corresponding vertices are not children of $\td u$ in $\td\ct$, and let $\widetilde\sfch_{\td u}$ denote the analogous set of children of $\td u$ in $\td\ct$ that do not correspond to children of $u$ in $\ct$. We refer to these as the exclusive children of $u$ and $\td u$, respectively.
    We define the following three events:
    \begin{itemize}
        \item $\ca_1^*$ is the event that under the maximal coupling~\cite{levin2017markov} between random variables $|\sfch_u^*|$, $|\sfch_u|$, and  $|\wdtd\sfch_{\td u}|$, and independent Poisson random varables, $X_u^*\sim\poi(\lambda s)$, $X_u\sim\poi(\lambda(1-s))$, and $\td X_{\td u}\sim\poi(\lambda(1-s))$, we have $|\sfch_u^*|=X_u^*$, $|\sfch_u|=X_u$ and $|\wdtd\sfch_{\td u}|=\td X_{\td u}$.
        
        \item $\cb_1$ is the event that $|\sfch_u^*\cup\sfch_u\cup\wdtd\sfch_{\td u}|\le \lambda(2-s)\log n$.
        \item $\cc_1$ is the event that there are no edges in the union graph $\bar G$ between any pair of vertices in $\sfch_u^*\cup\sfch_u\cup\wdtd\sfch_{\td u}$.

    \end{itemize}
    We claim that $\ca_1^*\cap \cc_1$ implies $\ell(\ct_1,\td\ct_1)=(\ch_1,\td\ch_1)$. First, under event $\cc_1$, no cycles exist in the $1$-hop-neighborhood of $u$ and $\td u$. Therefore, both $\ct_1$ and $\td \ct_1$ are non-empty.
    Under the correlated tree-pair distribution $\bp$, the number of children of the root in the intersection tree has distribution $\poi(\lambda s)$, while the numbers of exclusive children added to the two trees have distribution $\poi(\lambda(1-s))$ independently. Hence, on $\ca_1^*$, the numbers of common and exclusive root children in $(\ct_1,\td\ct_1)$ agree with those in $(\ch_1,\td\ch_1)$. Moreover, the definition of $\ell$ orders common children first and exclusive children afterwards, with each block ordered increasingly by index. 
    Therefore, after applying $\ell$, the labels of all vertices up to depth one agree with those in $(\ch_1,\td\ch_1)$.
    
    The event $\cb_1$ is not needed for this depth-one agreement. It is used to control the size of the explored neighborhood when bounding $\P(\cc_1^c)$ and the other error events at deeper depths.
    
    We claim that
    \begin{equation}
    \label{eq:root_multi}
    (|\sfch_u^*|,|\sfch_u|,|\wdtd\sfch_{\td u}|,|\mathsf{NC}_u|)\sim \mathrm{Multi}\left(n-1,\frac{\lambda s}{n},\frac{\lambda (1-s)}{n},\frac{\lambda (1-s)}{n},1-\frac{\lambda(2-s)}{n}\right),
    \end{equation}
    where $|\mathsf{NC}_u|$ is the number of vertices that are not connected to $u$ in $G$ and not connected to $\td u$ in $\td G$. 
    This follows by the definition of the correlated \erdos--\renyi graph pair model, in which an edge appears in both $G$ and $\td G$ with probability $\frac{\lambda s}{n}$, only in $G$ with probability $\frac{\lambda(1-s)}{n}$ and only in $\td G$ with probability $\frac{\lambda(1-s)}{n}$. Also, these edges are mutually independent.
    
    By Lemma~\ref{lem:multi-poi}, we can couple these random variables with $X_u^*$, $X_u$ and $\td X_{\td u}$ such that
    \begin{equation}
        \label{eq:a1}
        \P(\ca_1^{*c})=O\left(\frac{\log^2 n}{n}\right).
    \end{equation}
    By the Chernoff bound (equation~\eqref{eq:CB_Upper} in Lemma~\ref{lem:CB}), we have
    \begin{equation}
        \label{eq:b1}
        \P(\cb_1^c)\le \P\left(\bin\left(n-1,\frac{\lambda(2-s)}{n}\right)\ge \lambda(2-s)\log n\right)= n^{-\omega(1)}.
    \end{equation}
    By the union bound over all vertex pairs in $\sfch_u^*\cup\sfch_u\cup\wdtd\sfch_{\td u}$, we have 
    \begin{equation}
        \label{eq:c1}
        \P(\cc_1^c\,|\,\cb_1)=O\left(\frac{\lambda^3\log^2 n}{n}\right).
    \end{equation}
    Equations~\eqref{eq:a1},~\eqref{eq:b1} and~\eqref{eq:c1} together imply that 
    \begin{equation}
        \label{eq:depth1}
        \P(\ca^{*c}_1\cup\cb_1^c\cup\cc_1^c)=O\left(\frac{\lambda^3\log^2 n}{n}\right).
    \end{equation}

    \textbf{\underline{Step 2. Coupling $(\ct_k,\td\ct_k)$ with $(\ch_k,\td\ch_k)$ given $\ell(\ct_{k-1},\td\ct_{k-1})=(\ch_{k-1},\td\ch_{k-1})$:}} 
    For each depth $k\ge 1$, we recursively define three sets of vertices: $\com_k$, $\excl_k$, and $\wdtd{\excl}_k$. The set $\com_k$ consists of the vertices at depth $k$ that appear in both local trees under the correspondence, while $\excl_k$ and $\wdtd{\excl}_k$ consist of the vertices appearing exclusively in $\ct$ and $\td\ct$, respectively. At depth one, we initialize
    \[
        \com_1:=\sfch_u^*,\qquad
        \excl_1:=\sfch_u,\qquad
        \wdtd{\excl}_1:=\wdtd{\sfch}_{\td u}.
    \]

    Now we move on to consider a general $k\ge 2$, and recursively define the sets $\com_k$, $\excl_k$, and $\wdtd{\excl}_k$.
    For each $v\in\com_{k-1}$, let $\sfch_v^*$ denote the set of common children of $v$ and $\td v$ in $\ct$ and $\td\ct$; that is, $\sfch_v^*$ contains all children $y$ of $v$ in $\ct$ such that $\td y$ is a child of $\td v$ in $\td\ct$. Let $\sfch_v$ denote the set of children of $v$ in $\ct$ whose counterparts are not children of $\td v$ in $\td\ct$, and let $\wdtd{\sfch}_{\td v}$ denote the set of children of $\td v$ in $\td\ct$ whose counterparts are not children of $v$ in $\ct$. For each $v\in\excl_{k-1}$, let $\sfch_v$ denote the set of children of $v$ in $\ct$. Similarly, for each $\td v\in\wdtd{\excl}_{k-1}$, let $\wdtd{\sfch}_{\td v}$ denote the set of children of $\td v$ in $\td\ct$.

    We then define
    \[
    \com_{k}
    :=
    \bigcup_{v\in\com_{k-1}}\sfch_v^*,
    \]
    \[
    \excl_{k}
    :=
    \left(\bigcup_{v\in\com_{k-1}}\sfch_v\right)
    \cup
    \left(\bigcup_{v\in\excl_{k-1}}\sfch_v\right),
    \]
    and
    \[
    \wdtd{\excl}_{k}
    :=
    \left(\bigcup_{v\in\com_{k-1}}\wdtd{\sfch}_{\td v}\right)
    \cup
    \left(\bigcup_{\td v\in\wdtd{\excl}_{k-1}}\wdtd{\sfch}_{\td v}\right).
    \]

    For each depth $k\ge 2$, we define the following events that can be shown to imply the desired property $\ell(\ct_k,\td\ct_k)=(\ch_k,\td\ch_k)$.
    \begin{itemize}
        \item $\ca_k^*$ is the event that for each $v\in \com_{k-1}$, 
        under the maximal coupling between random variables $|\sfch_v^*|$, $|\sfch_v|$, and  $|\wdtd\sfch_{\td v}|$, and independent Poisson random varables, $X_v^*\sim\poi(\lambda s)$, $X_v\sim\poi(\lambda(1-s))$, and $\td X_{\td v}\sim\poi(\lambda(1-s))$, we have $|\sfch_v^*|=X_v^*$, $|\sfch_v|=X_v$ and $|\wdtd\sfch_{\td v}|=\td X_{\td v}$.

        \item $\ca_k$ is the event that for each $v\in \excl_{k-1}$, 
        under the maximal coupling between the random variable $|\sfch_v|$ and the Poisson random variable $X_v\sim\poi(\lambda)$, we have $|\sfch_v|=X_v$.

        \item $\td\ca_k$ is the event that for each $\td v\in \wdtd \excl_{k-1}$,  under the maximal coupling between the random variable $|\wdtd\sfch_{\td v}|$ and the Poisson random variable $\td X_{\td v}\sim\poi(\lambda)$, we have $|\wdtd\sfch_{\td v}|=\td X_{\td v}$.

        \item $\cb_k$ is the event that $|\com_k\cup\excl_k\cup\wdtd\excl_{k}|\le 2^{k-1}(\lambda(2-s))^k\log n$.
        
        \item $\cc_k$ is the event that there are no edges in the union graph $\bar G$ between any pair of vertices in $\com_k\cup\excl_k\cup\wdtd\excl_k$.

        \item $\cd_k$ is the event that in the union graph $\bar G$, there exists no vertex outside the two trees $\ct_{k-1}$ and $\td\ct_{k-1}$ that is adjacent to more than one vertex in $\com_{k-1}\cup\excl_{k-1}\cup\wdtd\excl_{k-1}$.

        \item $\ci_k$ is the event that there are no edges in $G$ that has an endpoint in $\cup_{i=1}^{k}(\excl_i\cup\com_i)$ and the other endpoint corresponding to any vertex in $\cup_{i=1}^{k}\wdtd\excl_i$.

        \item $\td\ci_k$ is the event that there are no edges in $\td G$ that has an endpoint in $\cup_{i=1}^{k}(\wdtd\excl_i\cup\com_i)$ and the other endpoint corresponding to any vertex in $\cup_{i=1}^{k}\excl_i$.
    \end{itemize}

    We claim that, for every $k\ge 1$,
    \[
        (\ca^*_1\cap\cb_1\cap\cc_1)
        \cap
        \bigcap_{i=2}^k
        (\ca_i^*\cap\ca_i\cap\wdtd{\ca}_i
        \cap \cb_i\cap\cc_i\cap\cd_i\cap\ci_i\cap\td\ci_i)
        \subseteq
        \{\ell(\ct_k,\td\ct_k)=(\ch_k,\td\ch_k)\}.
    \]
    We prove this claim by induction on $k$. The case $k=1$ has already been
    established in Step 1.

    Now suppose that the claim holds up to depth $k-1$, where $k\ge 2$. That is, assume
    \[
        (\ca^*_1\cap\cb_1\cap\cc_1)
        \cap
        \bigcap_{i=2}^{k-1}
        (\ca_i^*\cap\ca_i\cap\wdtd{\ca}_i
        \cap \cb_i\cap\cc_i\cap\cd_i\cap\ci_i\cap\td\ci_i)
        \subseteq
        \{\ell(\ct_{k-1},\td\ct_{k-1})=(\ch_{k-1},\td\ch_{k-1})\}.
    \]
    We show that, on the additional event
    \[
        \ca_k^*\cap\ca_k\cap\wdtd{\ca}_k
        \cap\cb_k\cap\cc_k\cap\cd_k\cap\ci_k\cap\td\ci_k,
    \]
    the equality extends to depth $k$.

    By the induction hypothesis, the order-labeled structures of $(\ct,\td\ct)$ and $(\ch,\td\ch)$ agree up to depth $k-1$. It remains to compare the children generated by vertices at depth $k-1$. For each common vertex $v\in\com_{k-1}$, the event $\ca_k^*$ couples the children numbers $|\sfch_v^*|$, $|\sfch_v|$ and $|\wdtd{\sfch}_{\td v}|$
    with independent Poisson random variables of means $\lambda s$, $\lambda(1-s)$, and $\lambda(1-s)$, respectively. These are exactly the offspring distributions used in the correlated tree-pair construction $\bp$ for vertices in the intersection tree and for the exclusive children added to the two trees. Similarly, for each exclusive vertex $v\in\excl_{k-1}$, the event $\ca_k$ couples $|\sfch_v|$ with a $\poi(\lambda)$ random variable, which is the offspring distribution of an exclusive Galton--Watson subtree attached to $\ch$. The event $\wdtd{\ca}_k$ gives the analogous coupling for every exclusive vertex $\td v\in\wdtd{\excl}_{k-1}$ in $\td\ct$. Therefore, the numbers of children generated from all vertices at depth $k-1$ agree with those in the coupled Galton--Watson tree pair. Moreover, by the definition of the relabeling map $\ell$, common children are ordered first, followed by exclusive children, and each block is ordered increasingly by index. This is precisely the deterministic ordering convention used in the construction of $(\ch,\td\ch)$. Hence, after applying $\ell$, the labels of all newly generated vertices at depth $k$ agree with those in $(\ch_k,\td\ch_k)$.

    Moreover, event $\cc_k$ rules out edges among vertices in $\com_k\cup\excl_k\cup\wdtd{\excl}_k,$ and event $\cd_k$ rules out any vertex outside the previously explored trees that is connected to more than one vertex in $\com_{k-1}\cup\excl_{k-1}\cup\wdtd{\excl}_{k-1}.$ Together, these events ensure that the newly revealed vertices do not create any cycles. Therefore, we have $\ell(\ct_k,\td\ct_k)=(\ch_k,\td\ch_k),$ which completes the induction.

    We note that the events $\cb_k$, $\ci_k$ and $\td\ci_k$ are not needed for the immediate implication $\ell(\ct_k,\td\ct_k)=(\ch_k,\td\ch_k)$. The event $\cb_k$ provides a size bound on the newly explored vertices, which is used to control the error events. Events $\ci_k$ and $\td\ci_k$ prevent overlaps of the exclusive part of the trees at later depths, thus ruling out undesired correlations that could affect the coupling.

    Next, we separately bound the probability of these events at depth $k$ conditioned on the event
    \[
        \cg_{k-1}:=(\ca^*_1\cap\cb_1\cap\cc_1)
        \cap
        \bigcap_{i=2}^{k-1}
        (\ca_i^*\cap\ca_i\cap\wdtd{\ca}_i
        \cap \cb_i\cap\cc_i\cap\cd_i\cap\ci_i\cap\td\ci_i).
    \]
    Notice that given $\cap_{i=1}^{k-1}\cb_i$ the total number of vertices within $k-1$ depth of the two trees satisfies
    \begin{align*}
            |\bar\cv(\ct_{k-1})\cup\bar\cv(\td\ct_{k-1})|=& 1+\sum_{i=1}^{k-1}|\com_i\cup\excl_i\cup\wdtd\excl_{i}|\\
            \le & 1+\sum_{i=1}^{k-1}2^{i-1}(\lambda(2-s))^i\log n\\
            \le & (2\lambda(2-s))^{k-1}\log n,
    \end{align*}
    where $\bar\cv(\ct_{k-1})$ (resp. $\bar\cv(\td\ct_{k-1})$) denotes the image of the nodes in $\ct_{k-1}$ (resp. $\td\ct_{k-1}$) in the union graph $\bar G$. 

    Given event $\cg_{k-1}$, for each $v\in \com_{k-1}$, the number of children has distribution  
    \begin{align}
    &(|\sfch_v^*|,|\sfch_v|,|\wdtd\sfch_{\td v}|,|\mathsf{NC}_v|)\nonumber\\& \quad \quad \quad\quad\quad\quad\sim \mathrm{Multi}\left(n-|\bar\cv(\ct_{k-1})\cup\bar\cv(\td\ct_{k-1})|,\frac{\lambda s}{n},\frac{\lambda (1-s)}{n},\frac{\lambda (1-s)}{n},1-\frac{\lambda(2-s)}{n}\right),\label{eq:general_multi}
    \end{align}
    where $|\mathsf{NC}_v|$ is the total number of vertices outside $|\bar\cv(\ct_{k-1})\cup\bar\cv(\td\ct_{k-1})|$ that are not connected to $v$ in $G$ and not connected to $\td v$ in $\td G$.
    The reasoning for~\eqref{eq:general_multi} is similar to that for~\eqref{eq:root_multi}. 
    The only difference is that we need to exclude all the already explored vertices in $|\bar\cv(\ct_{k-1})\cup\bar\cv(\td\ct_{k-1})|$, and hence there are $n-|\bar\cv(\ct_{k-1})\cup\bar\cv(\td\ct_{k-1})|$ vertices taken into consideration. 
    
    By Lemma~\ref{lem:multi-poi}, we can couple $(|\sfch_v^*|,|\sfch_v|,|\wdtd\sfch_{\td v}|)$ with $(X^*_v,X_v,\td X_{\td v})$ such that
    \[
    \P(|\sfch_v^*|=X_v^*,|\sfch_v|=X_v,|\wdtd\sfch_{\td v}|=\td X_{\td v})=1-O\left(\frac{(2\lambda(2-s))^k\log n+\log^2 n}{n}\right).
    \]
    Taking a union bound over the set $\com_{k-1}$, whose size is bounded by $(2\lambda(2-s))^{k-1}\log n$ conditioned on event $\cap_{i=1}^{k-1}\cb_i$, yields 
    \begin{equation}
        \label{eq:a*k}
        \P(\ca^{*c}_k\,|\,\cg_{k-1})=O\left(\frac{(2\lambda(2-s))^{2k}\log^2 n+(2\lambda(2-s))^{k}\log^3 n}{n}\right).
    \end{equation}

    For each $v\in \excl_{k-1}$, we have
    \[
    |\sfch_v|\sim \bin\left(n-|\bar\cv(\ct_{k-1})\cup\bar\cv(\td\ct_{k-1})|,\frac{\lambda}{n}\right).
    \]
    By Lemma~\ref{lem:binom-poi-tv}, we can couple $|\sfch_v|$ with $X_v$ such that
    \[
    \P(|\sfch_v|=X_v)\ge 1-\frac{\lambda^2+(2\lambda(2-s))^{k-1}\lambda\log n}{n}=1-O\left(\frac{(2\lambda(2-s))^k\log n+\log^2 n}{n}\right),
    \]
    and a union bound over $\excl_{k-1}$,  whose size is bounded by $(2\lambda(2-s))^{k-1}\log n$ conditioned on event $\cap_{i=1}^{k-1}\cb_i$ yields 
    \begin{equation}
        \label{eq:ak}
        \P(\ca^{c}_k\,|\,\cg_{k-1})=O\left(\frac{(2\lambda(2-s))^{2k}\log^2 n+(2\lambda(2-s))^{k}\log^3 n}{n}\right).
    \end{equation}
    By an analogous argument, we can also get
    \begin{equation}
        \label{eq:tdak}
        \P(\td\ca^{c}_k\,|\,\cg_{k-1})=O\left(\frac{(2\lambda(2-s))^{2k}\log^2 n+(2\lambda(2-s))^{k}\log^3 n}{n}\right).
    \end{equation}

    Because we have $|\com_{k-1}\cup\excl_{k-1}\cup\wdtd\excl_{k-1}|\le 2^{k-2}(\lambda(2-s))^{k-1}\log n$ on the event $\cap_{i=1}^{k-1}\cb_i$,
    there are at most $n2^{2k-4}(\lambda(2-s))^{2k-2}\log^2 n$ ways to pick two distinct vertices from $\com_{k-1}\cup\excl_{k-1}\cup\wdtd\excl_{k-1}$ and one vertex from the unexplored part of the graph. In the union graph $\bar G$, the probability that the two picked vertices both have an edge with the vertex in the unexplored part is $\frac{\lambda(2-s)}{n}$. 
    Then, we can apply the union bound to obtain
    \begin{equation}
        \label{eq:dk}
        \P(\cd_k^c\,|\,\cg_{k-1})=O\left(\frac{\lambda^2n2^{2k-4}(\lambda(2-s))^{2k-2}\log^2 n}{n^2}\right)=O\left(\frac{(2\lambda(2-s))^{2k}\log^2 n}{n}\right).
    \end{equation}

    Notice that 
    \[
    |\com_{k}\cup\excl_{k}\cup\wdtd\excl_{k}|\sto\bin\left(n,\frac{\lambda(2-s)}{n}|\com_{k-1}\cup\excl_{k-1}\cup\wdtd\excl_{k-1}|\right).
    \]
    We can apply the Chernoff bound (equation~\eqref{eq:CB_upper2} in Lemma~\ref{lem:CB}) to obtain
    \begin{align}
        \P(\cb^c_k\,|\, \cg_{k-1})&\le \P\left(\bin\left(n,\frac{\lambda(2-s)}{n}2^{k-2}(\lambda(2-s))^{k-1}\log n\right)\ge 2^{k-1}(\lambda(2-s))^k\log n\right)\nonumber\\
        &\le \exp\left(-\frac{2^{k-2}(\lambda(2-s))^k\log n}{3}\right)=n^{-\omega(1)}.
        \label{eq:bk}
    \end{align}

   Now further condition on event $\cb_k$, which implies there are at most $2^{2k-2}(\lambda(2-s))^{2k}\log^2 n$ vertex pairs in $\com_{k}\cup\excl_{k}\cup\wdtd\excl_{k}$. we can apply a union bound over these vertex pairs to obtain
    \begin{equation}
        \label{eq:ck}
        \P(\cc_k^c\,|\,\cg_{k-1}\cap\cb_k)\le 2^{2k-2}(\lambda(2-s))^{2k}\frac{\lambda(2-s)}{n}\log^2 n=O\left(\frac{(2\lambda(2-s))^{2k+1}\log^2 n}{n}\right).
    \end{equation}

    By a union bound over the vertex pairs in $(\cup_{i=1}^{k}(\excl_i\cup\com_i))\times (\cup_{i=1}^{k}\wdtd\excl_i)$,  whose size is bounded by $$\left(\sum_{i=1}^{k} 2^{i-1}(\lambda(2-s))^i\log n\right)^2$$ conditioned on event $\cap_{i=1}^k\cb_i$, we get 
    \begin{align}
        \P(\ci_k^c\,|\,\cg_{k-1}\cap\cb_k)&\le \left(\sum_{i=1}^{k} 2^{i-1}(\lambda(2-s))^i\log n\right)^2\frac{\lambda}{n}\nonumber\\
        &\le (2\lambda(2-s))^{2k}\frac{\lambda}{n}\log^2 n \nonumber\\
        &=O\left(\frac{(2\lambda(2-s))^{2k+1}\log^2 n}{n}\right).\label{eq:fk}
    \end{align}
    Similarly, we can bound 
    \begin{equation}
        \label{eq:tdfk}
        \P(\td\ci_k^c\,|\,\cg_{k-1}\cap\cb_k)=O\left(\frac{(2\lambda(2-s))^{2k+1}\log^2 n}{n}\right).
    \end{equation}
    By equations~\eqref{eq:ak} -~\eqref{eq:tdfk}, we have
    \begin{align}
        \label{eq:depthk}
        &\P(\ca^{*c}_k\cup\ca^c_k\cup\td\ca_k^c\cup\cb_k^c\cup\cc_k^c\cup\cd^c_k\cup\ci^c_k\cup\td\ci^c_k\,|\,\cg_{k-1})\nonumber\\
        ={}&O\left(\frac{(2\lambda(2-s))^{2k+1}\log^2 n+(2\lambda(2-s))^{k}\log^3 n}{n}\right).
    \end{align}

    Finally, by equations~\eqref{eq:depth1} and~\eqref{eq:depthk},
    under the coupling we constructed
    \begin{align*}
        &\P(\ell(\ct_d,\td\ct_d)=(\ch_d,\td\ch_d))\\
        \ge{}& 1-\P(\ca^{*c}_1\cup\cb_1^c\cup\cc_1^c)-\sum_{k=2}^d \P(\ca^{*c}_k\cup\ca^c_k\cup\td\ca_k^c\cup\cb_k^c\cup\cc_k^c\cup\cd^c_k\cup\ci^c_k\cup\td\ci^c_k\,|\,\cg_{k-1})\\
        ={}&1-O\left(\frac{\lambda^3\log^2 n}{n}\right)-\sum_{k=2}^dO\left(\frac{(2\lambda(2-s))^{2k+1}\log^2 n+(2\lambda(2-s))^{k}\log^3 n}{n}\right)\\
        ={}&1-O\left(\frac{d(2\lambda(2-s))^{2d+1}\log^2 n + d(2\lambda(2-s))^{d}\log^3 n}{n}\right).
    \end{align*}
    Recall that $\lambda=(\log n)^{\alpha+o(1)}$ and $d=(\log n)^{\gamma}$ for constants $\alpha,\gamma\in(0,1)$. We then have 
    \[
    d(2\lambda(2-s))^{2d+1}\le \exp\left(\log d +(2(\log n)^{\gamma}+1)(\log 4+(\alpha+o(1))\log\log n)\right)=n^{o(1)}.
    \]
    It then follows that 
    \[
    \P(\ell(\ct_k,\td\ct_k)=(\ch_k,\td\ch_k))=1-O(n^{-1/2}),
    \]
    which completes the proof.
\end{proof}

\section{The Poisson approximation and concentration inequalities}\label{appd:concentration}
\begin{lem}[Lemma 22 in~\cite{wang26d}]
    \label{lem:binom-poi-gap}
Let $n\in \mathbb{N}^+$, $\mu\in\mathbb{R}^+$ and $x,y\in\mathbb{N}$ be such that $x<n$ and $y\le n-x$.
    Then we have
    \[
    \frac{\P(\bin(n-x,\mu/n)=y)}{\P(\poi(\mu)=y)}\le \exp\left(\frac{\mu}{n}(x+y)\right).
    \]
\end{lem}

\begin{restatable}{lem}{lembinompoi}
    \label{lem:binom-poi-tv}
        If $m$ and $n$ are positive integers such that $m\le n$ and $\mu$ is a positive real number, then
        \[
        \dtv\left(\bin\left(m,\frac{\mu}{n}\right),\poi(\mu)\right)\le \frac{\mu^2+(n-m)\mu}{n}.
        \]
    \end{restatable}
\begin{proof}
    By the triangle inequality, we have
    \begin{align*}
        &\dtv\left(\bin\left(m,\frac{\mu}{n}\right),\poi\left(\mu\right)\right)\\
        \le{}&  \dtv\left(\bin\left(m,\frac{\mu}{n}\right),\poi\left(\frac{m\mu}{n}\right)\right)+\dtv\left(\poi(\mu),\poi\left(\frac{m\mu}{n}\right)\right),
    \end{align*}
    and it suffices to bound these two total variation distances separately. 
    
    It follows by the classic Poisson approximation of binomial random variables (see for example~\cite{hodges1960}) that
    \begin{equation}
        \label{eq:poi-approx}
        \dtv\left(\bin\left(m,\frac{\mu}{n}\right),\poi\left(\frac{m\mu}{n}\right)\right)\le \frac{m\mu^2}{n^2}\le \frac{\mu^2}{n}.
    \end{equation}

    On the other hand, for a random variable $X\sim \poi(\mu)$, we can write $X=Y+Z$, where $Y\sim \poi(m\mu/n)$ and $Z\sim\poi(\mu-m\mu/n)$ are two independent Poisson random variables. We then have
    \begin{equation}
        \dtv\left(\poi(\mu),\poi\left(\frac{m\mu}{n}\right)\right)\le \P(X\neq Y)=\P(Z\neq 0)= 1-\exp\left(-\mu+\frac{m\mu}{n}\right)\le \frac{(n-m)\mu}{n}.\label{eq:tv_poi}
    \end{equation}
    Equations~\eqref{eq:poi-approx} and~\eqref{eq:tv_poi} together imply the desired inequality.
\end{proof}

\begin{restatable}{lem}{lemmultipoi}
        \label{lem:multi-poi}
        Let $m$ and $n$ be two positive integers with $n/2\le m\le n$ and $s\in (0,1]$ be a constant independent of $n$.
        Suppose $\lambda$ is a positive real number that satisfies $\lambda=(\log n)^{\alpha+o(1)}$ for some constant $\alpha\in (0,1)$. 
        Consider random variables $$(X_1,X_2,X_3,X_4)\sim\mathrm{Multi}(m,\lambda s/n,\lambda(1-s)/n,\lambda(1-s)/n,1-(2-s)\lambda/n),$$ and three independent Poisson random variables $X_1'\sim \poi(\lambda s)$, $X_2'\sim \poi(\lambda(1-s))$ and $X_3'\sim\poi(\lambda(1-s))$. Then we have 
        \begin{equation}
        \label{eq:multi_dtv}
        \dtv(\cl(X_1,X_2,X_3),\cl(X_1',X_2',X_3'))=O\left(\frac{(n-m)\lambda+\log^2 n}{n}\right).
        \end{equation}
        
    \end{restatable}
    \begin{proof}
        First assume $s=1$. In this case,~\eqref{eq:multi_dtv} simply follows by Lemma~\ref{lem:binom-poi-tv}. We assume $s<1$ in the rest of the proof.

        By the definition of the multinomial distribution, random variables $X_1,X_2,X_3$ can be sampled in a sequential manner:
        \begin{enumerate}
            \item Sample $X_1\sim\bin(m,\lambda s/n)$.
            \item Given $X_1$, sample $X_2\sim\bin\left(m-X_1,\frac{\lambda(1-s)}{n-\lambda s}\right)$.
            \item Given $X_1$ and $X_2$, sample $X_3\sim\bin\left(m-X_1-X_2,\frac{\lambda(1-s)}{n-\lambda}\right)$.
        \end{enumerate}
    By Lemma~\ref{lem:binom-poi-tv}, we can couple $X_1$ and $X_1'$ such that under the coupling, 
        \begin{equation}
        \label{eq:X1}
        \P(X_1\neq X_1')\le \frac{\lambda^2 s^2+(n-m)\lambda s}{n}=O\left(\frac{(n-m)\lambda+\log^2 n}{n}\right).
        \end{equation}

    By the Chernoff bound (equation~\eqref{eq:CB_Upper} in Lemma~\ref{lem:CB}), we know that 
        \begin{equation}
        \label{eq:X1_bound}
        \P(X_1> \log n)=n^{-\omega(1)}.
        \end{equation}

    Conditioned on the event $\{X_1=X_1'=x\}$, 
    we have 
    \[
    X_2\sim \bin\left(m-x,\frac{\lambda(1-s)}{n-\lambda s}\right),
    \]
    and 
    \[
    X_2'\sim\poi(\lambda (1-s)).
    \]
    We can couple these two random variables such that
    \begin{align}
        \P(X_2\neq X_2'\cond X_1=X_1'=x)
        &\le \dtv\left(\bin\left(m-x,\frac{\lambda(1-s)}{n-\lambda s}\right),\poi(\lambda (1-s))\right)\nonumber\\
        &\le \dtv\left(\bin\left(m-x,\frac{\lambda(1-s)}{n-\lambda s}\right),\poi\left(\frac{\lambda(1-s)(m-x)}{n-\lambda s}\right)\right)\nonumber\\
        &\;\;+\dtv\left(\poi(\lambda (1-s)),\poi\left(\frac{\lambda(1-s)(m-x)}{n-\lambda s}\right)\right)\nonumber\\
        &\le \frac{(m-x)\lambda^2(1-s)^2}{(n-\lambda s)^2}+\Bigg|\lambda(1-s)-\frac{\lambda(1-s)(m-x)}{n-\lambda s}\Bigg|,\label{eq:littlex}
    \end{align}
    where the last inequality follows by~\eqref{eq:poi-approx} and~\eqref{eq:tv_poi} in the proof of Lemma~\ref{lem:binom-poi-tv}. 
    
    Under the assumption that $x\le \log n$, we further have
    \[
    \frac{(m-x)\lambda^2(1-s)^2}{(n-\lambda s)^2}=O\left(\frac{\log ^2 n}{n}\right)
    \]
    and
    \[
    \Bigg|\lambda(1-s)-\frac{\lambda(1-s)(m-x)}{n-\lambda s}\Bigg|=\Bigg|\frac{\lambda(1-s)(n-m-\lambda s+x)}{n-\lambda s}\Bigg|=O\left(\frac{(n-m)\lambda+\log^2 n}{n}\right).
    \]
    So the right-hand side of~\eqref{eq:littlex} is equal to
    \[
     O\left(\frac{(n-m)\lambda+\log^2 n}{n}\right).
    \]
    This implies that conditioned on events $\{X_1=X_1'\}$ and $\{X_1\le \log n\}$, we can couple $X_2$ and $X_2'$ so that
    \begin{equation}
        \P(X_2\neq X_2'\cond X_1=X_1',X_1\le \log n)=O\left(\frac{(n-m)\lambda+\log^2 n}{n}\right)\label{eq:x2}.
    \end{equation}

    We can apply the Chernoff bound (equation~\eqref{eq:CB_Upper} in Lemma~\ref{lem:CB}) to get
        \begin{equation}
        \label{eq:X2_bound}
        \P(X_2> \log n\cond X_1=X_1',X_1\le \log n)= n^{-\omega(1)}.
        \end{equation}
     Conditioned on events $\{X_1=X_1'=x\}$ and $\{X_2=X_2'=y\}$, we have
     \[
    X_3\sim \bin\left(m-x-y,\frac{\lambda(1-s)}{n-\lambda }\right),
    \]
    and 
    \[
    X_3'\sim\poi(\lambda (1-s)).
    \]
    Similarly to~\eqref{eq:littlex}, we can couple these two random variables so that
    \begin{equation}
        \label{eq:littlexy}
        \P(X_3\neq X_3'\cond X_1=X_1'=x,X_2=X_2'=y)\le\frac{(m-x-y)\lambda^2(1-s)^2}{(n-\lambda )^2}+\Bigg|\lambda(1-s)-\frac{\lambda(1-s)(m-x-y)}{n-\lambda }\Bigg|.
    \end{equation}
    Under the assumptions $x\le \log n$ and $y\le \log n$, the right-hand side of~\eqref{eq:littlexy} is
    \[
    O\left(\frac{(n-m)\lambda+\log^2 n}{n}\right).
    \]
    This implies that conditioned on events $\{X_1=X_1'\}$, $\{X_1\le \log n\}$, $\{X_2=X_2'\}$ and $\{X_2\le \log n\}$, we can further couple $X_3$ with $X_3'$ such that
    \begin{equation}
        \label{eq:x3}
        \P(X_3\neq X_3'|X_1=X_1',X_2=X_2',X_1\le \log n,X_2\le \log n)=O\left(\frac{(n-m)\lambda+\log^2 n}{n}\right).
    \end{equation}

    Finally, putting~\eqref{eq:X1},~\eqref{eq:X1_bound},~\eqref{eq:x2},~\eqref{eq:X2_bound} and~\eqref{eq:x3} yields that we can couple $(X_1,X_2,X_3)$ and $(X_1',X_2',X_3')$ such that
    \begin{align*}
    &\P(X_1=X_1',X_2=X_2',X_3=X_3')\\
    \ge{}& \P(X_1=X_1',X_2=X_2',X_3=X_3',X_1\le\log n,X_2\le\log n)\\
    \ge{}& 1-\P(X_1\neq X_1')-\P(X_1> \log n)-\P(X_2\neq X_2'\cond X_1=X_1',X_1\le \log n)\\
    &\;\;\;\;-\P(X_2> \log n\cond X_1=X_1',X_1\le \log n)-\P(X_3\neq X_3'|X_1=X_1',X_2=X_2',X_1\le \log n,X_2\le \log n)\\
    ={}&1-O\left(\frac{(n-m)\lambda+\log^2 n}{n}\right),
    \end{align*}
    which completes the proof.
    
    \end{proof}
    
\begin{lem}[Bennett's inequality~\cite{bennett1962}]
\label{lem:bennett}
    
    Suppose $X_1,\ldots,X_n$ are independent random variables with finite second moments. Assume $X_i\le B$ a.s. for each $i\in [n]$. Let $V=\sum_{i=1}^n\E[(X_i-\E[X_i])^2]$. Then, for every $x\ge 0$, 
    \[
    \P\left(\sum_{i=1}^n(X_i-\E[X_i])\ge x\right)\le \exp\left(-\frac{V}{B^2}\cdot \phi(xB/V)\right),
    \]
    where $\phi(t):=(1+t)\log (1+t)-t$ for $t\ge 0$.
\end{lem}

\begin{lem}[Chernoff bound for Binomial distribution~\texorpdfstring{\cite[Theorems~20~and~21]{mitzenmacher2017probability}}{}]\label{lem:CB}
    Let $X\sim \mathrm{Bin}(n,p)$ be a Binomial distribution with mean $\mu:=np$. Then the following bounds hold.
    \begin{enumerate}[label=\textbf{(\roman*)}]
        \item \textbf{(Upper Tails).} For any $\tau> \mu$, 
        \begin{align}
            \P(X\ge \tau)< \Bigl(\frac{e\mu}{\tau}\Bigr)^{ \tau}=\exp \Bigl(\tau\bigl(1+\log\mu-\log\tau\bigr)\Bigr).\label{eq:CB_Upper}
        \end{align}
        For any $\delta\in(0,1]$,
        \begin{equation}
            \P(X\ge (1+\delta)\mu)<\exp\Bigl(-\frac{\mu\delta^2}{3}\Bigr).
            \label{eq:CB_upper2}
        \end{equation}
        \item \textbf{(Lower Tail).} For any $\delta\in(0,1)$, 
        \begin{align}
            \P(X\le (1-\delta)\mu)\le \exp\Bigl(-\frac{\mu\delta^2}{2}\Bigr).\label{eq:CB_Lower}
        \end{align}
    \end{enumerate}
\end{lem}

\section{An explicit choice of the tie-breaking map $f$}
\label{appd:treecode} 
The rank-based test \texttt{RBAlign} breaks ties in the likelihood-ratio ranking through an arbitrary injective map $f$ that assigns to each local tree pair in Definition~\ref{de:local_tree} a unique real number in $[0,1]$. In this section, we construct an explicit example of such a map $f$ that can be calculated in time $n^{o(1)}$. Throughout this section, $(t,\tilde t)$ denotes a local tree pair in the sense of Definition~\ref{de:local_tree}.

To define $f$, we first introduce a code for a local tree $t\in\cz_d$ as follows.
\begin{itemize}
  \item {\bf Step 1:} Start from the leaves of $t$, and encode each of them as $0$.
  \item {\bf Step 2:} For a node $v$ such that all of its $c_v$ children  have already been assigned codes, order these child codes in non-increasing lexicographic order as $w_1\ge w_2\ge\cdots\ge w_{c_v}$. We then assign to $v$ the label
  $c_v\circ w_1\circ\cdots\circ w_{c_v}$, where $\circ$ here denotes the concatenation of finite words.
  \item {\bf Step 3:} Iterate Step~2 until the root of $t$ has a label; then $t$ is
  encoded as the label of the root.
\end{itemize}

Recall that every local tree has degree at most $\tau$. Since every vertex contributes exactly one digit, the above procedure gives every rooted unlabeled tree $t\in\cz_d$ a code $s_1\ \ldots\ s_{|\cv(t)|}$ of length
$|\cv(t)|$ with digits in $\{0,1,\ldots,\tau\}$. The code $\td s_1 \ \ldots\ \td s_{|\cv(\td t)|}$ of $\td t$
is constructed in the same way. We then define
\begin{equation}
  f(t,\td t)=\sum_{i=1}^{|\cv(t)|}\frac{s_i}{(\tau+2)^{i}}
  +\frac{\tau+1}{(\tau+2)^{|\cv(t)|+1}}
  +\sum_{j=1}^{|\cv(\td t)|}\frac{\td s_{j}}{(\tau+2)^{|\cv(t)|+j+1}}\,\in[0,1).
  \label{eq:treecode}
\end{equation}
Equivalently, $f(t,\td t)$ is the base-$(\tau+2)$ fraction obtained by concatenating the code of $t$, the separator $\tau+1$ and the code of $\td t$. An example of the map $f$ is defined in Figure~\ref{fig:treecode}.

\begin{figure}[t]
\centering
\begin{tikzpicture}[every node/.style={circle,draw,inner sep=1pt,minimum size=6mm,font=\small}]
  \node (v1) at (0,2.1) {$v_1$};
  \node (v3) at (-1.1,1.05) {$v_3$};
  \node (v2) at (1.1,1.05) {$v_2$};
  \node (v4) at (-1.85,0) {$v_4$};
  \node (v5) at (-0.35,0) {$v_5$};
  \draw (v1)--(v3); \draw (v1)--(v2); \draw (v3)--(v4); \draw (v3)--(v5);
  \node[draw=none,font=\scriptsize] at (0,2.6) {$2\,2\,0\,0\,0$};
  \node[draw=none,font=\scriptsize] at (-2.05,1.05) {$2\,0\,0$};
  \node[draw=none,font=\scriptsize] at (1.75,1.05) {$0$};
  \node[draw=none,font=\scriptsize] at (-1.85,-0.5) {$0$};
  \node[draw=none,font=\scriptsize] at (-0.35,-0.5) {$0$};
  \node[draw=none,font=\normalsize] at (1.7,2.5) {$t$};
  \node (w1) at (4.6,2.1) {$\td v_1$};
  \node (w2) at (4.6,1.05) {$\td v_2$};
  \node (w3) at (4.6,0) {$\td v_3$};
  \draw (w1)--(w2); \draw (w2)--(w3);
  \node[draw=none,font=\scriptsize] at (5.5,2.1) {$1\,1\,0$};
  \node[draw=none,font=\scriptsize] at (5.4,1.05) {$1\,0$};
  \node[draw=none,font=\scriptsize] at (5.1,0) {$0$};
  \node[draw=none,font=\normalsize] at (3.9,2.5) {$\td t$};
\end{tikzpicture}
\[
  f(t,\td t)=\Bigl(0.\,\underbrace{2\,2\,0\,0\,0}_{t}\,
  \underbrace{(\tau{+}1)}_{\text{separator}}\,
  \underbrace{1\,1\,0}_{\td t}\Bigr)_{\tau+2}
\]
\caption{An example of calculating $f(t,\td t)$. Each vertex in $t$ is attached with a code. The three leaves get $0$ in Step~1. The child carrying two leaves gets $2\,0\,0$ in Step~2. At the root the two children are labelled $2\,0\,0$ and $0$, so in decreasing order the root's label is $2\,2\,0\,0\,0$, by Step~2. Then $t$ is labeled as $2\,2\,0\,0\,0$, by Step~3. By the same construction, the path $\td t$ gets $1\,1\,0$. Then $f(t,\td t)$ joins the two
codes by the digit $\tau+1$ and reads them in base $\tau+2$.}
\label{fig:treecode}
\end{figure}

\begin{lem}\label{lem:treecode}
The map $f:\cz_d\times\cz_d\to[0,1]$ in \eqref{eq:treecode} is injective and takes
values in $[0,1)$. Moreover, for the local trees defined in Definition~\ref{de:local_tree}, $f(t,\td t)$ can be computed in $n^{o(1)}$ time.
\end{lem}

\begin{proof}
We separately prove the injectivity of $f$ and bound the running time of computing it.

\underline{\textbf{Injectivity.}} Suppose the codes of $t$ and $\td t$ are $s_1s_2\cdots s_{|\cv(t)|}$ and $\td s_1 \td s_2\cdots \td s_{|\cv(\td t)|}$, respectively. We now show how to uniquely decode $(t,\td t)$ from $f(t,\td t)$.

First, expand $f(t,\tilde t)\in[0,1)$ in base $\tau+2$. Since the digit $\tau+1$ appears only as the separator, it uniquely determines the boundary between the code of $t$ and the code of $\tilde t$. Therefore, it suffices to reconstruct a rooted tree from its code.

We now reconstruct $t$ from the code  $s_1s_2\cdots s_{|\cv(t)|}$, starting from the first digit $s_1$. If $s_1=0$, then $t$ consists of a single vertex. Suppose now that $s_1>0$. Then the remaining code $s_2\cdots s_{|\cv(t)|}$ is the concatenation of the codes of the root's $s_1$ children, written in decreasing order. To identify the first child's code, scan the word from left to right until the first $s_k$ where $\sum_{i=2}^k s_i=k-2$. 

We now show that the first child's code is exactly  $s_2\cdots s_{k}$. Let $s_2\cdots s_{m}$ be the code of the first child for some $m\geq 2$, and suppose $t_1$ be the subtree rooted at the first child. We have the following two observations.
\begin{itemize}
    \item By our construction, the number of vertices of $t_1$ is $m-1$, and the total number of children among all the vertices of $t_1$ is $\sum_{i=2}^{m}s_i$, which equals to the number of edges of $t_1$. Since $t_1$ is a tree, we have $\sum_{i=2}^m s_i=m-2$.
    \item Fix any $r\in\mathbb{N}$ such that $2<r<m$. Then $s_2\cdots s_r$ records the numbers of children of exactly $r-1$ nodes of $t_1$. Among these $r-1$ nodes, one is the root of $t_1$, while each of the remaining $r-2$ nodes has a parent, which by our construction is still one of those $r-1$ nodes. Therefore, among the $\sum_{i=2}^r s_i$ children of these $r-1$ nodes, there are still $\sum_{i=2}^r s_i-(r-2)$ children that do not belong to these $r-1$ nodes. Such children must exist, as $s_2\cdots s_r$ does not contain all vertices of $t_1$. Therefore, we have $\sum_{i=2}^r s_i>(r-2)$.
\end{itemize}
  This shows that $m$ is the first index such that $\sum_{i=2}^m s_i=m-2$, so $k=m$, and the first child's code is exactly  $s_2\cdots s_{k}$.

Next, by removing the prefix $s_2\cdots s_k$ and repeating the above procedure to the remaining sequence, one can recover the codes assigned to all children of the root. Applying the same reconstruction recursively to each child code uniquely determines the tree $t$.

To recover $\td t$, we notice that the base-$(\tau+2)$ expansion of $f(t,\tilde t)$ may omit trailing zeros from the code of $\tilde t$. Therefore, we first append zeros to the code of $\td t$ until its length is equal to the sum of its digits plus one. The resulting word is exactly the code of $\tilde t$, and hence $\tilde t$ is reconstructed by the same procedure. Therefore, $f$ is injective, as both $t$ and $\tilde t$ are uniquely determined by $f(t,\tilde t)$.

\underline{\textbf{Running time.}} Recall that in our construction, the code of $t$ is computed in a bottom-up pass from the leaves, and each node is processed once. For each node $v$, we first sort the codes of its at most $\tau$ children and then concatenate them with one additional digit recording the number of children of $v$. Therefore, constructing the code of $t$ takes $O(|\cv(t)|\tau\log\tau)$ time, and the same bound holds for $\td t$. Once the two codes are constructed, the value of $f(t,\tilde t)$ is calculated directly from \eqref{eq:treecode}.

Finally, every local tree has depth at most $d$ and maximum degree at most $\tau$, so $|\cv(t)|+|\cv(\td t)|\le 2\tau^{d+1}$. Since $\gamma<1$ and $\log\tau=O(\log\log n)$,
\[
  O((|\cv(t)|+|\cv(\td t)|)\tau\log\tau)=O(\tau^{d+2}\log\tau)
  =\exp\bigl(O((\log n)^{\gamma}\log\log n)\bigr)=n^{o(1)}.
\]
Therefore, $f(t,\td t)$ can be computed in $n^{o(1)}$ time.
\end{proof}

\end{document}